%% file: paper-exp.tex
\documentclass[11pt,letterpaper]{article}
\usepackage[left=1in,top=1in,right=1in,bottom=1in]{geometry}
\usepackage{amsmath} 
\usepackage{enumitem} 

\usepackage{tikz}  
\usetikzlibrary{arrows}
\usetikzlibrary{patterns,snakes}
\usetikzlibrary{decorations.shapes}
\tikzstyle{overbrace text style}=[font=\tiny, above, pos=.5, yshift=5pt]
\tikzstyle{overbrace style}=[decorate,decoration={brace,raise=5pt,amplitude=3pt}]
\usetikzlibrary{shapes.geometric}

\newsavebox{\fmbox}
\newenvironment{fmpage}[1]
{\begin{lrbox}{\fmbox}\begin{minipage}{#1}}
		{\end{minipage}\end{lrbox}\fbox{\usebox{\fmbox}}}

\usepackage{thmtools}
\usepackage{thm-restate}
\usepackage{verbatim}
\usepackage{amsmath}
\usepackage{amssymb}
\usepackage{placeins} 

\usepackage{algorithm,algorithmicx}
\usepackage{algpseudocode}
 \usepackage{mathtools}

\usepackage{comment}

\newcount\Comments
\definecolor{darkgreen}{rgb}{0,0.7,0}
\newcommand{\kibitz}[2]{\ifnum\Comments=1\textcolor{#1}{#2}\fi}

\newtheorem{theorem}{Theorem}[section]
\newtheorem{definition}{Definition}

\newtheorem{corollary}[theorem]{Corollary}
\newtheorem{observation}[theorem]{Observation}
\newtheorem{lemma}[theorem]{Lemma}

\newcommand{\rset}{\mbox{{\normalfont I\hspace{-.4ex}R}}}
\newcommand{\np}{\mbox{{\normalfont NP}}}
\newcommand{\conp}{\mbox{{\normalfont co-NP}}}
\newcommand{\ppad}{\mbox{{\normalfont PPAD}}}

\newcommand{\ppa}{\mbox{{\normalfont PPA}}}
\newcommand{\ppp}{\mbox{{\normalfont PPP}}}
\newcommand{\pls}{\mbox{{\normalfont PLS}}}
\newcommand{\tfnp}{\mbox{{\normalfont TFNP}}}

\newcommand{\true}{\mbox{{\sc true}}}
\newcommand{\false}{\mbox{{\sc false}}}

\newcommand{\hide}[1]{}
\newcommand{\lplus}{\ensuremath{A_+}}
\newcommand{\lminus}{\ensuremath{A_-}}
\newcommand{\sensor}{\ensuremath{{\cal S}}}

\def\tucker{{\sc Tucker}}
\def\twodtucker{2D-{\sc Tucker}}
\def\twodmstucker{2D-MS-{\sc Tucker}}
\def\ch{{\sc Consensus-halving}}

\def\vt{{\sc Variant Tucker}}

\def\ns{{\sc Necklace-splitting}}

\def\es{{\sc Equal Subsets}}
\def\fac{{\sc Factoring}}

\def\smith{{\sc Smith}}

\def\egc{{\sc $\epsilon$-Gcircuit}}
\newcommand{\necklace}{{\sc NeckSplit}}
\newcommand{\CKD}{{\sc Con-$1/k$-Division}}

\newcommand{\y}{100}

\newlength{\boxwidth}
\makeatletter
\DeclareRobustCommand{\qed}{%
  \ifmmode 
  \else \leavevmode\unskip\penalty9999 \hbox{}\nobreak\hfill
  \fi
  \quad\hbox{\qedsymbol}}
\newcommand{\openbox}{\leavevmode
  \hbox to.77778em{%
  \hfil\vrule \vbox to.675em{\hrule width.6em\vfil\hrule}%
  \vrule\hfil}}
\newcommand{\qedsymbol}{\openbox}

\newenvironment{proof}[1][\proofname]{\par \normalfont
  \topsep6\p@\@plus6\p@ \trivlist
  \item[\hskip\labelsep\bfseries\itshape #1.]\ignorespaces }{%
  \qed\endtrivlist }

\newcommand{\proofname}{Proof}

\def\fac{{\sc Factoring}}
\def\phc{{\sc Pigeonhole Circuit}}

\author{Aris Filos-Ratsikas\footnote{Oxford University, UK {\tt Aris.Filos-Ratsikas@cs.ox.ac.uk}}
\and
Paul W. Goldberg\footnote{Oxford University, UK {\tt Paul.Goldberg@cs.ox.ac.uk}
}}

\date{\today}

\begin{document}

\title{Consensus Halving is PPA-Complete}

\maketitle

\begin{abstract}
\noindent
We show that the computational problem CONSENSUS-HALVING is PPA-complete,
the first PPA-completeness result for a problem whose definition
does not involve an explicit circuit.
We also show that an approximate version of this problem is
polynomial-time equivalent to NECKLACE SPLITTING, which
establishes PPAD-hardness for NECKLACE SPLITTING, and suggests
that it is also PPA-complete.
\end{abstract}

\begin{paragraph}{Keywords:}
Computational complexity; TFNP; Tucker's lemma
\end{paragraph}

\section{Introduction}\label{sec:intro}

The class \tfnp\ \cite{MP91} of {\em total} search problems in \np\ (where every instance
has an easily-checkable solution) does not seem to have complete problems.
Moreover, no problem in \tfnp\ can be \np-complete unless \np=\conp.
Consequently, alternative notions of computational hardness need to be
developed and applied
in our effort to understand the many and varied problems in \tfnp\
that seem to be intractable.

The complexity class \pls\ (Johnson et al.~\cite{JPY}),
and the classes \ppad, \ppa, and \ppp\ (Papadimitriou~\cite{Pap}) are
subclasses of \tfnp\ associated with various
combinatorial principles that guarantee totality.
Each principle has a corresponding definition of
a computational problem whose totality applies that principle in the
most general way possible, and a complexity class of problems
reducible to it. In more detail:
\begin{itemize}
\item \pls\ consists of problems whose totality invokes the principle
that every directed acyclic graph has a sink vertex;
\item \ppad\ consists of problems whose totality is based on the principle
that given a source in a directed graph whose vertices have in-degree
and out-degree at most 1, there exists another degree-1 vertex;
\item \ppa\ differs from \ppad\ in that the graph need not be directed;
being a more general principle, \ppa\ is thus a superset of \ppad;
\item \ppp, based on the pigeonhole principle, consists of problems
reducible to \phc.
\end{itemize}
Of these complexity classes, so far only \pls\ and \ppad\ has succeeded in capturing
the complexity of ``natural'' computational problems, and the main
point of the present paper is to show for the first time that this is also
true for \ppa. 

The \ch\ problem involves a set of $n$ agents each of whom has
a valuation function on a 1-dimensional line segment $A$ (the ``cake'', in
cake-cutting parlance). Consider the problem of selecting $k$
``cut points'' in $A$ that partition $A$ into $k+1$ pieces,
then labelling each piece either ``positive'' or ``negative'' in such
a way that each agent values the positive pieces equally to the negative ones.
In 2003, Simmons and Su \cite{SS03} showed that this can always be
done for $k=n$; their proof applies the Borsuk-Ulam theorem and
is a proof of existence analogous to Nash's famous
existence proof of equilibrium points of games, proved using Brouwer's or
Kakutani's fixed point theorem.
Significantly, Borsuk-Ulam is the {\em undirected} version of Brouwer, and
already from \cite{Pap} we know that it relates to \ppa, making
\ch\ a candidate for \ppa-completeness.
As detailed in Definition~\ref{def:ch}, we assume that valuations are
presented as step functions using the logarithmic cost model of
numbers.

\subsection{Related work}

The complexity class \ppad\ has been successful in capturing the
complexity of many versions of Nash equilibrium
\cite{DGP,CDT,EGG06,M14,R15,CDO15} and market equilibrium
computation \cite{CSVY08,CDDT09,VY11,CPY13,SSB},
also cake-cutting~\cite{DQS09}.
Kintali et al.~\cite{KPRST13} extend \ppad-completeness to further domains
including network routing, coalitional games, combinatorics, and social networks.
Rubinstein~\cite{R16} introduced an exponential-time hypothesis for \ppad\ to
rule out a PTAS for approximate Nash equilibrium computation on bimatrix games.
The class \pls\ represents
the complexity of an even larger number of local optimisation problems.
These results speak to the importance of \ppad\ and \pls\ as complexity classes.
By contrast, hitherto the only problems known to be \ppa-complete
are ones that involve circuits
(or equivalently, polynomial-time Turing machines)
in their definition, which represented a critique of \ppa.
Noting that consensus-halving is a kind of social-choice problem, our result
can be seen as an example of computational social choice helping to
populate ``lonely'' complexity classes, a phenomenon recently reviewed by Hemaspaandra~\cite{H17}.
The complexity class \ppp\ still suffers from
that problem, although the present paper should raise our hope that problems
such as \es\ will turn out to be complete for \ppp.
Oracle separations of all these classes are known from \cite{BCEIP}.

The distinction between \ppad\ and \ppa\ revolves around whether we
are searching for a fixpoint in an oriented topological space,
or an unoriented one.
For example, while Papadimitriou~\cite{Pap} showed that it's
\ppad-complete to find a Sperner solution in a 3D cube,
Grigni~\cite{Grigni} showed that it's \ppa-complete to find a solution
to Sperner's lemma in a 3-manifold consisting of the product of a M\"{o}bius
strip and a line segment.
The 2-dimensional versions of these results are given in \cite{CD09,DEFKQX}.
Despite the apparent similarity between the definitions of \ppad\ and \ppa,
there is more progress in basing the hardness of \ppa\ on standard
cryptographic assumptions: \fac\ can be reduced to \ppa\ (with a
randomised reduction) \cite{J12}, while so far, the hardness of \ppad\ has
relied on problems from indistinguishability obfuscation~\cite{BPR15,GPS15};
Garg et al.~\cite{GPS16} make progress in weakening the cryptographic
assumptions on which to base the hardness of \ppad, but these are
still less satisfying than in the case of \ppa.

Examples of problems known to be \ppa-complete include the following.
Aisenberg et al.~\cite{ABB} introduce the problem 2D-\tucker:
suppose we have a colouring of
an exponentially-fine grid on a square region, the colouring being
concisely represented via a circuit.
Tucker's Lemma (the discrete version of Borsuk-Ulam) guarantees that
if certain boundary conditions are obeyed, then two adjacent squares in
the grid will get opposite colours.
2D-\tucker\ is the search for such a solution, or alternatively a violation of
the boundary conditions.
As it happens, we use 2D-\tucker\ as the starting-point for
our reductions here.
Deng et al.~\cite{DEFKQX} show \ppa-completeness for finding fully-coloured
points of triangulations of various non-oriented surfaces; the colourings
are presented concisely via a circuit.
Recently, Deng et al.~\cite{DFK} showed that {\em Octahedral Tucker}
is \ppa-complete, reducing from 2D-\tucker\ and using a snake-embedding
style technique that packages-up the exponential grid in 2 dimensions,
into a grid of constant size in high dimension.
Belovs et al.~\cite{BIQSY17} show \ppa-completeness for novel
problems presented in terms of arithmetic circuits representing instances
of the Chevalley-Warning Theorem, and Alon's Combinatorial Nullstellensatz.
There remain other problems in \ppa\ that are not defined in terms
of circuits, and are conjectured to be \ppa-complete.
They include \smith, the problem of finding a second
Hamiltonian cycle (given one as part of the input) in a odd-degree graph
\cite{Pap,Th}, and the discrete Ham Sandwich problem \cite{Pap}
(given $n$ sets of $2n$ points in general position in $n$-space, find
a hyperplane that splits each of these sets into two subsets of size $n$).
Also the problem \ns\ \cite{Alon87,AW86}, discussed in~\cite{Pap}, Simmons
and Su \cite{SS03} note the connection with consensus-halving.

A precursor of this paper \cite{FFGZ} established that \ch\ is \ppad-hard,
even when we allow constant-size approximation errors for the agents.
Taken with the computational equivalence of \ch\ and \ns\ established
here, we immediately obtain \ppad-hardness of \ns, thus in a well-established
sense, \ns\ is computationally intractable.
This partially answers a question posed in \cite{ABB} about the
hardness of \ns.

\subsection{Overview of the proof}

We begin by explaining the ground covered by \cite{FFGZ} (where \ppad-hardness
was established), and then give an overview of the proof in the present paper.
In \cite{FFGZ}, each agent $a$ in a \ch\ instance, has a particular cut $c(a)$
associated with $a$.
In an instance $I_{CH}$ of \ch, we refer to the interval $A$ on
which agents have valuation functions, as the {\em domain} of $I_{CH}$.

\cite{FFGZ} established \ppad-hardness by reduction from the \ppad-complete problem \egc\ 
($\epsilon$-approximate Generalised Circuit) in which the challenge is
to find a fixpoint of a circuit in which each node computes (with error at most $\epsilon$)
a real value in the range $[0,1]$, consisting of a function of
at most two other nodes in the circuit; these may be certain simple arithmetic operations,
or boolean operations (regarding 0 and 1 as representing \false\ and
\true\ respectively).
In \cite{FFGZ}'s reduction from \egc\ to \ch, each node $\nu$ of a
generalised circuit has a corresponding agent $a_\nu$, and the value computed at $\nu$ 
is represented by the position taken by the cut $c(a_\nu)$.
$a_\nu$'s valuation function is designed to enforce the relationship
that $\nu$'s value has with the node(s) providing input to $\nu$.
Here we re-use some of the circuit ``gate gadgets'' of \cite{FFGZ}, in particular the boolean ones.
A cut that encodes the value computed at a boolean gate is expected to lie
in one of two short intervals, associated with \true\ and \false.

In moving from \ppad-hardness to \ppa-hardness, we encounter a fundamental
limitation to the above approach, which is that distinct cuts are constrained to
lie in distinct (non-overlapping) regions of $A$, and {\em collectively,
the cuts lie in an oriented domain}. A new idea is needed, and we
construct two special agents (the ``coordinate-encoding
agents'') along with two cuts that correspond to those agents, which
are less constrained regarding where, in principle, they may occur,
in a solution to the resulting \ch\ instance $I_{CH}$.
These two cuts are regarded as representing a point on a
M\"{o}bius strip, and a distance metric between two pairs of positions
for these cuts, does indeed correspond to distance between points on
a M\"{o}bius strip.
New problems arise from this freedom regarding where these cuts can occur,
mainly the possibility that one of them may occur outside of the
intended ``coordinate-encoding region'' of the domain of $I_{CH}$.
Consequently it may interfere with the circuitry that $I_{CH}$ uses
to encode an instance of \twodtucker\ (which, recall, is the problem
we reduce from).
We deal with this possibility by making multiple copies of the circuit,
so that an unreliable copy is ``out-voted'' by the reliable ones.
The duplication (we use 100 copies) of the circuit serves a further purpose
reminiscent of the  the ``averaging manoeuvre'' introduced in \cite{DGP}:
we need to deal with the possibility of values occurring at nodes of the
circuit that fail to correspond to boolean values.
The duplication corresponds to a sampling of a cluster of points on
the M\"{o}bius strip, most of which get converted to boolean values.

One other significant obstacle addressed here, is due to
coordinate-encoding cuts directly representing the location
of a point on the M\"{o}bius strip {\em with exponential precision}.
We construct a novel mechanism that reads off $\Theta(n)$ bits
of precision from the locations of these cuts, which are then fed in to
the circuit-encoding part of the consensus-halving instance.
(It is this part of the proof that requires us to work with a definition
of $\epsilon$-\ch\ that may require $\epsilon$ to be inverse exponential.
\ppa-hardness for inverse polynomial $\epsilon$ would lead to \ppa-completeness
of \ns, but there seems to be no way to achieve this while reducing
from \twodtucker\ in a way that directly encodes the location of
a solution to \twodtucker.)

Our reductions start out from the 2D-\tucker\ result of \cite{ABB}.
In Section~\ref{sec:tucker-reds} we give a straightforward proof
of \ppa-completeness of a restricted version called \twodmstucker,
in which two opposite sides of the domain are each monochromatic.
We reduce from this to
an artificial-looking problem called \vt, which is essentially
a messy-looking version of \twodmstucker: the purpose of introducing
\vt\ is to extract some of the technical clutter from the main
event, which is the reduction from there to \ch\ (Section~\ref{sec:consensus}).

\section{Preliminaries}\label{sec:prelims}

\subsection{The Consensus Halving problem}

Simmons and Su~\cite{SS03} were not concerned with computational issues;
their result is essentially topological and shows that a solution exists
provided that agents' valuations are infinitely divisible.
A computational analogue requires us to identify how functions are represented,
and we assume they are given as step functions, or piecewise constant
functions, as have also been considered in the cake-cutting literature \cite{CLPP13,AP16}.
A problem instance also includes an approximation parameter $\epsilon$,
the allowed difference in value between the two sides of the partition,
applicable to any agent.

\begin{definition}\label{def:ch}
$\epsilon$-\ch:
An instance $I_{CH}$ incorporates, for each of $i\in[n]$, a non-negative measure $\mu_i$
of a finite line interval $A=[0,x]$, where each $\mu_i$ integrates to 1 and $x>0$ is part of the input.
We assume that $\mu_i$ are step functions represented in a standard
way, in terms of the endpoints of intervals where $\mu_i$ is constant,
and the value taken in each such interval.
We use the bit model (logarithmic cost model) of numbers.
$I_{CH}$ also incorporates $\epsilon\geq 0$ also represented using the
bit model.
We regard $\mu_i$ as the value function held by agent $i$ for subintervals of $A$.

A solution consists firstly of a set of $n$ {\em cut points} in $A$ (also given in
the bit model of numbers).
These points partition $A$ into (at most) $n+1$ subintervals, and the second
element of a solution is that
each subinterval is labelled $\lplus$ or $\lminus$.
This labelling is a correct solution provided that for each $i$,
$|\mu_i(\lplus)-\mu_i(\lminus)|\leq\epsilon$, i.e.\ each agent has a value
in the range $[\frac{1}{2}-\frac{\epsilon}{2},\frac{1}{2}+\frac{\epsilon}{2}]$ for the
subintervals labelled $\lplus$ (thus, also values the subintervals
labelled $\lminus$ in that range).
\end{definition}

\noindent 
A version where the domain $A$ is take to be $[0,1]$ is polynomial time equivalent to that of
Definition \ref{def:ch} (by scaling the valuations appropriately). In the instances that we construct,
$x$ is polynomial in $n$ but one can equivalently allow $x$ to be exponential in $n$; the rescaling
changes the size of the encoding of the problem instances by a polynomial factor.

Note that it's not hard to check that an instance of \ch\ is well-formed in
the sense that the valuation functions should integrate to 1.
Also, note that the bit complexity of numbers involved in an approximate
solution need not be excessive, so, together with the proof of \cite{SS03}
we have containment in \ppa.
A couple of relevant remarks are the following:
\begin{itemize}
	\item Definition \ref{def:ch} allows the accuracy parameter $\epsilon$ to be inverse exponential in $n$, which will be essential for our reduction. In fact, \cite{FFGZ} established \ppad-hardness of the problem even for constant $\epsilon$. An interesting open question is whether our \ppa-hardness result can be extended for constant or even inverse polynomial $\epsilon$ (which would lead to \ppa-completeness for \ns; Section~\ref{sec:necklace}).
	\item The fact that the functions $\mu_i$ are step-functions which integrate to $1$ over the whole interval $A$ is desirable, since this makes the hardness result stronger, compared to arbitrary functions. Note that while the step functions $\mu_i$ must have polynomially-many steps, the values they may take can differ by exponential (in $n$) ratios. The ``in PPA'' result on the other hand is established for arbitrary (bounded, non-atomic) functions, which also makes it as strong as possible, given that it is a containment result.
\end{itemize}

%

\begin{paragraph}{Solutions with alternating labels:}
We assume without loss of generality that we seek solutions to \ch\ in which the labels $\lplus$ and $\lminus$ alternate as we consider the subintervals formed by the cuts, from left to right.
If, say, there are two consecutive subintervals labelled $\lplus$ in a solution, we could combine them into a single subinterval, leaving us with a un-needed cut, which could be placed at the right-hand endpoint of $A$.
We can also assume without loss of generality that the labelling sequence starts with $\lplus$ on the leftmost subinterval of $A$ defined by the set of cuts.
\end{paragraph}

\subsection{The $2D$-\tucker\ problem}

We review the total search problem $2D$-\tucker,
as defined and shown \ppa-complete in \cite{ABB}. 
(Definition~\ref{def:2d-tucker} is a variant of it, that we use, that's easily
seen to be equivalent to the version of \cite{ABB}.)
An instance of $2D$-\tucker\ consists of a labelling
$\lambda:[m]\times[m]\rightarrow\{\pm 1,\pm 2\}$
satisfying the boundary conditions: for $1\leq i,j\leq m$, $\lambda(i,1)=-\lambda(m-i-1,m)$
and $\lambda(1,j)=-\lambda(m,m-j+1)$.
A solution to such an instance of 2D-\tucker\ is a pair $(x_1,y_1)$, $(x_2,y_2)$
($x_1,x_2,y_1,y_2\in[m]$)
with $|x_1-x_2|\leq 1$ and $|y_1-y_2|\leq 1$ such that $\lambda(x_1,y_1)=-\lambda(x_2,y_2)$.

In the above definition, $m$ is exponential, and $\lambda$
is presented via a circuit that computes it.
We use Definition~\ref{def:2d-tucker}, a variant of the
above whose \ppa-completeness easily follows; it is a more
convenient version for us to use.

\begin{definition}\label{def:2d-tucker}
An instance $I_T$ of $2D$-\tucker\ (with complexity parameter $n$) is defined as follows.
Consider the square region $[0,2^n]\times[0,2^n]$.
For $1\leq i,j\leq 2^n$, the $(i,j)$-squarelet denotes the unit square whose top right vertex is at $(i,j)$.
$I_T$ consists of a boolean circuit $C$ having $2n$ input bits representing
the coordinates of a squarelet,
and 2 output bits representing values $1,-1,2,-2$.
$C$'s labelling should obey the boundary conditions of \cite{ABB} noted above.
A solution consists of two squarelets that touch at at least one point,
and have opposite labels (i.e. labels that sum to 0).
\end{definition}

The containment of the problem in \ppa\ was known from \cite{Pap}.
Aisenberg et al. \cite{ABB} proved that the problem is also \ppa-hard.

\begin{theorem}\cite{ABB,Pap}\label{thm:tuckerppa}
$2D$-\tucker\ is \ppa-complete.
\end{theorem}

\subsection{Organization of the paper}
In Section~\ref{sec:tucker-reds}, we reduce from \twodtucker\ 
(the version of Definition~\ref{def:2d-tucker})
to a restricted version \twodmstucker\ where two opposite sides of the Tucker 
square are completely labelled $1$ and $-1$ (a monochromatic sides version). 
From there, we reduce to an artificial looking variant, \vt, which however 
will prove to be very useful for our main reduction to \ch. Section~\ref{sec:consensus}
presents the main reduction and in Section \ref{sec:pf}, we establish 
the correctness of the reduction.

Finally, in Section \ref{sec:necklace}, we show a computational equivalence between 
approximate Consensus Halving and the well-known Necklace Splitting problem.

\section{Reducing from $2D$-TUCKER to VARIANT-TUCKER}\label{sec:tucker-reds}

In this section, we reduce from \twodtucker\ to a variant of
the Tucker problem, which will be more appropriate to use for proving
\ppa-hardness of approximate Consensus Halving. The \ppa-hardness of
the \vt\ problem will be established through a sequence of two reductions.

\begin{itemize}
	\item First, we reduce from \twodtucker\ to a version of the problem when two opposite sides of the square are assigned only a single label (with opposite signs), e.g. $1$ and $-1$ (Definition~\ref{def:2d-ms-tucker}).
	We will refer to this problem as the \twodmstucker\ (where MS stands for ``monochromatic sides''). See Definition~\ref{def:2d-ms-tucker}, Figure \ref{fig:mstucker}.
	\item Then, we reduce from \twodmstucker\ to \vt, by embedding the regions of the \twodmstucker\ instance (the squarelets) into a triangle-domain and extending the labelling function to points outside these regions. In this process, there is a designated significant sub-domain which contains the embedded regions along with diagonal strips that emerge from the embedded regions and go out the edge of the triangle-domain. The embedding is such that the lines separating the regions are piecewise rectilinear, with sufficiently long pieces. Intuitively, the regions will not be separated by diagonal lines but rather by ``zig-zag'' rectilinear lines that approximate the diagonal ones, which results in set of regions that we refer to as \emph{tiles}. See Figure~\ref{fig:vt}.
\end{itemize}

\begin{definition}\label{def:2d-ms-tucker}
An instance $I_{MS}$ of \twodmstucker\ (with complexity parameter $n$) is defined as follows.
Consider the square region $[0,2^n]\times[0,2^n]$.
For $1\leq i,j\leq 2^n$, the $(i,j)$-squarelet denotes the unit square whose top right vertex is at $(i,j)$.
$I_{MS}$ consists of a boolean circuit $C$ having $2n$ input bits representing
the coordinates of a squarelet,
and 2 output bits representing values $1,-1,2,-2$.
$C$'s labelling should obey the boundary conditions of \cite{ABB} noted above,
but in addition, all squarelets $(x,y)$ with $y=1$ get labelled 1,
and all squarelets $(x,y)$ with $y=2^n$ get labelled -1.
(So, two opposite sides are monochromatic.)
As before, a solution consists of two squarelets that touch at at least one point.
\end{definition}

We start from the \ppa-hardness of \twodmstucker.

\begin{lemma}\label{lem:tuckertomstucker}
	\twodtucker\ is polynomial-time reducible to \twodmstucker.
\end{lemma}

\begin{proof}
	Let $I_T$ be an instance of \twodtucker\ of size $m \times m$. We will construct an instance $I_{MS}$ of \twodmstucker\ of size $3m \times 3m$ such that a solution to $I_{MS}$ will let us efficiently recover a solution to $I_T$. We will need to establish three facts, namely that (i) $I_{MS}$ will be defined on square that satisfies the labelling conditions of Definition~\ref{def:2d-ms-tucker}, (ii) that we don't introduce any solutions during the construction, i.e. that for any solutions to $I_{MS}$ there is a corresponding solution to $I_T$ and (iii) that given such a solution to $I_{MS}$, we can find a solution to $I_T$ in polynomial time.
	
	First, we augment instance $I_T$ with squares of size $m \times m$ attached to the top side (denoted $\mathcal{T}$) and the bottom side (denoted $\mathcal{B}$) of the square\footnote{Equivalently, we can attach squares of size $m \times m$ to the left and right sides of the square.}, which results in a $3m \times m$ rectangle, where only the squarelets with coordinates $(i,j)$ with $i=m+1,\ldots,2m$ and $j = 1, \ldots, m$ are labelled (i.e. the squarelets of the original square, before the augmentation). For the labelling of the remaining squarelets, we will explain how to label squarelets $(i,j)$ with $i=1,\ldots,m$, $j=1,\ldots,m$, i.e. the top-attached square $top$; the colouring of the squarelets of the bottom-attached square $bottom$ is symmetric by the fact that the labels of the squarelets of $\mathcal{T}$ and $\mathcal{B}$ satisfy the antipodal labelling of Tucker's lemma.
	
	To describe the labelling of the square, let $\mathcal{L}$ be the set of squarelets of $top$ that have been assigned a label; obviously at the beginning of the labelling $\mathcal{L} = \emptyset$. The detailed labelling procedure is described in Algorithm \ref{alg:mstucker}.
	\begin{algorithm}
		\begin{algorithmic}
			\State Set $r=m$.
			\While{$r \neq 0$}
			\For{$i=1,\ldots,m$}
			\For{$j=1,\ldots,m-i+1$},
			\State Let $\ell(i,j)=(m+1,r)$
			\EndFor
			\EndFor
			\State $r=r-1$.	
			\EndWhile
			\For {all $(i,j) \notin \mathcal{L}$}
			\State Set $\ell(i,j)=\ell(m+1,j)$
			\EndFor	
		\end{algorithmic}
	\caption{The labelling procedure of the \twodmstucker\ square.}
		\label{alg:mstucker}
	\end{algorithm}
	Next, we transform the $(3m \times m)$-sized rectangle to a square, to ensure that $I_{MS}$ will be a valid instance to \twodtucker. This can be done by attaching two rectangles of size $3m \times m$, one on each side of the rectangle to create a $(3m \times 3m)$-sized square $C$. For the labelling of the rectangles, we do the following: For each row $i$, we label all squarelets $(i,j)$ such that $j < m$ (using the new coordinate system, where the top-most, right-most squarelet of $C$ has coordinates $(1,1)$) with $\ell(i,m+1)$ and all squarelets $(i,j)$ such that $j > 2m$ with $\ell(i,2m)$. Note that all other squarelets have already been coloured in the previous step. Additionally, note that the squarelets in row $1$ of the resulting square have a single label, opposite to that of the squarelets of row $3m$.

	It is not hard to see that the labelling function ensures that ``neighbouring'' squarelets are either assigned the same label or are assigned labels corresponding only to neighbouring labels of the top side $\mathcal{T}$ or the bottom side $\mathcal{B}$ of the original square. Therefore, if there is a complementary edge that appears in any of the added squares, we can follow the direction of the endpoints of the edge, first towards the centre of the corresponding squares $top$ or $bottom$ (i.e. left or right) and then towards the sides of ($\mathcal{T}$ or $\mathcal{B}$ with directions down or up respectively). This obviously can be done in time polynomial in $m$.
\end{proof}

\input{mstucker}

\subsection{Reduction to \vt}

\vt\ is essentially a technically-cluttered version
of \twodmstucker. 
It's helpful as an intermediate stage towards our eventual goal of \ch,
since the technical clutter emanates from the way we encode
\twodtucker\ in terms of \ch, and by reducing from \vt, we simplify
the proof that our final reduction to \ch\ does indeed work.

\begin{definition}\label{def:subregion}
A {\em subregion} of the plane consists of an equivalence class of points $(x,y)$ 
that are equivalent when their binary expansions are truncated after
$n+4$ bits of precision; thus any subregion is a square with an edge
length of $\frac{1}{16}2^{-n}$.
\end{definition}

\begin{definition}\label{def:tile}
Consider pairs $(a,b)$ of non-negative even numbers, for which
either $a$ and $b$ are both multiples of 4, or neither are.

Define the {\em $(a,b)$-tile} to be a union of 8 subregions arranged as in Figure~\ref{fig:vt},
with central point at $(\frac{1}{16}a\cdot2^{-n},\frac{1}{16}b\cdot2^{-n})$,
having a height and width of $\frac{1}{4}2^{-n}$.
(Thus all horizontal and vertical line segments have coordinates that
are multiples of $\frac{1}{16}2^{-n}$.)
If $a$ or $b$ is equal to zero, the tile consists of just the parts
of this region with non-negative coordinates.
\end{definition}
Observe that (for the values of $a,b$ allowed in Definition~\ref{def:tile})
tiles tessellate the positive quadrant of the plane as in Figure~\ref{fig:vt}.

\begin{definition}\label{def:vt}
An instance of \vt\ with complexity parameter $n$, consists of a boolean circuit
$C$ that takes as input $2n+22$ bits. These input bits represent the coordinates of a point
$(x,y)$ for $x,y\in[0,1]$, each of $x$ and $y$ represented as a bit string with $n+11$
binary places of precision.
$C$ has 4 boolean outputs that we use to represent the values $1,-1,2,-2$, respectively as $1110,0001,0111,1000$.\footnote{Note that we can enforce syntactically that these are only values that the output of the circuit can take. We use this convention instead of the usual $2$-bit circuit output in order to simplify the construction of the Consensus-Halving instance in Section \ref{sec:consensus}.} 
$C$ obeys the following constraints that may be enforced syntactically:
\begin{enumerate}
\item if $y<\frac{3}{8}-x$ then $C$ must output $1$;
\item if $y>\frac{5}{8}-x$ then $C$ must output $-1$;
\item\label{ppty3} if $y>x+\frac{1}{8}$ then the output of $C$ should be opposite to
its output on $(1-x,1-y)$, and similarly for points with $y<x-\frac{1}{8}$;
\item the output value of $C$ may not depend on the last 7 bits of $x$ or $y$;
\item Moreover, $C$'s output value is constant within tiles
(Definition~\ref{def:tile}, Figure \ref{fig:vt}):
a tile consists of 8 square regions with 128 discrete points along their
edges, arranged as in Figure~\ref{fig:vt}.
\item\label{exception} We allow the following exception to the above rules, which is that for
input bit-strings that represent
points that lie adjacent to the boundary of any subregion, $C$'s output value is unrestricted.
\end{enumerate}
A solution consists of a sequence of 100 points $(x_i,y_i)$ for $1\leq i\leq 100$,
where $y_1\leq 1-x_1$, and for $i>1$ we have $x_i=x_{i-1}+2^{-(n+11)}$ and $y_i=y_{i-1}-2^{-(n+11)}$,
where addition and subtraction are taken modulo 1.
These 100 points should
contain a set of 10 points that all produce the same output, and
another set of 10 points that produce the opposite output.
In the case that $y_1<100\cdot 2^{-(n+11)}$ and the sequence of points ``wraps around'',
this property must instead hold after
we negate the outputs of the wrapped-around subsequence.
\end{definition}

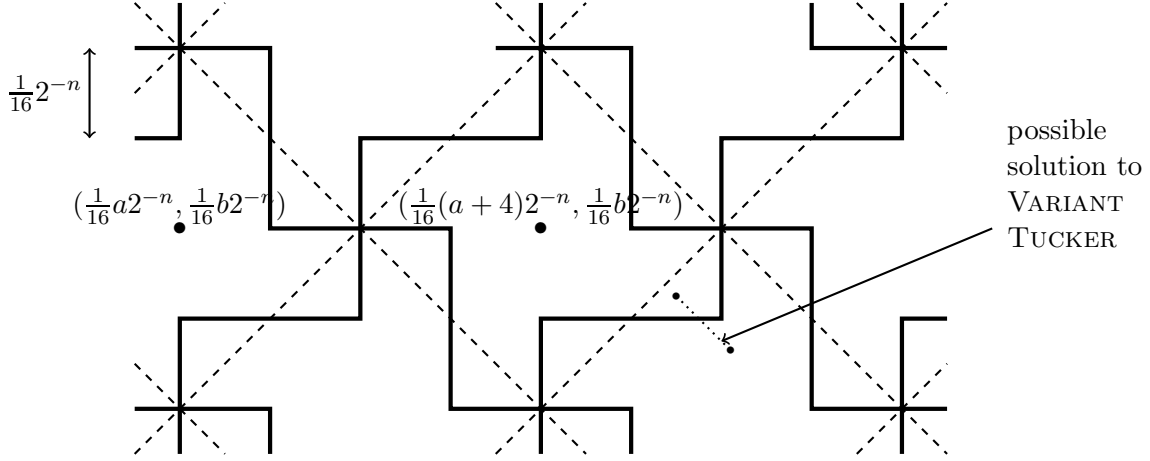
\begin{figure}
\center{
\begin{tikzpicture}[scale=0.6]

\draw[thick,dashed](-1,-1)--(9,9);
\draw[thick,dashed](7,9)--(17,-1);
\draw[thick,dashed](-1,9)--(9,-1);
\draw[thick,dashed](7,-1)--(17,9);
\draw[thick,dashed](-1,1)--(1,-1);
\draw[thick,dashed](-1,7)--(1,9);
\draw[thick,dashed](15,-1)--(17,1);
\draw[thick,dashed](15,9)--(17,7);

\draw[ultra thick](7,8)--(10,8)--(10,4)--(14,4)--(14,0)--(17,0);
\draw[ultra thick](8,9)--(8,6)--(4,6)--(4,2)--(0,2)--(0,-1);
\draw[ultra thick](-1,8)--(2,8)--(2,4)--(6,4)--(6,0)--(10,0)--(10,-1);
\draw[ultra thick](-1,6)--(0,6)--(0,9);
\draw[ultra thick](-1,0)--(2,0)--(2,-1);
\draw[ultra thick](8,-1)--(8,2)--(12,2)--(12,6)--(16,6)--(16,9);
\draw[ultra thick](14,9)--(14,8)--(17,8);
\draw[ultra thick](16,-1)--(16,2)--(17,2);

\draw[thick,<->](-2,6)--(-2,8);\node at(-3,7){$\frac{1}{16}2^{-n}$};

\node at(0,4){$\bullet$};\node at(0,4.5){$(\frac{1}{16}a2^{-n},\frac{1}{16}b2^{-n})$};
\node at(8,4){$\bullet$};\node at(8,4.5){$(\frac{1}{16}(a+4)2^{-n},\frac{1}{16}b2^{-n})$};

\draw[thick,dotted](11,2.5)--(12.2,1.3);
\node at(11,2.5){\tiny $\bullet$};\node at(12.2,1.3){\tiny $\bullet$};
\draw[thick,->](18,4)--(12,1.5);
\node[text width=2cm]at(20,5){possible solution to \vt};

\end{tikzpicture}
\caption{ \small{
Tiles in \vt\ are regions enclosed by the heavy lines.
Horizontal line segments have $y$-coordinates that are multiples of $2^{-n}$,
vertical line segments have $x$-coordinates that are multiples of $2^{-n}$.
Since numbers have $n+7$ bits of precision, each of the 8 square regions
contained in a tile has 128 discrete points along its edges.
}}\label{fig:vt}}
\end{figure}

\begin{figure}
\center{
\begin{tikzpicture}[scale=0.6]

\draw[ultra thick](0,0)--(12,0)--(0,12)--cycle;
\draw[thick](0,4.5)--(-0.3,4.5);\node at(-1,4.5){$(0,\frac{3}{8})$};
\draw[thick](0,7.5)--(-0.3,7.5);\node at(-1,7.5){$(0,\frac{5}{8})$};
\node at(-1,12){$(0,1)$};\node at(-1,0){$(0,0)$};\node at(12,-1){$(1,0)$};
\draw[thick,dashed](0,7.5)--(7.5,0);\draw[thick,dashed](0,4.5)--(4.5,0);
\draw[thick,dashed](1.5,3)--(3,4.5);\draw[thick,dashed](3,1.5)--(4.5,3);
\node at(1,1){$C=1$};\node at(5,5){$C=-1$};
\node at(3,3){$R$};\node at(1.5,4.5){$R'$};\node at(4.5,1.5){$R''$};

\end{tikzpicture}
\caption{ \small{
Region of points where we look for a solution to \vt.
$R$ is the region that contains a copy of $2D$-MS-\tucker.
}}\label{fig:tri-region}}
\end{figure}
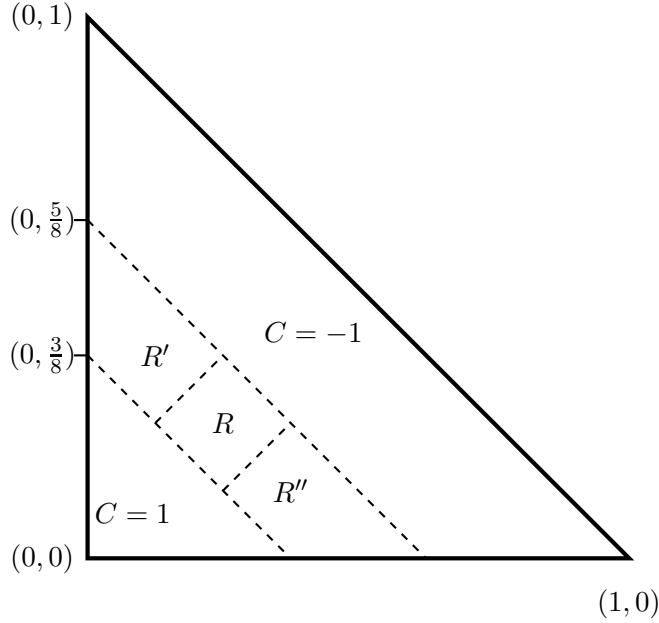

\begin{lemma}
\vt\ is \ppa-complete.
\end{lemma}

\begin{proof}
We reduce from \twodmstucker.
Squarelets in an instance $I_{2DMST}$ of \twodmstucker\ correspond to tiles
in an instance $I_{VT}$ of \vt\ as follows.

$I_{2DMST}$ contains squarelets $(i,j)$ for $1\leq i,j\leq 2^n$.
Each $(i,j)$ squarelet determines the value taken by $C$ on
the $(2\cdot2^n+2i+2j,4\cdot2^n-2i+2j)$-tile.
With this rule, the squarelets of $I_{2DMST}$ are mapped into tiles in the region
$R$ in Figure~\ref{fig:tri-region} in such a way that adjacencies are
preserved: two squarelets are adjacent if and only if their corresponding
tiles are adjacent.

Suppose that the monochromatic sides of $I_{2DMST}$ are squarelets $(i,j)$
with $j=1$ having label 1, and squarelets $(i,j)$ with $j=2^n$ having label $-1$.
As a result, these squarelets get mapped to sides
of the region $R$ that are adjacent to and match the monochromatic
regions adjacent to $R$
(the regions where $y<\frac{3}{8}-x$, alternatively $y>\frac{5}{8}-x$).

Any tile in the remaining parts $R',R''$ of the triangular domain of
Figure~\ref{fig:tri-region}
is allocated the same colour as its closest (Euclidean distance) tile in $R$.
That is, the colour of the $(2\cdot2^n+2j,4\cdot2^n+2j)$-tile is allocated
to the $(2\cdot2^n+2j-2k,4\cdot2^n+2j+2k)$-tile, for positive integers $k$,
and the colour of the $(4\cdot2^n+2j,2\cdot2^n+2j)$-tile is allocated
to the $(4\cdot2^n+2j+2k,2\cdot2^n+2j-2k)$-tile, for positive integers $k$.
Notice that this rule obeys Property~\ref{ppty3} of Definition~\ref{def:vt},
due to the boundary condition on the colouring of squarelets in $I_{2DMST}$.

Given that $I_{2DMST}$ has a concise circuit that labels its squarelets,
it's not hard to see that the corresponding instance $I_{VT}$ has
a concise circuit that takes as input, points in the triangular region
(at the slightly higher numerical precision), checks which tile
a point belongs to, and labels it according to the above rules.

We claim that for a sequence of 100 points to contain
two sets of 10 points having opposite labels, as required
for a solution, this will only happen when that sequence crosses
two adjacent tiles having opposite labels.
Any sequence of 100 points constructed as in the problem definition,
may cross the boundaries of subregions in at most 2 places, resulting
in at most 4 points where $C$ can disobey the tile colouring due to the
exception in item (\ref{exception}).
So most of the points in the two sets of 10 oppositely-labelled points must indeed come from
two oppositely-labelled tiles.

If this happens in region $R$ of Figure~\ref{fig:tri-region}, the two tiles correspond
directly to two adjacent squarelets in $I_{2DMST}$ having opposite
labels. It could also occur in the regions $R'$ or $R''$,
in which case we find a solution in the closest edge of $R$.
Suppose the sequence of 100 points ``wraps around'',
i.e. straddles $R'$ and $R''$, crossing the line
between $(0,0)$ and $(1,0)$
and appearing just to the right of the line between $(0,0)$ and $(0,1)$.
Labels get negated at the point where we wrap around, but recall that
in this case, we flip the suffix of the sequence occurring in $R'$,
before applying the test
that two sets of 10 points have equal and opposite labels.
From such a sequence of points we can identify either of two solutions
on the north-west or south-east sides of $R$ that are closest to
the sequence of points.
\end{proof}

\section{Construction of the Consensus Halving Instance}\label{sec:consensus}

In this section, we describe how to construct an instance $I_{CH}$ of $\epsilon$-\ch\ from an instance of \vt, for inverse-exponential precision parameter $\epsilon$. At a high level, the domain $A$ of $I_{CH}$ will have two designated regions -- a small one, typically containing $2$ cuts in a solution, which represent coordinates of points in the triangular domain of Figure~\ref{fig:tri-region} (the ``coordinate-encoding region'') and a larger one for the encoding of the labelling circuit (the ``circuit-encoding region''). Certain sensoring gadgets will detect the position of coordinate-encoding cuts and will feed this information to a set of gadgets which encode the inputs to the labelling circuit of the \vt\ instance. This information will be propagated through the circuit-encoding gadgets and fed back to the coordinate-encoding region. The idea is that two designated agents of the Consensus Halving instance, which will be associated with the coordinate-encoding region, will only be satisfied with the balance between \lplus\ and \lminus\ if the detected cuts correspond to points on a sequence that is a solution to \vt\ (Sections \ref{sec:cer} and \ref{sec:cdr}).

The construction will actually encode multiple copies of the labelling circuit of \vt\ for two different reasons. The main reason is that for each copy of the circuit, the cuts in the coordinate-encoding region will encode a different point in the domain, with these points being sufficiently close to each other (this will be achieved by small shifts in the valuation blocks corresponding to the circuits in the coordinate-encoding region) and with all of them lying on the same line segment. We will ensure that a solution to $I_{CH}$ will correspond to (sufficiently many) points of this segment with coordinates in squarelets with equal and opposite labels. The other reason is to deal with ``stray cuts'', i.e. cuts that are intended to lie in the coordinate-encoding region but actually cut through the circuit-encoding region. These cuts might ``invalidate'' the circuits that they cut through, but the construction will ensure that the rest of the circuits will remain unaffected, and there will still be sufficiently many reliable points in the sequence (Section \ref{sec:sol} and Section \ref{sec:pf}).\\

\noindent More concretely, given an instance $I_{VT}$ of \vt\ with
complexity parameter $n$,
we will construct an instance $I_{CH}$ of $\epsilon$-\ch\ for $\epsilon = 2^{-2n}$.
Let $A$ denote the Consensus-Halving domain, an interval of the form $[0,x]$ where
$x$ is of size polynomial in $n$. Any agent $a$ in $I_{CH}$ has a
measure $\mu_a:A\longrightarrow\rset$ which will be represented by
a step function (having a polynomial in $n$ number of steps). 

\subsection{Regions and agents of the instance $I_{CH}$}.
The domain $A$ will consist of two main regions:\begin{itemize}
	\item The \emph{coordinate-encoding region} $[0,1]$ (abbreviated as the \emph{c-e region}).
	\item The \emph{circuit-encoding region} $(1,x]$ (abbreviated as \emph{R}).
\end{itemize} 
Our construction contains $\y$ copies of an encoding of the labelling circuit $C$ of 
\vt\ and for the purpose, the circuit-encoding region $R$ will be further divided into $\y$ non-intersecting sub-regions $R_1,\ldots,R_{\y}$, one for each copy of the circuit.
The regions $R_i$ are of equal length and constitute a partition of $R$.
We further divide each region $R_i$ into three sub-regions $R_i^{in}$, $R_i^{mid}$ and $R_i^{out}$, which again are non-intersecting and partition $R_i$. These regions 
correspond with parts of a circuit that deal respectively with the input bits, the intermediate bits and the outputs.\\

\noindent The instance $I_{CH}$ will have the following sets of agents:
\begin{itemize}
	\item
	$2$ {\em coordinate-encoding agents}
	$\alpha_1,\alpha_2$ whose valuation functions $\mu_{a_1}$, $\mu_{a_2}$ are only positive in $\bigcup_{i=1}^{\y} R_i^{out}$. (See Subsection \ref{sec:cer} and Figure \ref{fig:ce-agents}). 
	\item
	\y\ \emph{circuit-encoders} $C_1,\ldots,C_{\y}$ (see Subsection \ref{sec:cdr}). 
	\begin{itemize}
		\item
		Each $C_i$ has an associated circuit-encoding region $R_i$ of the domain.
		\item
		With each $R_i$, there is a polynomial number of associated {\em circuit-encoding agents}. Let ${\cal{A}}_i$ be the set of those agents; the set ${\cal{A}}_i$ consists of the following sets of agents.
		\begin{itemize}
			\item
			A set ${\cal{S}}_i \subset {\cal A}_i$ of $8(n+8)+1$ \emph{sensor agents} with value in $[0,1] \cup R_i^{in}$. Among those, there will be a designated agent that we will refer to as the \emph{blanket sensor agent}. (See Subsection \ref{sec:sensors} and Figures \ref{fig:sensor_agents} and \ref{fig:sensor_agents2}).
			\item A set ${\cal G}_i \subset {\cal A}_i$ of polynomially-many \emph{gate agents}, with value in $R_i^{in} \cup R_i^{mid} \cup R_i^{out}$. (See Subsection \ref{sec:gate} and Figure \ref{fig:gates}).
		\end{itemize}
	\end{itemize}
\end{itemize}
\input{overview_general}
\input{overview}
We associate one cut with each agent; recall that for agent $\alpha$, $c(\alpha)$ is the cut associated with the agent. In a solution to $I_{CH}$, 
for any agent $\alpha_i \in \mathcal{A}_i$, these cuts will lie in a specific region, where most of the value of agent $\alpha_i$ 
will be concentrated. We will use $R(\alpha_i)$ to denote this region.
The cuts $c(a_1), c(a_2)$ for the coordinate-encoding agents, are called the \emph{coordinate-encoding cuts} 
and the associated region for them is the c-e region, i.e. $R(\alpha_1)=R(\alpha_2)=[0,1]$.
We will see that in any solution, either both or one of the coordinate-encoding
cuts must lie in the coordinate-encoding region and the other cuts must lie in region $R$.
In the event that a coordinate-encoding cut lies outside the c-e region, we refer to it as a \emph{stray cut}, 
and while such a cut may initially appear to interfere with the functioning
of the circuitry, we will see that the duplication of the circuit using \y\
circuit-encoders, allows it to be robust to this problem. 
(For the appropriate definitions and the details, see Section \ref{sec:pf}).

\subsection{Useful gadgetry}\label{sec:gadgets}

\noindent \textbf{Parity gadgets:} We will use term ``element'' to refer to a sub-interval $R(\alpha_i)$ where most of the valuation of agent $\alpha_1$ lies. For most elements of our construction, we will assume that for a cut that intersects the block (normally either a rectangle valuation block or two thin blocks of larger height or both), the label to the left of the cut is $\lplus$ and the label to the right of the cut is $\lminus$. Since the labels are generally alternating (as otherwise cuts can be merged), to achieve this, we will need to be able to switch the ``parity'' of the label sequence. This will be achieved with the following very simple \emph{parity gadget}. 

We construct an agent $\alpha_{par}$ that has a single valuation block (i.e. an interval where the agent has a constant, non-zero value) of sufficiently small height and width, in a region between two such distinct valuation blocks of some other agent or agents (where we need the parity switch to take place), and furthermore, no other agent has any value in that interval. Since we are only allowed to use $n$ cuts, in a solution to $I_{CH}$, only one cut is allowed to lie in the region $c(\alpha_{par})$ and therefore intersect this valuation block; obviously the cut has to lie close to the midpoint of the valuation block interval and it will switch the parity of the cut sequence. Throughout the reduction, we will not explain how to explicitly place the parity gadgets in the instance of $I_{CH}$ but rather we will assume without loss of generality that the left-hand sides of the cuts are labelled $\lplus$, unless stated otherwise.\\

\noindent \textbf{Bit detection gadgets:} Throughout the reduction, we will make use of specific blocks of valuations that we will refer to as \emph{bit detection gadgets}. The bit detection gadgets will be two \emph{thin} and \emph{dense} valuation blocks of relatively large height and relatively small length, situated next to each other, with no other valuation block in between them (e.g. see Figure \ref{fig:ce-agents}, the valuations of the top two agents in $R_1^{out}$ or Figure \ref{fig:boolean-gate-gadgets}, the values in the $out$ intervals). The precise volume of each valuation block will depend on the corresponding agents, but they will always constitute most of the agent's valuation over the related interval. The point of these gadgets is that if the discrepancy between \lplus\ and \lminus\ is (significantly) in the favour of one against the other, there will be a cut intersecting one of the two valuation blocks; which block is intersected will correspond to a $0/1$ value, i.e. a bit that indicates the ``direction'' of the discrepancy in the two labels.\\

\noindent \textbf{Boolean gate gadgets:} Consider a boolean gate that is either an AND an OR or a NOT gate, denoted $g_\land$, $g_\lor$ and $g_\neg$ respectively. Let $in_1$, $in_2$ and $out$ be intervals such that $|in_1| = |in_2| = |out| = 1$. We will encode these gates using the following gate-gadgets.\\

	\begin{center}
		$ g_{\neg}(in_1,out)=$ 
		$\begin{aligned}
		\begin{cases}
		0.25								& \text{if } t \in in_1 \\
		7.5			& \text{if } t \in [\ell(out),\ell(out)+1/20] \\
		7.5			& \text{if } t \in [r(out)-1/20,r(out)] \\
		0								& \text{otherwise}   
		\end{cases}
		\end{aligned}$
		\bigskip
		
			$g_{\lor}(in_1,in_2,out)=$
		$\begin{aligned}
		\begin{cases}
		0.125								& \text{if } t \in {in_1} \cup {in_2} \\
		6.25		& \text{if } t \in [\ell(out),\ell(out)+1/20] \\
		8.75		& \text{if } t \in [r(out)-1/20,r(out)] \\
		0								& \text{otherwise} 
		\end{cases}
		\end{aligned}$
		\bigskip
		
	$g_{\land}(in_1,in_2,out)=$
	$\begin{aligned}
	\begin{cases}
	0.125							& \text{if } t \in {in_1} \cup {in_2} \\
	8.75		& \text{if } t \in [\ell(out),\ell(out)+1/20] \\
	6.25		& \text{if } t \in [r(out)-1/20,r(out)] \\
	0								& \text{otherwise} 
	\end{cases}
	\end{aligned}$
	\end{center}
	\bigskip

\noindent Note that the gadget corresponding to the NOT gate only has one input, whereas the gadgets for the AND and OR gates have two inputs. In the interval $out$, each gadget has two bit detection gadgets - in the case of the NOT gate these are even, but in the case of the AND and OR gates, they are uneven. (see Figure \ref{fig:boolean-gate-gadgets}). Also note that for the inputs, as well as the output of the NOT gate, the label on the left-hand side of the cut is $\lplus$ and the label on the right-hand side will be $\lminus$, whereas for the outputs to the OR and AND gate, the label on the left-hand side of the cut is $\lminus$ and the label on the right hand side is $\lplus$ . This can be achieved with the appropriate use of parity gadgets (see Figure \ref{fig:boolean-gate-gadgets}).

\input{boolean_gate_gadgets}

\begin{observation}\label{obs:gate_gadgets}
	The boolean gate gadgets described above encode valid boolean NOT, OR and AND operations.
\end{observation}

\begin{proof}
These gadgets encode the boolean gate operations in the following way: We will interpret the position of a cut $c$ relatively to $\ell(in_1)$, $\ell(in_2)$ and $\ell(out)$ as the input or the output to the gates respectively. Specifically, for $j \in \{1,2\}$, if $c \in [\ell(in_1),\ell(in_1)+\epsilon]$, the input will be $0$ and if $c \in [r(in_1)-\epsilon,r(in_1)]$, the input will be $1$. Similarly, if $c \in [\ell(out),\ell(out)+1/20]$, the output will be $0$ and if $c \in [r(out)-1/20,r(out)]$, the output will be $1$. If the inputs or the outputs lie on any other point in the corresponding intervals, the gate inputs and outputs are undefined, but it will be enforced by our construction that in a solution to $I_{CH}$, this will never happen.

For $g_{\neg}$, let's assume that the cut in $in_1$ lies in $[\ell(in_1),\ell(in_1)+\epsilon]$, which means that a total value of approximately $0.25$ is assigned to $\lminus$ in the interval $in_1$ (recall that all cuts have an $\lplus$ label on their left-hand side). To compensate, since the agent only has further value in $out$, the cut would have to lie in $[r(out)-1/20,r(out)]$ and therefore by the interpretation of the inputs above, we can see that when the input is $0$, the output is $1$ and the gate constraint is satisfied. If the cut in $in_1$ lies in $[r(in_1)-\epsilon,r(in_1)]$ then the value in the interval $in_1$ has been labelled $\lplus$ and for the same reason, the cut in $out$ has to lie in $[\ell(out),\ell(out)+1/20]$ which encodes the case when the input is $1$ and the output is $0$. The arguments for the $g_{\lor}$ and $g_{\land}$ gadgets encoding the OR and AND gates respectively are very similar (noting that the cut intersecting $out$ will have $\lminus$ on its left-hand side). See Figure \ref{fig:boolean-gate-gadgets} of an illustration.
\end{proof}

\subsection{The coordinate-encoding agents}\label{sec:cer}

The coordinate-encoding region $[0,1]$ is the region from which the value of the solution to $\epsilon$-\ch\ will be read and will be translated to coordinates of a grid point on $I_{VT}$. Associated with this region, there are $2$ coordinate-encoding agents $\alpha_1$ and $\alpha_2$. However, these agents will have $0$ value in the subinterval $[0,1]$ and all of their value will lie in the circuit-encoding region $R_i$ and specifically in $\bigcup_{i=1}^{\y} R_i^{out}$. 

The value of the c-e agents in $R_i^{out}$ will corresponds to the feedback mechanism from the blanket sensor agent of $R_i$ (see Subsection \ref{sec:sensors}) and the feedback mechanisms from the gate-agents corresponding to the output gates of the circuit (see Subsection and \ref{sec:feedback}). Concretely the valuation of the coordinate-encoding agents is defined as follows:

\[ \mu_{\alpha_1}(t)= \begin{cases*}
\ \  30/800  & if $t \in \bigcup_{i=1}^{\y}[\ell(R_{i}^{out})+1/10,\ell(R_i^{out})+2/10]$,\\
\ \ 30/800  & if $t \in \bigcup_{i=1}^{\y}[\ell(R_{i}^{out})+1-2/10,\ell(R_i^{out})+1-1/10]$,\\
\ \ 1/400 & if $t \in \bigcup_{i=1}^{\y}[\ell(R_{i}^{out})+1.25,\ell(R_i^{out})+1.75]$,\\
\ \ 1/400 & if $t \in \bigcup_{i=1}^{\y}[\ell(R_{i}^{out})+3.25,\ell(R_i^{out})+4.75]$.
\end{cases*} \]%

\[ \mu_{\alpha_2}(t)= \begin{cases*}
\ \  30/800 & if $t \in \bigcup_{i=1}^{\y}[\ell(R_{i}^{out})+1/10,\ell(R_i^{out})+2/10]$,\\
\ \ 30/800  & if $t \in \bigcup_{i=1}^{\y}[\ell(R_{i}^{out})+1-2/10,\ell(R_i^{out})+1-1/10]$,\\
\ \ 1/400 & if $t \in \bigcup_{i=1}^{\y}[\ell(R_{i}^{out})+5.25,\ell(R_i^{out})+5.75]$,\\
\ \ 1/400 & if $t \in \bigcup_{i=1}^{\y}[\ell(R_{i}^{out})+7.25,\ell(R_i^{out})+7.75]$.
\end{cases*} \]%
Note that for each $i=1,\ldots,100$, the value of the agent in $R_{i}$ adds up to $1/100$ and therefore the agent's total valuation in $R$ is $1$. Intuitively, each coordinate encoding region has values that consist of the following components in each region $R_i$:
\begin{itemize}
	\item A bit detection gadget positioned in the interior of the interval where the bit detection gadget of the corresponding blanket sensor agent is situated (see also Subsection \ref{sec:sensors}).
	\item Two blocks of valuation situated in the interior of the intervals where the bit detection gadgets of the output gate agents are situated in $R_i^{out}$ (see also Subsection \ref{sec:gate}). There are four output gate agents in each $R_i$; $\alpha_i$ has value in the corresponding intervals for two of those and $\alpha_2$ has value in the corresponding intervals for the other two.
\end{itemize}
For an illustration, see Figure \ref{fig:ce-agents}.

\input{ce_agents}

\subsection{The circuit encoders}\label{sec:cdr}

In this subsection, we explain how to design the circuit-encoders $C_1, \ldots, C_{\y}$. Recall that these are sets of agents of $I_{CH}$ that encode the labelling circuit $C$ of \vt, including the inputs and the outputs to the circuit, via the use of sets ${\cal A}_1,\ldots {\cal A}_{\y}$ of circuit-encoding agents.
In the set ${\cal A}_i$, there are two different types of circuit-encoding agents:
\begin{itemize}
	\item The sensor agents ${\cal S}_i$ that are responsible for extracting the binary representation of the positions of the cuts in the c-e region, which will be used as inputs to the remaining circuit-encoding agents. These agents have value in $[0,1] \cup R_i^{in}$. Among those agents, there is a designated agent $\alpha_i^{bs}$ that we refer to as the \emph{blanket sensor agent}. 
	\item The gate agents ${\cal G}_i$ implement a circuit $C_i$, consisting of sub-circuits: a pre-processing circuit $C_i^\mathrm{pre}$ and a main circuit $C_i^\textrm{main}$, which further consists of a copy $C_i^{VT}$ of the labelling circuit $C$ of \vt\ as well as a small ``XOR operator'' circuit $C_i^{det}$. The pre-processing circuit will be responsible for transforming the information extracted from the sensor agents into the encoding of a point on the domain, which is then fed to $C_i^{VT}$. In particular, for each gate of the circuit $C_i$, we will have one associated agent of $I_{CH}$. For each gate agent that corresponds to an input gate of $C_i$, the agent will have value in $R_i^{in} \cup R_i^{mid}$ and for each gate agent that corresponds to an output gate of $C$, the agent will have value in $R_i^{out} \cup R_i^{mid}$. All other gate agents will have value only in $R_i^{mid}$. (See Subsection \ref{sec:gate}).
\end{itemize}
In the next subsections, we design the values of those agents explicitly. We will first explain how to construct the circuit-encoder $C_1$ and then based on this, we will construct the remaining circuit-encoders $C_2, \ldots, C_{\y}$. 

\subsubsection{The sensor agents}\label{sec:sensors}

In this subsection, we will design the set of \emph{sensor agents}, which is perhaps the most vital part of the construction. Roughly speaking, these agents will be responsible for detecting the position of a cut in the c-e region and extracting its binary representation. The set ${\cal S}_1$ contains $8(n+8)+1$ sensor agents, which consist of
\begin{itemize}
	\item the \emph{blanket sensor} agent $\alpha_1^{bs}$, that is responsible for detecting large discrepancies in the lengths of $\lplus$ and $\lminus$ in the c-e region,
	\item $8(n+8)$ \emph{bit-extractors}: 8 sets of $n+8$ agents, each set responsible for extracting a bit string of length $n+8$, which indicate the positions of cuts with respect to $8$ different intervals that span the c-e region; we will refer to these inputs as the \emph{raw data}. The raw data will then be inputted by the pre-processing circuit-encoding $C_i^{\mathrm{pre}}$ and will be transformed into the $n+11$ most significant bits of the binary representation of the positions of the cuts.
\end{itemize}
For an illustration, see Figures \ref{fig:ce-agents}, \ref{fig:sensor_agents} and \ref{fig:sensor_agents2}. We design these sets of agents below.

\subsubsection*{The blanket sensor agent} The valuation of the blanket sensor agent $\alpha_1^{bs}$ is defined as:

\[ \mu_{\alpha_1^{bs}}(t)= \begin{cases*}
\ \ 0.1 & if  $t \in [0,1]$,    \\
\ \  8.5 & if $t \in [\ell(R_{1}^{out}),\ell(R_1^{out})+1/20]$,\\
\ \ 8.5 & if $t \in [\ell(R_{1}^{out})+19/20,\ell(R_1^{out})+1]$, \\
\ \ 0.05 & if $t \in [\ell(R_{1}^{out})+1/4,\ell(R_1^{out})+3/4]$.\\
\end{cases*} \]%
In other words, the blanket sensor agent has a valuation block of volume $0.1$ spanning over the whole coordinate encoding region and two dense valuation blocks of volume $0.425$ over the intervals $[\ell(R_{1}^{in}),\ell(R_1^{in})+1/20]$ and $[\ell(R_{1}^{in})+19/20,\ell(R_1^{in})+1]$, with a valuation block of volume $0.025$ between them, in $[\ell(R_{1}^{in})+0.25,\ell(R_1^{in})+0.75]$, (see Figure \ref{fig:ce-agents} or Figure \ref{fig:sensor_agents}). Note that this latter part of the valuation is quite similar to a bit detection gadget, except for the fact that there is a small valuation block in between the two valuation blocks of large volume, which still constitute most of the agent's valuation over $A$. Furthermore, note that since the length of $R_1^{out}$ is polynomial in $n$, the whole valuation of agent $a^{bs}_1$ lies in $[0,1] \cup R^{out}_1$.

The blanket sensor agent is responsible for detecting large enough discrepancies in $\lplus$ and $\lminus$. As we will see, if such a discrepancy exists, the blanket sensor agents will provide feedback to the c-e agents, making sure that this is not a solution to $I_{CH}$. We state the lemma here but postpone its proof for Section \ref{sec:pf}.
\begin{restatable}{lemma}{sensor}\label{lem:blanket}
Let $\lplus^{(c-e)}$ and $\lminus^{(c-e)}$ be the total fraction of the c-e region labelled $\lplus$ and $\lminus$ respectively. The blanket sensor agents ensure that in a solution to $I_{CH}$, it holds that $$|\lplus^{(c-e)} - \lminus^{(c-e)}| \leq 1/4.$$
\end{restatable}
\noindent Whenever the blanket sensor agent $\alpha_i^{bs}$ does detect such a discrepancy (and therefore the cut $c(\alpha_i^{bs})$ in $R(\alpha_i^{bs})$ assumes one of the extreme positions, left or right), we will say that the blanket sensor agent is \emph{active} and that it \emph{overrides the circuit} $C_i$. Otherwise, we will say that the blanket sensor agent is $\emph{passive}$.\\

\subsubsection*{The bit extractors} 

The second set of agents in ${\cal S}_1$ will be responsible for detecting the position of the cuts and extracting their binary expansion. To be more precise, these agents will extract $8$ binary numbers of length $n+8$, from $8$ consecutive intervals of length $1/8$ each, which span the c-e region, and this number will encode the position of the cut within the interval. We will refer to these extracted binary strings as the \emph{raw data}. 

 Lemma \ref{lem:blanket} ensures that it is not possible for two cuts to intersect the same interval of length $1/8$. If for some interval of length $1/8$ there are no cuts intersecting it, the corresponding bit extractors will output a binary string which will consists of only $1$'s or only $0$'s; we will refer to such bit strings as $\emph{solid}$. 

\begin{definition}[Solid String]
A binary string is called \emph{solid} if either all of its bits are $1$ or all of its bits are $0$.
\end{definition}
\noindent If the interval is intersected by one cut, the bit extractors will output a binary string consisting of a non-trivial mixture of $0$'s and $1$'s.

The raw data extracted from the bit-extractors will be fed into the encoders of the input gates of $C_i$, and in particular to the pre-processing circuit $C_i^\textrm{pre}$ that will transform the extracted information into the binary representation of the coordinate of a point $(x,y)$ on the domain, which will then be fed into the encoding $C_i^{VT}$ of the labelling circuit $C$. We explain this in more detail in Section \ref{sec:gate}.\\

\noindent For each bit extractor, there are another $n+7$ sensor agents that will have exactly the same value in $[0,1]$. In particular, this value will be $1/10$ in volume, spanning over an interval of length $1/8$. We will refer to those $n+8$ sensor agents as \emph{c-e identical}, precisely because they have the same valuation in the c-e region. There will be exactly $8$ sets of $n+8$ c-e identical sensor agents. For $i=2,\ldots,8$, the values of the c-e identical agents for $i$ will be shifted by $1/8$ to the right, compared to the values of the c-e identical agents of $i-1$. Therefore, the set of sensor agents covers the whole c-e region. (See Figure \ref{fig:sensor_agents2}).

We will use $\alpha_{j,k}^s$ to denote a bit extractor agent, where $j \in \{1,\ldots,8\}$ and $k \in \{1,\ldots,n+8\}$. Note that here, we drop the subscript referring to the specific circuit encoder $C_1$ for ease of presentation and since there is no ambiguity. \\

\noindent \textbf{The agents in $\mathbf{[0,1/8]}$:} First, we will define the valuations of agents $\alpha_{1,k}^s,\ldots,\alpha_{1,n+8}^s$ and we will explain how to construct the valuations of the remaining agents from these agents. Note that these agents are c-e identical. First, let 

\[ \mu_{\alpha_{1,1}^s}(t)= \begin{cases*}
4/5 & if $t \in [0,1/8]$\\ 
9& if $t \in [\ell(R_i^{in}),\ell(R_i^{in})+1/20]$\\
9& if $t \in [\ell(R_i^{in})+1-1/20,\ell(R_i^{in})+1]$
\end{cases*}
\]
Then, define for $k=2,\ldots,n+8$,
\[ \mu_{\alpha_{1,k}^{s}}(t)= \begin{cases*}
4/5 & if $t \in [0,1/8]$\\ 
9-\sum_{j=1}^{k-1} (1/2)^j& if $t \in [\ell(R_i^{in})+2(k-1),\ell(R_i^{in})+2(k-1)+1/20]$\\
9-\sum_{j=1}^{k-1} (1/2)^j& if $t \in [\ell(R_i^{in})+2(k-1)+1-1/20,\ell(R_i^{in})+2(k-1)+1]$
\end{cases*}
\]
Additionally, for $j=1,\ldots,k-1$, for $t\in [\ell(R_i^{in})+2j+0.25,\ell(R_i^{in})+2j+0.75]$, we have that $\mu_{a_{1k}}^{s}(t)=1/5 \cdot (1/2)^j$.

\begin{restatable}{proposition}{bitextractors}\label{obs:bitextractors}
	Given a cut in the interval where $n+8$ c-e identical bit extractors have their value in the c-e region (e.g. the interval $[0,1/8]$ for the first $n+8$ c-e identical agents), the bit extractors recover a binary string of length $n+8$ which encodes the cut position in that interval.
\end{restatable}

\input{bitextractors}

\noindent \textbf{The remaining bit extractors:} Next, we design the remaining $8$ sets of c-e identical agents. These will be shifted versions of the first $n+8$ bit-extractors, where their valuations in the c-e region will be shifted by $1/8$ to the right (thus spanning the whole c-e region) and their valuations in $R_1$ will lie in ``clean'', non-overlapping intervals.
More concretely, for the agents $\alpha_{jk} \in \mathcal{S}_i$, $j=2,\ldots,8$, we define a \emph{correspondence function} $h_{\mathcal{R}_A,\mathcal{R}_B}: \mathcal{R}_A \longrightarrow \mathcal{R}_B$, mapping points of an interval $\mathcal{R}_A$ to an interval $\mathcal{R}_B$ in the most straightforward way: For $t \in \mathcal{R}_A$, let $h_{\mathcal{R}_A,\mathcal{R}_B}(t) = t-\ell(\mathcal{R}_A)+\ell(\mathcal{R}_B)$. In other words, any two points $x \in \mathcal{R}_A$ and $y \in \mathcal{R}_B$ such that $x-\ell(\mathcal{R}_A) = y-\ell(\mathcal{R}_B)$ are \emph{corresponding points} with regard to the two sub-regions. For $j=1,\ldots,8$, let $\mathcal{R}_j \in R_i^{in}$ be the sub-interval $[\ell(R_i^{in})+(j-1)2N+2,\ell(R_i^{in})+(j-1)2N+2N+1]$, where $N=2n+7$. Then for all $j = 2, \ldots, 8$, for all $k \in \{1,\ldots, n+8\}$,
\begin{itemize}
	\item
	$\mu_{\alpha_{jk}^s}(x)=\mu_{a_{(j-1),k}^s}(y)+1/8$ if $x \in \mathcal{R}_j$ and
	\item
	$\mu_{\alpha_{jk}^s}(x)=\mu_{a_{1k}^s}(y)$ if $x \in \mathcal{R}_j$, $y \in \mathcal{R}_1$ and $y = h_{\mathcal{R}_1,\mathcal{R}_j}(x)$.
\end{itemize}

\noindent The role of the bit extractors is to cover the whole c-e region in order to be able to detect the positions of cuts that lie anywhere in it. The reason for having $8$ shifted versions instead of a single detector is that the bit-extraction units are only operative if their inputs are intersected by at most one cut. Using these smaller valuation blocks, this guarantee is provided by the blanket sensor agent, according to Lemma \ref{lem:blanket}.

\input{bitextractors2}

\begin{restatable}{proposition}{allbitextractors}\label{prop:allbitextractors}
	The bit-extractors $\mathcal{S}_i$ can extract the binary representation of a point $(x,y)$ on the domain, represented by a set of cuts in the c-e region.
\end{restatable}

\noindent The proof of the proposition is left for Section \ref{sec:pf}.

\subsubsection{The gate agents}\label{sec:gate}

In this section, we will design the agents that will be responsible for encoding (i) the pre-processing circuit $C_i^\mathrm{pre}$ that transforms the raw data into coordinates of points $(x,y)$ of the domain and (ii) the circuit $C_1^\textrm{main}$, which will consist of the encoding $C_i^{VT}$ of the labelling circuit of \vt, as well as a ``XOR operator'' circuit that will flip the label of the final outcome when needed. These agents will eventually provide feedback (in terms of a discrepancy of labels $\lplus$ and $\lminus$) to the c-e agents. 

\subsubsection*{Implementing the circuit using the gate gadgets}
For both circuits (which we will view as a combined circuit in our implementation), at a high level, we will simulate the gates by gate agents, using the boolean gate gadgetry that we presented in Subsection \ref{sec:gadgets}. In particular, for any two-input gate $g$ of the circuit with inputs $in_1,in_2$, agent $\alpha^g$ will have a bit detection value gadget that will encode the output of the gate and furthermore, it will have value in some intervals $\mathcal{R}_k$ and $\mathcal{R}_{\ell}$ where the values of $in_1$ and $in_2$ lie respectively, where $in_1$ and $in_1$ can either be the outputs of some gates $g_1$, $g_2$ of some previous level, or the outputs of the sensor agents, if $g$ is an input gate of $C_i^\textrm{pre}$. For an illustration, see Figure \ref{fig:gates}. The case of $g$ being a single-input gate is similar. The construction will make sure that agent $\alpha^g$ will only be satisfied with the consensus-halving solution if the gate constraint is satisfied.\\

\input{gates}

\noindent Concretely, we will use the gate gadgets from Subsection \ref{sec:gadgets} that will encode the gates of the circuit. For each gate of $C_1$, we will associate a gate agent $\alpha_{1}^g,\ldots,\alpha_{|C_1|}^g$ with valuation given by the gadget
\[ \mu_{a^g_i}(t)= \begin{cases*}
g_{T}(in^i_1,out^i) & if $T \in \{ \lor,\land\}$\\ 
g_{T}(in^i_1,in^i_2,out^i)& if $T= \neg$%
\end{cases*}
\]
where $in_1^i$, $in_2^i$ and $out^i$ are non-overlapping intervals that will be defined separately based on whether $a_i^g$ corresponds to an input gate, an output gate or an intermediate gate of the circuit. Again here, we drop the subscript corresponding to the circuit-encoder $C_1$ for notational convenience.

First for the $L$ input gates of the circuit $C_1^\mathrm{pre}$, for each input gate-agent $a_i^g$ with $i \in \{1,\ldots,L\}$, we have that $in_1^i,in_2^i \in R_1^{in}$ and $out^i \in R^{mid}_1$. Specifically, $in^i_1$ and $in_2^i$ are subintervals in $$\bigcup_{i=1}^{N}\left([\ell(R_1^{in})+2k,\ell(R_1^{in})+2(k-1)+1/20] \cup [\ell(R_1^{in})+2(k-1)+19/20,\ell(R_1^{in})+2k+1]\right)$$ 
where the output of the bit extractors of Subsection \ref{sec:sensors} lie (see Subsection \ref{sec:sensors} and Figure \ref{fig:sensor_agents}). In simple words, the intervals from which the binary outcomes of the bit extractors are read (via the position of the cuts) are the input intervals to the input gate gadgets of $C_1^\textrm{pre}$. The output intervals $out^i$ are subintervals of $R_1^{mid}$ that do not overlap with any other intervals in $R_1^{mid}$.\\

\noindent Next, for any agent $\alpha_i^g \in \{\alpha^g_{L+1}, \ldots, \alpha^g_{|C_1| -4}\}$, 
\begin{itemize}
	\item $in_1^i$ and $in_2^i$ are the intervals $out^{j_1}$ and $out^{j_2}$ of agents $\alpha_{j_1}^g, \alpha_{j_2}^g \in \mathcal{G}_1$, where $\alpha_{j_1}^g, \alpha_{j_2}^g$ correspond to the inputs $g_{j_1}$ and $g_{j_2}$ of gate $g_i$ in $C_1$
	\item $out^i$ is an interval of $R^{mid}_i$ which does not overlap with any interval $in_1^k, in_2^k$ or $out^k$, for any $k < i$.  
\end{itemize} 
The definitions for the intervals $in_1^i$ and $out^i$ of the agents $\alpha_i^g$ that correspond to NOT gates are very similar.\\

\noindent Finally, for the gate agents $\alpha_i^g \in \{\alpha^g_{|C_1| -3}, \ldots, \alpha^g_{|C_1|}\}$, corresponding to the output gates $g_{out}^{1},g_{out}^{2}, g_{out}^{-1}$ and $g_{out}^{-2}$ of $C_1$, we have that 
\begin{itemize}
	\item $in_1^i$ and $in_2^i$ are the intervals $out^{j_1}$ and $out^{j_2}$ of agents $\alpha_{j_1}^g, \alpha_{j_2}^g \in \mathcal{G}_1$, where $\alpha_{j_1}^g, \alpha_{j_2}^g$ correspond to the inputs $g_{j_1}$ and $g_{j_2}$ of gate $g_i$ in $C_1$.
	\item $out^i$ is one of the subintervals in which a $1/400$-fraction of the value of a coordinate-encoding agent lies, i.e. 
	\begin{eqnarray*}
		out^i \in &[\ell(R_{1}^{out})+1.25,\ell(R_1^{out})+1.75] \cup \\
		&[\ell(R_{1}^{out})+3.25,\ell(R_1^{out})+3.75] \cup \\
		&[\ell(R_{1}^{out})+5.25,\ell(R_1^{out})+5.75] \cup \\
		&[\ell(R_{1}^{out})+7.25,\ell(R_1^{out})+7.75]. \ \ 
	\end{eqnarray*}	
	Note that there are $4$ such subintervals and the output of each of the four gates $g_{i}$ for $i \in \{|C_1|-3,\ldots,|C_1|\}$ lies in one of those subintervals.
\end{itemize} 

\subsubsection*{The pre-processing circuit $C_1^{\textrm{pre}}$} 
As we mentioned earlier, the pre-processing circuit inputs the raw data extracted from the circuit encoders and outputs the binary expansion of the coordinate of the detected position $(x,y)$ in the c-e region. Then, the outcome of $C_1^{\textrm{pre}}$ is fed directly into the input gates of $C_1^{VT}$ and the information is propagated through the circuit, resulting in the assignment of a label for the encoded point. Here, we explain the operation of the pre-processing circuit in more detail.

For ease of reference, let $R_j^{1/8}=[(j-1/8),j/8]$, for $j=1,\ldots,8$. Assume first that there are two cuts $c(\alpha_1)$ and $c(\alpha_2)$ in the c-e region and observe that by Lemma \ref{lem:blanket}, they can not intersect consecutive regions $R_{j-1}^{1/8}$ and $R_j^{1/8}$ for any $j \in \{2,\ldots,8\}$. Let $R_{k}^{1/8}$ and $R_{\ell}^{1/8}$ be the regions intersected by the cuts, with $k < \ell$. Assume first that $k \neq 1$, i.e. the region intersected by the first cut is not $[0,1/8]$. Then, for regions $R_1^{1/8},\ldots, R_{k-1}^{1/8}$, the bit-extractors will output the same solid string $b_1=\ldots=b_{k-1}$, consisting of either only $1$'s or only $0$'s, while $R_k^{1/8}$ outputs a binary string $b_k$, with a non-trivial mix of $0$'s and $1$'s. Let $B_k$ represent the number corresponding to the string $b_k$. Then, the pre-processing circuit sets:
\begin{itemize}
	\item $x = \frac{k-1}{8} + \frac{B_k}{2^{n+8}}$, if $b_1$ is a string of $1$'s and
	\item $x = \frac{k-1}{8} + \frac{2^{n+8} - B_k}{2^{n+8}}$, if $b_1$ is a string of $0$'s.
\end{itemize}
In simple words, the circuit reads the raw data from the bit-encoders for regions $R_1,\ldots,R_{k-1}$ and depending on whether it is a string of $1$'s or a string of $0$'s, it interprets the bit extracted from $R_{k}$ as the distance from the left or the right endpoint of the interval respectively. In the event where $c(\alpha_1)$ intersects $R_1^{1/8}$, the information is obtained similarly from the bit-extractors of region $R_2^{1/8}$ (which have to output a solid string by Lemma \ref{lem:blanket}, as discussed previously). Specifically, the pre-processing circuit in that case sets:
\begin{itemize}
	\item $x = \frac{B_k}{2^{n+8}}$, if $b_2$ is a string of $0$'s and
	\item $x = \ \frac{2^{n+8} - B_k}{2^{n+8}}$, if $b_2$ is a string of $1$'s.
\end{itemize}
where $b_2$ is the binary string extracted from the bit-extractors of $R_2^{1/8}$.

For the $y$ coordinate, notice that in every region $R_{j}^{1/8}$ with $j \in (k,\ell)$, the bit-extractors will output the same solid string, consisting of either only $1$'s or only $0$'s. The fact that set is non-empty is again guaranteed by Lemma \ref{lem:blanket} according to the discussion above. Again, let $b_{\ell}$ be the binary string detected by bit-extractors of region $R_\ell^{1/8}$ and let $B_\ell$ be the corresponding number and let $b_{k+1}=\ldots=b_{\ell-1}$ be the solid strings outputted by the bit-extractors of regions $R_{k+1}^{1/8},\ldots,R_\ell^{1/8}$. The pre-processing circuit sets:
\begin{itemize}
	\item $y = \frac{\ell}{8}-\frac{B_\ell}{2^{n+8}}$, if $b_{\ell-1}$ is a string of $0$'s and
	\item $y = \frac{\ell}{8}- \frac{2^{n+8} - B_\ell}{2^{n+8}}$, if $b_{\ell-1}$ is a string of $1$'s.
\end{itemize}
Now consider the case when there is only one cut in the c-e region. In that case, in a solution to $I_{CH}$, the cut can not intersect regions $R_1^{1/8}$ or $R_8^{1/8}$ as that would violate Lemma \ref{lem:blanket} and there exist regions both to the left and to the right of the region $R_k^{1/8}$ intersected by the cut which provide solid strings as inputs to the pre-processing circuit. In that case, the pre-processing circuit sets $x=0$ and
\begin{itemize}
	\item $y = \frac{k}{8}-\frac{B_k}{2^{n+8}}$, if $b_{1}$ is a string of $0$'s and
	\item $y = \frac{k}{8}-\frac{2^{n+8} - B_k}{2^{n+8}}$, if $b_{1}$ is a string of $1$'s.
\end{itemize}
In the event that some cut lies exactly in the boundary of two consecutive regions $R_{j-1}^{1/8}$ and $R_j^{1/8}$ for some $j \in \{2,\ldots,8\}$, the only difference is that the circuit does not receive an input with a non-trivial mix of $1$'s and $0$'s for this cut, but rather a sequence of $k-1$ solid strings consisting of $z$'s followed by a solid string of $|1-z|$'s, extracted by the bit-extractors of region $R_k^{1/8}$. In that case, the circuit operates exactly as before, with the cut lying in region $R_k^{1/8}$. 

\subsubsection*{The main circuit $C_1^{\textrm{main}}$} 

The encoding of the circuit $C_1^\textrm{main}$ will consist of the encodings of two subcircuits, the labeling circuit $C_i^{VT}$ of \vt\ and the XOR operator circuit. 

The input to the circuit $C_1^{VT}$ is the binary representations of the coordinates of a grid point within a squarelet of $I_{VT}$ outputted by the pre-processing circuit. Recall that each squarelet contains a set of grid points with a resolution of $2^{7}$ in each dimension (see Figure \ref{fig:vt}). The output is a label $\{\pm 1, \pm 2\}$; in particular, the output gates of $C$ are $g_{out}^{1},g_{out}^{-1}, g_{out}^{2}$ and $g_{out}^{-2}$ and the following correspondence is syntactically enforced (Definition \ref{def:vt}): 
$$(1 \rightarrow 1110),\\
(-1 \rightarrow 0001),\\
(2 \rightarrow 0111), \\
(-2 \rightarrow 1000).$$

\noindent \textbf{The XOR operator circuit:} This circuit will perform a simple operation, using the raw data gathered from the bit-extractors and the ouputs of $C_1^{VT}$. Recall the definition of the solid strings as well as the regions $R_j^{1/8}$ for $j=1,\ldots,8$ from earlier, with $b_j$ denoting the raw data outputted by the bit-extractors of region $R_j^{1/8}$. Let $rep(b)$ denote any bit of a solid bit-string $b$ (since they are all the same) and let $\overline{rep(b)}$ denote its complement. 
First, a sub-circuit $C_i^\textrm{det}$ will perform the following operation on the raw data. 
\begin{itemize}
	\item If the output $b_1$ of the bit-extractors in $[0,1/8]$ is solid, then output a bit $z=\overline{rep(b_1)}$.
	\item If the output $b_1$ of the bit-extractors in $[0,1/8]$ is not solid, then output a bit $z=rep(b_2)$.
\end{itemize}
Note that by Lemma \ref{lem:blanket}, if $b_1$ is not solid, $b_2$ has to be solid.
Finally, if $z_1z_2z_3z_4$ are the outputs of $C_i^{VT}$, the XOR operator outputs a string $y_1y_2y_3y_4$ with $y_i = z_i \oplus z$, for $i=\{1,2,3,4\}$, where $\oplus$ denotes the Exclusive-OR operation. The effect of this operation is that either the circuit produces the same label or all the outputs are flipped and the circuit produces the opposite label, depending on the raw data. This will be particularly useful when arguing about the stray cut in Section \ref{sec:pf}.\\

\noindent For an illustration of the encoded circuits, see Figure \ref{fig:circuit}. \\

\input{circuit}

\noindent It is not hard to see that our gadgets simulate the operation of the circuit. The proof of the following proposition basically follows from the construction and is included in Section \ref{sec:pf}.

\begin{restatable}{proposition}{gates}\label{obs:gateagents}
	The gate agents in $R_1$ simulate the circuit $C_1$, consisting of the pre-processing circuit $C_1^\mathrm{pre}$ and the copy $C_1^\textrm{main}$ of $C$.
\end{restatable}

\subsubsection{Construction of the circuit-encoders $C_2,\ldots,C_{\y}$.}\label{sec:restofcircuits}
In the previous subsections, we constructed the circuit-encoder $C_1$ corresponding to the first copy of the labelling circuit; here we explain how to construct the remaining circuit-encoders relatively to the agents of $C_1$. Recall the definition of the correspondence function $h_{\mathcal{R}_A,\mathcal{R}_B}: \mathcal{R}_A \longrightarrow \mathcal{R}_B$, from Section \ref{sec:sensors}: for $t \in \mathcal{R}_A$, let $h_{\mathcal{R}_A,\mathcal{R}_B} = t-\ell(\mathcal{R}_A)+\ell(\mathcal{R}_B)$. In other words, any two points $x \in \mathcal{R}_A$ and $y \in \mathcal{R}_B$ such that $x-\ell(\mathcal{R}_i) = y-\ell(\mathcal{R}_B)$ are \emph{corresponding points} with regard to the two sub-regions.

We construct the circuit-encoder $C_i$ as follows. For $j=1,\ldots, |A_i|$, let $a_{ij}$ denote the $j$'th circuit-encoding agent in $\mathcal{A}_i$. Then for all $i = 2, \ldots, \y$, for all $j \in \{1,\ldots, |A_i|\}$,
\begin{itemize}
	\item
	let $\mu_{a_{ij}}(x)=\mu_{a_{1j}}(y)$ if $x \in R_i$, $y \in R_1$ and
	$y = h_{R_1,R_i}(x)$.
	\item
	let $\mu_{a_{ij}}((x+(i-1)\cdot 2^{-(n+11)})\mod(1))=\mu_{a_{1j}}(x)$, if $x \in [0,1]$, i.e. if $x$ is in the c-e region.
\end{itemize}
The second of these items says that in the c-e region, the valuation function of the agents that make up $C_i$ differ from those of
$C_1$ by having been shifted to the left by $2^{-(n+11)}\cdot (i-1)$, where
this shift ``wraps around'' in the event that we shift below $1$
(the left-hand point of the c-e region). In other respects, $C_i$ is an exact copy of $C_1$, save that $C_i$'s
internal circuitry lies in $R_i$ rather than $R_1$.\\

\noindent \textbf{The virtual cuts:} For the circuit encoders $C_2,\ldots,C_{100}$ it will often be useful to think of the following alternative interpretation of their inputs. 
Consider the two cuts (the case of one cut is similar) in the c-e region, at positions $c^1_1,c^1_2$, encoding a point $(x,y)$ of the domain (also see Section \ref{sec:sol}). Since $C_2$ is a version
of $C_1$ where all the values in the c-e region are shifted by $2^{-(n+11)}$ to the left (wrapping around for some valuations), we can equivalently think of the output
of $C_2$ as \emph{what the output of $C_1$ would have been if the cuts had been moved slightly to the right, i.e. to $c^2_1=c_1+2^{-(n+11)}$ and $c^2_2=c_2+2^{-(n+11)}$ respectively.}
In other words, for each circuit-encoder $C_i$, we can think of its output as the output of $C_1$ if the cut were placed at $c^i_1$ and $c^i_2$. We will refer to such cuts as \emph{virtual cuts}.

\subsection{Recovery of a solution of $I_{VT}$ from a solution of $I_{CH}$}\label{sec:sol}

In this subsection, we explain how to obtain a solution to $I_{VT}$ from a solution to $I_{CH}$. Recall from Section \ref{sec:tucker-reds} that a solution to $I_{VT}$ is a sequence of points $(x_1,y_1),(x_2,y_2),\ldots,(x_{100},y_{100})$ of the (discrete) domain, lying on a line segment such that two sets of at least $10$ of these points each have coordinates in squarelets of equal and opposite labels. Specifically, for each point $(x_i,y_i)$ on the segment, with $i<\y$, it holds that $x_{i+1}=x_i+2^{-(n+11)}$ and $y_{i+1}=y_i-2^{-(n+11)}$ (See Figure \ref{fig:vt}).

Now consider a solution $\mathcal{H}$ to $I_{CH}$. As we will establish in Section \ref{sec:pf}, in $\mathcal{H}$ there must exist one or two cuts situated in the coordinate-encoding region $[0,1]$. \begin{itemize}
	\item If there is only one cut in $[0,1]$, situated at $z\in[0,1]$, let $x=0$ and $y=1-z$ be the coordinates of a point on the domain.
\item If there are 2 cuts in $[0,1]$, situated at $z,z'$, let $x=z$ and $y=1-z'$ be the coordinates of a point in the domain.
\end{itemize}
If we use $n+11$ bits of precision to represent the coordinates $(x,y)$ of the point that correspond to the solution $\mathcal{H}$ above, we end up with the closest point $p$ on the discrete domain to $(x,y)$. Then, we can obtain a solution to $I_{VT}$ by generating a sequence of points $p_{1},p_{2},\ldots,p_{\y}$ by setting $p_1=p$ and $p_{i}=(x_{i-1}+2^{-(n+11)},y_{i-1}-2^{-(n+11)})$ for $2\leq i \leq \y$.

For an illustration of the mapping between the solutions to the two problems, see Figure \ref{fig:twosolutions}.

\input{twosolutions}

\subsection{The feedback mechanism to the c-e agents}\label{sec:feedback}

Now that we have explained what the construction of $I_{CH}$ looks like, we can explain how the c-e agents receive feedback from the circuit. The agents will only be satisfied with the partition of labels if the line segment of points $(x,y)$ crosses a boundary of tiles with same and opposite labels and there are sufficiently many points indisputably labelled with each one of those labels. 

First note that for a point $(x_1,y_1)$ of the domain recovered as described in Section \ref{sec:sol}, each point $(x_i,y_i)$ in the sequence of $100$ points will be labelled by a different copy of the circuit $C_i$. Consider such a copy and let $C_i(x_i,y_i)$ be its output; recall that $C_i(x_i,y_i) \in \{1110,0001,0111,1000\}$ (which can be syntactically enforced) and furthermore, we have the following correspondence of labels to outputs:
$$(1 \rightarrow 1110),
(-1 \rightarrow 0001),
(2 \rightarrow 0111) \text{ and }  
(-2 \rightarrow 1000).$$

\noindent Assume that the $C_i(x_i,y_i)=j$, for some $j \in \{1110,0001,0111,1000\}$ and let $\lambda_j \in \{1,-1,2,-2\}$ be the corresponding label. For each of the c-e agents $\alpha_1$ and $\alpha_2$, the contribution to $\lplus$ from its valuation on $R_{out}^i$ is\footnote{Assuming that $C_i$ behaves as expected, i.e. it receives good inputs and is reliable - see Section \ref{sec:pf} for the definitions.} \\

\begin{minipage}{0.47\textwidth}
\[
\text{For } \alpha_1, 
\begin{cases}
2/400, & \text{if } C_i(x_i,y_i)=1110 \\
-2/400, & \text{if } C_i(x_i,y_i)=0001 \\
0, & \text{otherwise.}
\end{cases}
\]
\end{minipage}
\begin{minipage}{0.47\textwidth}
\[
\text{For } \alpha_2, 
\begin{cases}
2/400, & \text{if } C_i(x_i,y_i)=0111 \\
-2/400, & \text{if } C_i(x_i,y_i)=1000 \\
0, & \text{otherwise.}
\end{cases}
\]
\end{minipage}

\bigskip
\noindent where a contribution of $-\mu$ to $\lplus$ here denotes a contribution of $\mu$ to $\lminus$. To see this, note that a set of the $4$ cuts corresponding to the output $1110$ of $C_i$ would lie respectively:
\begin{itemize}
	\item To the right of the leftmost valuation block of Agent $\alpha_1$ in $R_i^{out}$, thus labelling the whole block $\lplus$.
	\item To the right of the rightmost valuation block of Agent $\alpha_1$ in $R_i^{out}$, thus labelling the whole block $\lplus$.
	\item To the right of the leftmost valuation block of Agent $\alpha_2$ in $R_i^{out}$, thus labelling the whole block $\lminus$.
	\item To the left of the leftmost valuation block of Agent $\alpha_2$ in $R_i^{out}$, thus labelling the whole block $\lplus$.
\end{itemize}
See Figure \ref{fig:ce-agents} for an illustration. Since all of these valuation blocks have volume $1/400$ each, the total contribution to $\lplus$ from an output of $1110$ (and therefore a label of $1$) is $2/400$ for Agent $\alpha_1$, whereas for Agent $\alpha_2$, the total contribution is $0$ and the sub-partition restricted to $R_i^{out}$ is balanced. The argument for the remaining output/labels is very similar. \\

\noindent \textbf{The ``wrap-around'' labels:} In some cases, the circuit-encoders $C_i$ will detect points close to the boundary of the triangular domain of \vt\, in which case the sequence of points $1,\ldots,100$ extracted from the bit-extractors of the circuits will be part of a ``wrap-around'' line segment, i.e. a line segment that starts with some point with $y$ close to zero and ends with a point with $x$ close to $0$ (i.e. it crosses the bottom boundary of the triangle region). In this case, Definition \ref{def:vt} requires that the ``equal-and-opposite'' property holds after we flip the labels of the wrapped-around subsequence.

In terms of $I_{CH}$, this situation occurs when (i) either there are two cuts $c_1$ and $c_2$ in the c-e region and $c_2$ sits very close to $1$ or (ii) when there is only one cut in the c-e region (which can be thought of as another cut being situated exactly at $0$). In either case, since each circuit encoder $C_i$ detects a virtual cut, (which is a shifted version of the cut detected by $C_{i-1}$ as explained earlier), this sequence of points will be correctly generated by the reduction. For example, where there is only one cut $c$ in the c-e region, while the bit-extractors of $C_1$ only ``see'' that cut, the bit-extractors of each circuit-encoder $C_2, \ldots, C_{100}$ ``see'' another cut, situated at position $i \cdot 2^{-n-11}$. This is because the ``wrapped-around'' valuation of the first $n+8$ c-e identical sensors of $C_i$  ``sees'' both $\lplus$ (on the left side of $c$) and $\lminus$ (on the right side of $c$), and therefore detect that a cut intersects the region - this is the virtual cut $c_v$ detected by $\mathcal{S}_i$ (similarly for the case of two cuts).
 
Interpreting the virtual cut $c_v$ as the actual cut, the circuit-encoder $C_i$ now ``sees'' the label $\lminus$ on the left-hand side of $c_v$ and $\lplus$ on the right-hand side. Intuitively, $C_i$ interprets the input as if the left endpoint of the region was $1-i\cdot 2^{-n-11}$ (i.e. as if we cut the c-e region at the point where the wrap-around value starts and glued the cut piece to the end of the c-e region), with the sequence of labels starting with $\lminus$. The pre-processing circuit $C_i^\textrm{pre}$ ensures that the correct point of the wrapped-around subsequence is encoded, and the XOR operator circuit of $C_i$ flips the label of this point, as desired by Definition \ref{def:vt}.\footnote{This is similar to the operation of the ``double-negative effect'' (see Lemma \ref{lem:straycut}), but without the flip of the label sequence introduced by the stray cut, so it is rather a ``single negative'', flipping the label of the outcome, as intended.}

\section{Proof of the Reduction}\label{sec:pf}

In this section, we prove the correctness of the reduction, i.e. that given a solution $\mathcal{H}$ of $\epsilon$-\ch, we can recover a solution to \vt\ (and therefore to \tucker, given our results in Section \ref{sec:tucker-reds}.)
The main result of this section is the following:
\begin{theorem}\label{thm:main}
\vt\ is polynomial-time reducible to \ch.
\end{theorem}

We state the following useful definition regarding the copies of the circuit $C$.

\begin{definition}\label{def:goodinputs}
Let $C_i$,  be one of the 100 copies of the circuit $C$ in an instance $I_{CH}$
of $\epsilon$-\ch\ as constructed in Section~\ref{sec:consensus}.
We say that $C_i$ {\em receives good inputs} with respect to
positions $(x,y)$ of the c-e cuts, if $C_i$
receives valid boolean-encoding inputs extracted from $x$ and $y$.
\end{definition}
For example, in the case of $i=1$, $C_1$ receives good inputs provided
that the point $(x,y)$ of the domain of \vt\ is not too close to the boundary of a sub-region.
A simple observation is the following.

\begin{observation}\label{obs:onlyfourbad}
In $I_{CH}$, at most 4 copies of $C$ do not receive good inputs.
\end{observation}

\noindent The observation is based on the density of the domain of \vt. Given the resolution used for the grid points within the square regions, there can be at most $4$ points that are very close to the boundary of a subregion. 

Next, we prove Lemma \ref{lem:blanket} (firstly stated in Section \ref{sec:sensors}), which establishes that the blanket-sensor agent detects large discrepancies in the total length of the $\lplus$ and $\lminus$ intervals.

\sensor*

\begin{proof}
	Assume that in the c-e region, there is a discrepancy of labels which is larger than $1/4$, i.e. $|\lplus^{(c-e)} - \lminus^{(c-e)}| > 1/4$. Then, since the blanket sensor agent $\alpha_1^{bs}$ has valuation $0.1$ distributed uniformly on $[0,1]$, this implies that $|\mu_{\alpha_1^{bs}}(\lplus \cap [0,1]) - \mu_{\alpha_1^{bs}}(\lminus \cap [0,1])| > 1/40$, i.e the discrepancy that the blanket sensor detects in the volume of the two labels is at least $1/40$. Then, in a solution to $I_{CH}$, it must be the case that the cut in $c(\alpha_1^{bs})$ intersects one of the two thin valuation blocks in the interval, either the left one, if $\mu_{\alpha_1^{bs}}(\lplus \cap [0,1]) - \mu_{\alpha_1^{bs}}(\lminus \cap [0,1])| < 0$ or the right one,  if $\mu_{\alpha_1^{bs}}(\lplus \cap [0,1]) - \mu_{\alpha_1^{bs}}(\lminus \cap [0,1])| >0$, as otherwise agent $\alpha_1^{bs}$ would be dissatisfied with the balance of the labels in $A$. In that case however, Agents $\alpha_1$ and $\alpha_2$ will have all of their valuation in $R_1$ receive the same label, and this value constitutes $3/4$ of their total value in $R_1$ (see Figure \ref{fig:ce-agents}). Since for all $i=1,\ldots,100$, the value of each blanket-sensor agent $\alpha_i^{bs}$ is the same in the c-e region, the same will be true for most copies of the circuit $C_i$, in particular for each reliable copy of the circuit. 
	
	A bit more concretely, since there are at most $2$ stray cuts, there are at most $2$ 
	unreliable copies $C_j$ and $C_{j'}$ of the circuit. For any $k \in \{1,\ldots,100\}$ such that 
	$k \neq j,j'$, the blanket sensor agent $\alpha_{k}^{bs}$ will detect a value of $1$, 
	as the cut associated with $\alpha_{k}^{bs}$ in $\mathcal{R}_{\alpha_{k}^{bs}}$ must 
	intersect the right thin valuation block of the blanket sensor agent's value in $R_1^{in}$ 
	(see Figure \ref{fig:ce-agents}), or otherwise $\mathcal{H}$ would not be a solution to $I_{CH}$. 
	But then, in the cut-encoding regions $R_i$, all of the valuation blocks of agent $\alpha_1$
	in $\cup_{i \in \{1,\ldots,100\},i\notin \{ j,j' \}} R_i^{in}$ (see Figure \ref{fig:ce-agents}), 
	will receive the same label. Since in each $R_i$, the value of agent $\alpha_i$ in 
	$R_{i}^{in}$ is a $(3/4)$-fraction of its total value in $R_i$ and since at least 98 copies are reliable, 
	at least a $(147/200)$-fraction of agent $\alpha_1$'s valuation will receive the same label and 
	therefore the agent is dissatisfied with the partition and $\mathcal{H}$ is not a consensus halving solution.
\end{proof}	

Using Lemma \ref{lem:blanket}, we can now prove the following lemma regarding the number of cuts in the c-e region, in any solution to $I_{CH}$.

\begin{lemma}\label{obs:cuts}
	In a solution to an instance $I_{CH}$ of $\epsilon$-\ch\ constructed as in
	Section~\ref{sec:consensus},
	the two c-e cuts are the only cuts that may occur in the c-e region,
	and at least one c-e cut must occur in the c-e region.
\end{lemma}

\begin{proof}
	To see this, note first that in any solution $\mathcal{H}$ to $I_{CH}$, all
	cuts apart from the c-e cuts, are constrained to lie in various
	intervals outside the c-e region. In particular, for every agent $\alpha_j \in \mathcal{A}_i$ 
	(i.e. every agent besides the two c-e agents $\alpha_1,\alpha_2$), 
	it holds that most of the valuation of the agent (in particular, sufficiently more than 
	a $(1/2)$-fraction) lies in a designated interval, which we will denote by $\mathcal{R}_{\alpha_{j}}$.
	Agent $\alpha_j$ is not the only agent that has non-zero value in $\mathcal{R}_{\alpha_j}$, 
	but it holds that for $j' \neq j$, $\mathcal{R}_{\alpha_j} \cap \mathcal{R}_{\alpha_{j'}} = \emptyset$
	i.e. each agent in $\mathcal{A}_i$ has a different designated interval. Also, note that none of 
	this intervals intersects with the c-e region, i.e. $\mathcal{R}_{\alpha_j} \cap [0,1]=\emptyset$, 
	for all agents $\alpha_j \in \mathcal{A}_i$.
	
	Obviously, by construction, for such an interval $\mathcal{R}_{\alpha_j}$, if there is no cut that
	intersects the interval, then agent $\alpha_j$ will be dissatisfied with the balance of $\lplus$ 
	and $\lminus$ and $\mathcal{H}$ will not be a solution to $I_{CH}$. Additionally, since there 
	are $N-2$ such designated intervals which do not intersect with the c-e region, $\mathcal{H}$ 
	must place at most $2$ cuts in the c-e region. This establishes the first statement of the Lemma.
	
	Now for the second statement, suppose that neither c-e cut lies in the c-e region, 
	in which case the c-e region gets labelled entirely $\lplus$. By Lemma \ref{lem:blanket},
	the blanket sensor agents will detect the discrepancy and $\mathcal{H}$ can not be a solution.
\end{proof}

\noindent Next, we prove the lemmas regarding the operation of the bit-extractors, which we stated in Section \ref{sec:sensors}.
\bitextractors*

\begin{proof}We will argue for the c-e identical bit extractors with value in $[0,1/8]$; the argument for the rest is similar. First of all, note that in a solution to $I_{CH}$ there can be at most one cut intersecting the interval $[0,1/8]$, otherwise the blanket-sensor agent would not be satisfied, by Lemma \ref{lem:blanket}. Assume that such a cut $c$ intersects the interval $[0,1/8]$. To recover the position, Agent $\alpha_{11}^s$ is responsible for determining whether the cut lies in the first or the second half of $[0,1/8]$. 
	If the cut lies in the first half, then the bit-detection gadget of the agent in $[\ell(\mathcal{R}),\ell(\mathcal{R})+1/20] \cup [\ell(\mathcal{R})+1/20,\ell(\mathcal{R})+1]$ will detect a $0$, with a cut intersecting (or sitting close to) the leftmost thin valuation block of its bit-detection gadget. This follows by the construction, since the cuts that intersect the outputs of the bit-extractors in $R_1^{in}$ have $\lminus$ on their left-hand side (see Figure \ref{fig:sensor_agents} and Figure \ref{fig:sensor_agents2}).
	
	In turn, Agent $\alpha_{12}^s$ will detect whether the cut lies in the first or the second half \emph{of the previously detected half} and the bit will be set accordingly, with the corresponding cut lying on the left thin valuation block of the bit-detection gadget (in case of $0$) and on the right valuation block (in case of $1$). This is achieved with the extra small block of valuation $1/20$ in $[\ell(\mathcal{R})+0.25,\ell(\mathcal{R})+0.75]$, which has already been labelled by the cut that intersects the output of Agent $\alpha_{11}^s$. One can view this as adding a ``compensation'' to the portion that is not in excess for the second agent (e.g. more $\lplus$ assuming the first detected bit was $0$), compared to the first agent. In particular, while the bit-detection gadget of the first agent uses a bit to detect the ``direction'' of the discrepancy, the bit-detection gadget of the second agent uses a bit to determine the direction of the discrepancy if additional value of $0.05$ is added to the lesser label. The argument for the remaining agents $\alpha_{ik}^s$, for $k=1,\ldots,n+8$ is very similar. 
	
	It should be noted that for the other copies $C_i$ , $i=1,\ldots,100$ of the circuit, the valuation blocks of the c-e identical agents in the c-e region might ``wrap around'', i.e. they can consist of valuation blocks in $[0,z_2]$ and $[z_2,1]$ where $|[0,z_2] \cup [z_1,1]|=1/8$. In that case, exactly the similar arguments apply if we consider the interval to be $[z_1,z_2]$, i.e. the first half of the interval is considered to be $[z_1,z_1+1/4]$ if $z_1+1/4 \leq 1$ and $[z_1,1] \cup [0,z_2-1/4]$ if $z_1+1/4>1$.
\end{proof}

\allbitextractors*

\begin{proof}
	Consider a set $\mathcal{S}_i^j \subset \mathcal{S}_i$ of c-e identical agents with value in $[j/8, (j+1)/8 ]$ of the c-e region, for some $j \in [0,\ldots,7]$. Assume first that a cut lies in $[j/8, (j+1)/8 ]$ and that no other cut lies in $[0,j/8)$. Then, (since by convention the first cut in the c-e region has $\lplus$ on its left-hand side), the $n+8$ c-e identical agents of region $[j/8, (j+1)/8]$ will detect the position of the cut in the interval and their outputs will feed that to the gate-agents, exactly as described for the c-e identical agents of $[0,1/8]$ in Section \ref{sec:sensors}, and according to Proposition \ref{obs:bitextractors}. 
	
	Now assume that that the second cut in the c-e region lies in $[j/8, (j+1)/8 ]$ and the first cut lies somewhere in $[0,j/8)$.
	Observe that the first cut must have been detected by another set $\mathcal{S}_i^{j'}$ of c-e identical bit-extractors, with $j' < j$. Since the agents in $\mathcal{S}_i^j$  are now extracting the position of the second cut, notice that the label on the left-hand side of the cut is now $\lminus$, which effectively ``flips'' the outputs of the bit-extractors (the bit-detection gadgets) $\mathcal{S}_i^j$ in $R_i^{in}$. However, since all this information is provided to the pre-processing circuit, the circuit can infer how to interpret the outputs (and particularly it can lead the outputs of the set $\mathcal{S}_i^j$ through a set of NOT gates). 
	
	In simpler words, if a cut has already been detected by a set of sensors, this informs the circuit on how to interpret the remaining inputs that correspond to the second cut. Similarly, the circuit can use the information that no cuts occur in the region $[j/8, (j+1)/8]$, which will be either a string of $1$s (if no cut has been detected in a previous interval) or a string of $0$s (if a cut has been detected in a previous interval). Since the circuit knows whether a cut has been detected in an interval $[j'/8, (j'+1)/8]$, with $j'<j$, it also knows how to interpret these trivial inputs.
	
	Finally, the pre-processing circuit $C_i^\textrm{pre}$ can combine the inputs from all the different intervals into a $(n+4)$-bit string which encodes the coordinates of a point $(x,y)$ in the domain.
\end{proof}

\noindent Next, recall the following proposition about the gate-agents, first stated in Subsection \ref{sec:gate}.
\gates*
\begin{proof}
	This is a rather straightforward observation based on the following facts.
	\begin{itemize}
		\item The bit-extractors of $\mathcal{S}_1$ extract the binary representation of the cuts in the sub-regions of length $1/8$ of the c-e region.
		\item This input is fed to the circuit $C_1$.
		\item The gate-gadgets implement the valid AND,OR and NOT circuit operations.
	\end{itemize}
	The first statement is argued in Proposition \ref{obs:bitextractors} and \ref{prop:allbitextractors} and the second statement can easily be verified from the construction of the input gate gadgets of this section. In particular, a cut intersecting the bit-detection gadget of a bit extractor agent is directly supplied as an input to the corresponding input gate gadget of $C_1$ (see Figure \ref{fig:gates}). The last statement follows from the correctness of these boolean gate gadgets that are used in each step, which was explained in Section \ref{sec:gadgets} and the fact that the the gate agents in $R_1$ are being used to implement the pre-processing circuit $C_i^\textrm{pre}$ and the circuit $C_i^\textrm{main}$, which is an exact copy of $C$, using the corresponding gadgets as the gates (also see Figure \ref{fig:gates}).
\end{proof}

\subsection{Dealing with the stray cut}

As we mentioned in Section \ref{sec:intro}, all agents other than the coordinate-encoding ones are associated with separate cuts. For all the circuit-encoding agents, these cuts are constrained to lie in different regions in $R$, but for the c-e agents, it is not a-priori clear that these cuts will lie in the c-e region. Lemma \ref{obs:cuts} establishes that in any solution of $I_{CH}$, at least one of these cuts will actually lie in the c-e region, but the other might actually move into the circuit-encoding region $R$. We will use the following definition.

\begin{definition}[Stray cut]
	In a solution $\mathcal{H}$ of $I_{CH}$ as described in Section \ref{sec:consensus}, a c-e cut will be called a \emph{stray cut} if it is occurs outside of the c-e region.
\end{definition}

\noindent A stray cut may have two effects on $\mathcal{H}$.
\begin{enumerate}
	\item It intersects the circuit-encoding region $R_i$ of some circuit encoder $C_i$, for $i \in \{1,\ldots,100\}$.
	\item It flips the parity of the circuit encoders $C_i$, with $R_i < c$, where $c$ is the position of the stray cut in $R_{i-1}$. In other words, if
	the first cut in $R_i$ was expecting to see $\lplus$ on its left-hand side, it now sees $\lminus$ and vice-versa.
\end{enumerate}
The first effect is not much of a problem; we simply deem this circuit unreliable:
\begin{definition}[Reliable copy]
We will say that a copy $C_i$ of the circuit $C$ ($i \in \{1,\ldots,100\}$) is \emph{reliable} if none of the c-e cuts intersects $R_i$.  A copy $C_i$ of the circuit is \emph{unreliable} if it is not reliable.
\end{definition}
Since there is only one stray cut, there is at most one unreliable circuit $C_i$. The error that this copy will introduce to the volumes of the labels $\lplus$ and $\lminus$ for the c-e agents (see Section \ref{sec:feedback}) will be relatively small due to the fact that there are many reliable points that receive good inputs. This is argued formally in the proof of Lemma \ref{lem:onecut}. 

The second effect from the ones above is potentially more troublesome however, since the parity flip could alter the outputs of the bit-extractors. This problem however is actually being taken care of by the pre-processing circuit (and the XOR operator of the main circuit). If outputs of the bit-extractors are flipped, the pre-processing circuit actually inputs the bit-wise complements of the raw data that it would input before the flip; these consist of binary strings that encode the positions of the cuts within regions $[(j-1)/8,j/8]$, as well as the accompanying information (the solid strings) that indicate how to interpret the aforementioned binary strings as coordinates $(x,y)$ that get some label by $C_i^{VT}$. The effects of these flips cancel out and the circuit outputs exactly the same label, which is then flipped by the XOR sub-circuit, to ensure that the c-e agents receive the same feedback. This is proven in detail in the following lemma.

\begin{lemma}[Double-negative lemma]\label{lem:straycut}
	Consider a solution $\mathcal{H}$ of $I_{CH}$ and a circuit-encoder $C_i$. If a stray cut is placed in $(1,\ell(R_i))$ (i.e. to the right of the c-e region and to the left of $R_i$), then the c-e agents see exactly the same balance of $\lplus$ and $\lminus$ in $R_i$ as they did before the insertion of the stray cut.
\end{lemma}

\begin{proof}
Since the stray cut lies between the c-e region and $R_i$, there is one cut in the c-e region
(the cut $c(\alpha_1)$) at some position $c \in [\frac{3}{8},\frac{5}{8}]$, which is ensured by Lemma \ref{obs:cuts} and 
Lemma \ref{lem:blanket}. Circuit $C_i$ evaluates the label at a point $(x,y)$ where $x=i \cdot 2^{-n-11}$ and $y=1-(c+i \cdot 2^{-n-11})$,
i.e. outputs the evaluation of circuit $C_1$ on the set of virtual cuts $c_1^i$ and $c_2^i$ respectively.

Consider the operation of adding a cut between the c-e region and $R_i$.
This effectively causes the output bits of the the bit extractors of $C_i^\textrm{pre}$ to flip,
as the cut in $R(\alpha_i^s)$ for every bit extractor $\alpha_i^s \in S_i$ is now ``seeing'' 
$\lplus$ on its left-hand side, rather than $\lminus$. We claim that the outputs of
$C_i^{VT}$ will be the same as before, regardless of the flip.

Using the notation from the description of the pre-processing circuit in Section \ref{sec:gate}, 
let $R_{j}^{1/8}$, for $j=1,\ldots,8$, denote the interval $[(j-1)/8,j/8]$ of the c-e region and let $b_j$ be 
the binary string that the bit-extractors of $C_i^\mathrm{pre}$ extract from this interval. Let $R_{k}^{1/8}$ be the
interval where the cut $c(\alpha_1)$ is detected and notice that by the discussion above, it must
be the case that $R_k^{1/8} \neq R_1^{1/8}$ and therefore the bit-extractors of each $R_j^{1/8}$, for $j \in \{1,\ldots,k-1\}$ 
must output the same solid binary string, consisting of either only $1$'s or only $0$'s.

Consider first the output of $C_i^{VT}$ in the absence of the stray cut and assume without loss of generality that 
the solid strings $b_1, \ldots, b_{k-1}$ outputted by the bit-extractors of the regions $R_1^{1/8}, \ldots, R_{k-1}^{1/8}$ 
are strings of $1$'s (the argument for when they are strings of $0$'s is completely symmetric). Then, according to 
the operation of the pre-processing circuit, the position of the cut is calculated by adding up $(k-1)/8$ and the 
distance $z$ between the detected position and $(k-1)/8$. See Figure \ref{fig:same_output} for an example when $k=3$.

Now consider the output of $C_i^{VT}$ in the presence of the stray cut in $(0,R_i)$. As we mentioned earlier, 
the raw data that the pre-processing circuit inputs have been bit-wise flipped. In effect, the following two things happen:
\begin{itemize}
	\item The solid strings $b_1 \ldots, b_{k-1}$ outputted by the bit-extractors of regions $R_1^{1/8}, \ldots, R_{k-1}^{1/8}$ 
	are strings of $0$'s,
	\item The bit extractors of $R_{k}^{1/8}$ now extract the string $\overline{b_{k}}$, i.e., the bit-wise complement of $b_{k}$.
	Note that if the bit-string $b_k$ encodes the position of a cut at $(k-1)/8 + z$, then $\overline{b_k}$ encodes the position of
	a cut at $(k-1)/8+z'=k/8-z$. Therefore, if the cut is actually placed at $(k-1)/8+z$, the effect of the flip is that
	the bit-extractors of $R_{k}^{1/8}$ now detect the cut as being placed at $k/8-z$.
\end{itemize}
However, since the solid binary strings of the bit-encoders preceding $R_k^{1/8}$ are now strings of $0$'s, the 
pre-processing circuit calculates the position of the cut as $k/8-z' = (k-1)/8 + z$ (again, see Figure \ref{fig:same_output} for an example with $k=3$). The effect of this ``double-negative'' operation is that the position of the cut is the same in both cases, which results in the circuit $C_i^{VT}$ producing the same label, whether a stray cut lies
in $(1,\ell(R_i))$ or not. 

If we used the output of $C_i^{VT}$ directly to provide feedback to the c-e agents, then the following undesired effect would take place. Comparing the situations before and after the cut, the balance of $\lplus$ and $\lminus$ shown to the c-e agents in $R_i$ would flip, because (i) the output of $C_i^{VT}$ is unaffected by the flip but (ii) the stray cut changes the parity of the label sequence, causing the parts of $R_i$ that were labelled $\lplus$ to now be labelled $\lminus$ and vice-versa. That would introduce a false discrepancy of the balance of $\lplus$ and $\lminus$ for one of the c-e agents, or more precisely, the correct ``amount'' of discrepancy but in the wrong direction.

This is taken care of by the XOR operator circuit; the circuit detects the value of the solid string and applies an appropriate XOR operation to the outcomes of $C_i^{VT}$. For example, if before the insertion of the stray cut, the solid strings $b_1 \ldots, b_{k-1}$ outputted by the bit-extractors of regions $R_1^{1/8}, \ldots, R_{k-1}^{1/8}$ 
are strings of $0$'s after the insertion of the stray cut, the outputs of the circuit $C_i$ are $z_j \oplus \overline{rep(b_1)}$, for $j \in \{1,2,3,4\}$, where $z_i$, $i \in \{1,2,3,4\}$ are the outputs of $C_i^{VT}$. In other words, the output of $C_i$ is the bit-wise complement of the ouput of $C_i^{VT}$ (or alternatively, $C_i$ outputs label $-\lambda$ if the label ouputted by $C_i^{VT}$ is $\lambda$. The effect of this ``double negative'' operation is that
\begin{itemize}
	\item Before the insertion of the stray cut in $(0,\ell(R_i))$, the outputted label was $\lambda$ and label sequence in $R_i^{out}$ started with $\lplus$.
	\item After the insertion of the stray cut n $(0,\ell(R_i))$, the outputted label is $-\lambda$ and label sequence in $R_i^{out}$ started with $\lminus$.
\end{itemize}
Given how labels correspond to discrepancies on $\lplus$ and $\lminus$ for the c-e agents, as explained in Section \ref{sec:feedback}, the c-e agents receive exactly the same feedback before and after the insertion of the stray cut. Note that if the XOR operator circuit received different input from the raw data (i.e. a solid string of $1$'s which indicates that no flip has taken place), then the XOR operation leaves the outputs of $C_1^{VT}$ unaffected (and there is no flipping of the label sequence in $R_i^{out}$ either).\\

\noindent For an illustration of the ``double negative'' effect, see Figure \ref{fig:doubleneg}.

%
%
%
\end{proof}
\input{same_output}
\input{doublenegative}
%
%
%
%
%
%
%

\subsection{Correctness lemmas}

\noindent By Lemma~\ref{obs:cuts}, we are left with two cases to consider: 
the first case when both c-e cuts lie in the c-e region, and the second case when
just one of the lies in the c-e region.

\begin{lemma}\label{lem:twocuts}
Consider a solution $\mathcal{H}$ to $I_{CH}$ in which both c-e cuts lie in the c-e region and consider
a set of points $(x_1,y_1),\ldots,(x_{100},y_{100})$ recovered from $\mathcal{H}$ as described in Section \ref{sec:sol}. 
Then $(x_1,y_1),\ldots,(x_{100},y_{100})$ is a solution to \vt.
\end{lemma}

\begin{proof}
Let $c_1(\alpha_1)$ and $c_2(\alpha_2)$ be the positions of the c-e cuts, which are assumed
in the statement of the lemma to both lie in $[0,1]$.
Since there are two cuts in the c-e region, by the recovery of the solution to \vt, 
we have $x=c_1(\alpha_1)$ and $y=1-c_2(\alpha_2)$, and the sequence of 100 points that
forms a solution to $I_{VT}$ consists of
$$(x_1,y_1),(x_1+2^{-(n+11)},y_1-2^{-(n+11)}),\ldots,(x_1+99\cdot2^{-(n+11)},y_1-99\cdot2^{-(n+11)}),$$
where addition/subtraction are taken modulo 1. 

By construction of the solution according to Section \ref{sec:sol}
and by the resolution of the domain, the bit extractors of $C_1$ extract the binary representation of the coordinates $(x_1,y_1)$,
according to Proposition \ref{obs:bitextractors} in Subsection \ref{sec:sensors}.
Then, as explained in Section \ref{sec:gate} and Proposition \ref{obs:gateagents}, these coordinates are propagated
via the gate agents in $\mathcal{G}_1$ and correspond to an output of $C_1$ (a bit-string of length $4$, where
there is a one-to-one correspondence between the labels $\{-1,1,2,-2\}$ and $4$ distinguished output bit-strings, namely 
$0001,1110,0111,1000$ respectively).

Since each copy of the circuit in the c-e region is a shifted version of the previous copy by $2^{-(n+11)}$, it is not hard to see
that the bit extractors of a reliable circuit $C_i$ that receives good inputs, actually detect the representation of point $(x_i,y_i)$ in the sequence of 100 
points originating with $(x_1,y_1)$. In precisely the same way, the output of this circuit feeds a discrepancy back to the c-e agents. Therefore,
in a solution $\mathcal{H}$ to $I_{CH}$, the points that are detected from the bit extractors of the circuits $C_1,\ldots,C_{100}$ will actually correspond to
the points in the sequence $(x_1,x_2), \ldots, (x_{100},y_{100})$.

As explained in Section \ref{sec:feedback}, each such output string corresponds to a labelling of the valuations of the c-e agents
in $R_1^{out}$ (the volumes of $\lplus$/$\lminus$ are balanced in $R_1^{in}$, since the blanket sensor agent $\alpha_1^{bs}$ is passive) 
and therefore there is a discrepancy in $R_1^{out}$ for exactly one c-e agent. Specifically, for Agent $\alpha_j$, with $j \in \{1,2\}$,
the discrepancy is in favour of $A_{k}$, with $k \in \{+,-\}$, if the label of the circuit output is $k j$.

One can easily check that for a c-e agent to be satisfied with the balance of the labels, it has to be the case that the excess in $\lplus$ or $\lminus$ due to a specific output in region $R_i^{out}$ has to be ``cancelled out'' from an excess of the opposite label ($\lminus$ or $\lplus$ respectively) in another interval $R_j^{out} \subseteq R_j$. For this to be possible, by construction, it has to be the case that the output of the corresponding circuit $C_j$ corresponds to the opposite label of the output of $C_i$, if that copy of the circuit
operates as intended. Therefore, if the points $(x_i,y_i)$ and $(x_j,y_j)$ are detected by reliable copies $C_i$ and $C_j$ that receive good inputs, they must have coordinates in different tiles of the domain, which are labelled with opposite labels. However, by the density of the domain and since there are at most $100$ points of the domain in the line-segment between $(x_i,y_i)$ and $(x_j,y_j)$, this is only possible if these points lie in neighbouring tiles of equal and opposite labels, i.e. in a solution to \vt.

A degenerate case occurs when some of the points $(x_i,y_i)$ in the sequence correspond to circuits that do not achieve good inputs (note that since both $c(\alpha_1)$ and $c(\alpha_2)$ lie in the c-e region, there are no stray cuts by definition). These are the points that lie close to the boundary of two tiles and their labels assigned by the circuit are 
unconstrained. This in principle can cause a cancellation effect and ``balance out'' the discrepancies of some unambiguously labelled point, when both of these points lie in the same tile 
(the former near the boundary and the latter in the interior). For example, for a point $p_1$ labelled $-1$ in some tile $j$, there can be a point $p_2$ close to the boundary with some neighbouring tile $j'$ (with tile $j'$ labelled $-1$ as well), that receives label $1$ by the circuit (due to the fact that the labelling rules of boundary points are unconstrained). In a sequence of points that contain both $p_1$ and $p_2$, the $\lplus/\lminus$ discrepancy due to $p_2$ will cancel out the $\lplus/\lminus$ discrepancy due to $p_1$, although we are not at a solution.

This is being take care of by the averaging manoeuvre, which uses $100$ copies of the circuit and requires that at least $10$ of the points in the sequence receive a label and $10$ other points receive an equal and opposite label. More concretely, assume by contradiction that we are at a solution $\mathcal{H}$ of $I_{CH}$, but the sequence of $100$ points do not correspond to a solution to $I_{VT}$. Let $\lambda$ be the label of the majority of the points in the sequence (breaking ties arbitrarily) and assume wlog that $\lambda = 1$. Observe that by the chosen resolution of the domain, it holds that at least $40$ points in the sequence must be labelled $1$. By the discussion above, since $\mathcal{H}$ is a solution, for every point labelled $1$, there must be another point in the sequence labelled $-1$, for the cancellation to take place. By Observation \ref{obs:onlyfourbad}, there are at most $4$ such points that are arbitrarily labelled and therefore they can contribute to a cancellation of at most $1/10$ of the excess of $\lplus$ due to the contribution of the points labelled $1$. This means that there must be at least $36$ points labelled $-1$ in the sequence and the sequence $(x_1,y_2), \ldots (x_{100},y_{100})$ is actually a solution to $I_{VT}$. 
\end{proof}

\begin{lemma}\label{lem:onecut}
	Consider a solution $\mathcal{H}$ to $I_{CH}$ in which only one c-e cut lies in the c-e region and consider
	a set of points $(x_1,y_1),\ldots,(x_{100},y_{100})$ recovered from $\mathcal{H}$ as described in Section \ref{sec:sol}. 
	Then $(x_1,y_1),\ldots,(x_{100},y_{100})$ is a solution to \vt. 
\end{lemma}

\begin{proof}
	The proof of the lemma is very similar to the proof of Lemma \ref{lem:twocuts}. Here, if $c$ is the position of the 
	single cut in the c-e region, we have that $x=0$ and $y=1-c$. Again, the binary expansion of $(x,y)$ is extracted 
	from the bit extractors of $C_1$ and the output of the encoded circuit will correspond to a discrepancy for the c-e
	agents in $R_1^{out}$ similarly as before. Again, the same is true for the remaining $99$ circuits, with the exception of possibly
	one circuit that might be unreliable due to the stray cut. From Lemma \ref{lem:straycut}, it holds that the feedback 
	of any reliable copy to the c-e agents is unaffected by the stray cut. 
	
	A stray cut intersecting interval $R_i$ might introduce some additional discrepancy in the volume of the two labels in $R_i$, which is upper
	bounded by the valuation of the coordinate-encoding agents in $R_i$, i.e. $1/100$. The effect that this could have is that this discrepancy might
	cancel out the discrepancies in favour of $\lplus$ or $\lminus$ introduced by at most $3$ reliable circuits that receive good inputs
	(which happens if all of the valuation of the c-e agent in $R_i$ is labelled $\lminus$ or $\lplus$ respectively). 
	
	However, similarly to before, this can ``invalidate'' at most $7$ points overall and there will still be $33$ points labelled $1$ whose
	contribution to $\lplus$ needs to be cancelled out by points labelled $-1$ and we will be at a solution to $I_{VT}$.
\end{proof}

\section{Equivalence of Consensus Halving and Necklace Splitting}\label{sec:necklace}

\input{necklace}

\section{Conclusions}\label{sec:conclusions}

We hope that the present work will lead to more \ppa-completeness results, starting with the Necklace Splitting problem. The reason for believing that \ns\ is \ppa-complete, is that it would be surprising if, in relaxing the approximation parameter $\epsilon$ from inverse-exponential to inverse-polynomial, we should lose \ppa-hardness but retain \ppad-hardness. That is what would be the case if \ns\ were merely \ppad-complete.

\end{document}

%% file: mstucker.tex
\begin{figure}
\centering
 \begin{tikzpicture}
    \draw[step=0.5cm,color=gray] (-3,-3) grid (3,3);
    \node[fill=blue,opacity=.1,text opacity=1] at (-0.75,+0.75) {1};
    \node  at (-0.25,+0.75) {2};
    \node at (+0.25,+0.75) {-1};
    \node at (+0.75,+0.75) {-1};
    \node[fill=green,opacity=.1,text opacity=1] at (-0.75,+0.25) {2};
    \node at (-0.25,+0.25) {-1};
    \node at (+0.25,+0.25) {-1};
    \node[fill=blue,opacity=.1,text opacity=1] at (+0.75,+0.25) {-1};
    \node[fill=red,opacity=.1,text opacity=1] at (-0.75,-0.25) {-1};
     \node at (-0.25,-0.25) {-1};
      \node at (+0.25,-0.25) {-2};
     \node[fill=green,opacity=.1,text opacity=1] at (+0.75,-0.25) {-2};
    
    \node at (-0.75,-0.75) {1};
     \node at (-0.25,-0.75) {1};
      \node at (+0.25,-0.75) {-2};
    \node[fill=red,opacity=.1,text opacity=1]  at (+0.75,-0.75) {-1};
    
    \draw[color=red,very thick] (-1,1) -- (-1,-1);
     \draw[color=red,very thick] (-1,1) -- (1,1);
     \draw[color=red,very thick] (1,1) -- (1,-1);
     \draw[color=red,very thick] (1,-1) -- (-1,-1);
     
    \draw[very thick] (-1,0.5) -- (-1.5,0.5);
     \draw[very thick] (-1.5,0.5) -- (-1.5,3);
     \draw[very thick] (-1,1) -- (-1,3);
     
        \draw[very thick] (1,-0.5) -- (1.5,-0.5);
         \draw[very thick] (1.5,-0.5) -- (1.5,-3);
    
     \draw[very thick] (1,0) -- (2,0);
     \draw[very thick] (2,0) -- (2,-3);
     
      \draw[very thick] (1,0.5) -- (2.5,0.5);
     \draw[very thick] (2.5,0.5) -- (2.5,-3);
     
       \draw[very thick] (-1,0) -- (-2,0);
     \draw[very thick] (-2,0) -- (-2,3);
     
     \draw[very thick] (-0.5,1) -- (-0.5,3);
     \draw[very thick] (-1,-1) -- (-1,-3);
     \draw[very thick] (-0.5,-1) -- (-0.5,-3);
     
      \draw[very thick] (0,1) -- (0,3);
     \draw[very thick] (0,-1) -- (0,-3);
     
     \draw[very thick] (0.5,1) -- (0.5,3);
     \draw[very thick] (0.5,-1) -- (0.5,-3);
     
     \draw[very thick] (1,1) -- (1,3);
     \draw[very thick] (1,-1) -- (1,-3);

      \node(a)[fill=blue,opacity=.1,text opacity=1] at (-1.25,+0.75) {1};
      \node(a)[fill=blue,opacity=.1,text opacity=1] at (-1.25,+1.25) {1};
      \node(a)[fill=blue,opacity=.1,text opacity=1] at (-1.25,+1.75) {1};
      \node(a)[fill=blue,opacity=.1,text opacity=1] at (-1.25,+2.25) {1};
      \node(a)[fill=blue,opacity=.1,text opacity=1] at (-1.25,+2.75) {1};
      
       \draw[very thick] (-1,0) -- (-2,0);
      \draw[very thick] (-2,0) -- (-2,3);
    
      \node(a)[fill=green,opacity=.1,text opacity=1] at (-1.25,+0.25) {2};
      \node(a)[fill=green,opacity=.1,text opacity=1] at (-1.75,+0.25) {2};
      \node(a)[fill=green,opacity=.1,text opacity=1] at (-1.75,+0.75) {2};
      \node(a)[fill=green,opacity=.1,text opacity=1] at (-1.75,+1.25) {2};
      \node(a)[fill=green,opacity=.1,text opacity=1] at (-1.75,+1.75) {2};
      \node(a) [fill=green,opacity=.1,text opacity=1] at (-1.75,+2.25) {2};
      \node(a)[fill=green,opacity=.1,text opacity=1] at (-1.75,+2.75) {2};
      
      \draw[very thick] (-1,-0.5) -- (-2.5,-0.5);
      \draw[very thick] (-2.5,-0.5) -- (-2.5,3);
      
      \node(a)[fill=red,opacity=.1,text opacity=1] at (-1.25,-0.25) {-1};
      \node(a)[fill=red,opacity=.1,text opacity=1] at (-1.75,-0.25) {-1};
      \node(a)[fill=red,opacity=.1,text opacity=1] at (-2.25,0.75) {-1};
       \node(a)[fill=red,opacity=.1,text opacity=1] at (-2.25,-0.25) {-1};
       \node(a)[fill=red,opacity=.1,text opacity=1] at (-2.25,0.25) {-1};
      \node(a)[fill=red,opacity=.1,text opacity=1] at (-2.25,1.25) {-1};
      \node(a)[fill=red,opacity=.1,text opacity=1] at (-2.25,1.75) {-1};
      \node(a)[fill=red,opacity=.1,text opacity=1] at (-2.25,2.25) {-1};
      \node(a)[fill=red,opacity=.1,text opacity=1] at (-2.25,+2.75) {-1};
      
       \node(a) at (-2.75,2.75) {1};
        \node(a) at (-2.75,2.25) {1};
         \node(a) at (-2.75,1.75) {1};
          \node(a) at (-2.75,1.25) {1};
           \node(a) at (-2.75,0.75) {1};
           
          \node(a) at (-2.75,0.25) {1}; 
          \node(a) at (-2.75,-0.25) {1}; 
          \node(a) at (-2.75,-0.75) {1};
           \node(a) at (-2.75,-1.25) {1};
              \node(a) at (-2.75,-1.75) {1};
               \node(a) at (-2.75,-2.25) {1};
                \node(a) at (-2.75,-2.75) {1};

     \node(a) at (-2.25,-0.75) {1};
     \node(a) at (-2.25,-1.25) {1};
     \node(a) at (-2.25,-1.75) {1};
     \node(a) at (-2.25,-2.25) {1};
     \node(a) at (-2.25,-2.75) {1};
     
     \node(a) at (-1.75,-0.75) {1};
     \node(a) at (-1.75,-1.25) {1};
     \node(a) at (-1.75,-1.75) {1};
     \node(a) at (-1.75,-2.25) {1};
     \node(a) at (-1.75,-2.75) {1};
     
     \node(a) at (-1.25,-0.75) {1};
     \node(a) at (-1.25,-1.25) {1};
     \node(a) at (-1.25,-1.75) {1};
     \node(a) at (-1.25,-2.25) {1};
     \node(a) at (-1.25,-2.75) {1};
     
     \node(a) at (-0.75,-0.75) {1};
     \node(a) at (-0.75,-1.25) {1};
     \node(a) at (-0.75,-1.75) {1};
     \node(a) at (-0.75,-2.25) {1};
     \node(a) at (-0.75,-2.75) {1};
     
     \node(a) at (-0.25,-0.75) {1};
     \node(a) at (-0.25,-1.25) {1};
     \node(a) at (-0.25,-1.75) {1};
     \node(a) at (-0.25,-2.25) {1};
     \node(a) at (-0.25,-2.75) {1};
     
     \node(a) at (0.25,-0.75) {-2};
     \node(a) at (0.25,-1.25) {-2};
     \node(a) at (0.25,-1.75) {-2};
     \node(a) at (0.25,-2.25) {-2};
     \node(a) at (0.25,-2.75) {-2};
     
     \node(a) at (0.75,-0.75) {-1};
     \node(a) at (0.75,-1.25) {-1};
     \node(a) at (0.75,-1.75) {-1};
     \node(a) at (0.75,-2.25) {-1};
     \node(a) at (0.75,-2.75) {-1};

     \node(a) at (-0.75,1.25) {1};
     \node(a) at (-0.75,1.75) {1};
     \node(a) at (-0.75,2.25) {1};
     \node(a) at (-0.75,2.75) {1};
     
     \node(a) at (-0.25,0.75) {2};
     \node(a) at (-0.25,1.25) {2};
     \node(a) at (-0.25,1.75) {2};
     \node(a) at (-0.25,2.25) {2};
     \node(a) at (-0.25,2.75) {2};
     
     \node(a) at (0.25,0.75) {-1};
     \node(a) at (0.25,1.25) {-1};
     \node(a) at (0.25,1.75) {-1};
     \node(a) at (0.25,2.25) {-1};
     \node(a) at (0.25,2.75) {-1};
     
     \node(a) at (0.75,0.75) {-1};
     \node(a) at (0.75,1.25) {-1};
     \node(a) at (0.75,1.75) {-1};
     \node(a) at (0.75,2.25) {-1};
     \node(a) at (0.75,2.75) {-1};
     
     \node(a) at (2.75,2.75) {-1};
     \node(a) at (2.75,2.25) {-1};
     \node(a) at (2.75,1.75) {-1};
     \node(a) at (2.75,1.25) {-1};
     \node(a) at (2.75,0.75) {-1};
     \node(a) at (2.75,0.25) {-1}; 
     \node(a) at (2.75,-0.25) {-1}; 
     \node(a) at (2.75,-0.75) {-1};
     \node(a) at (2.75,-1.25) {-1};
     \node(a) at (2.75,-1.75) {-1};
     \node(a) at (2.75,-2.25) {-1};
     \node(a) at (2.75,-2.75) {-1};
     
     \node(a) at (2.25,2.75) {-1};
     \node(a) at (2.25,2.25) {-1};
     \node(a) at (2.25,1.75) {-1};
     \node(a) at (2.25,1.25) {-1};
     \node(a) at (2.25,0.75) {-1};
     \node(a)[fill=blue,opacity=.1,text opacity=1] at (2.25,0.25) {-1}; 
     \node(a)[fill=blue,opacity=.1,text opacity=1] at (2.25,-0.25) {-1}; 
     \node(a)[fill=blue,opacity=.1,text opacity=1] at (2.25,-0.75) {-1};
     \node(a)[fill=blue,opacity=.1,text opacity=1] at (2.25,-1.25) {-1};
     \node(a)[fill=blue,opacity=.1,text opacity=1] at (2.25,-1.75) {-1};
     \node(a)[fill=blue,opacity=.1,text opacity=1] at (2.25,-2.25) {-1};
     \node(a)[fill=blue,opacity=.1,text opacity=1] at (2.25,-2.75) {-1};
     
     \node(a) at (1.75,2.75) {-1};
     \node(a) at (1.75,2.25) {-1};
     \node(a) at (1.75,1.75) {-1};
     \node(a) at (1.75,1.25) {-1};
     \node(a) at (1.75,0.75) {-1};
     \node(a)[fill=blue,opacity=.1,text opacity=1] at (1.75,0.25) {-1}; 
     \node(a)[fill=green,opacity=.1,text opacity=1] at (1.75,-0.25) {-2}; 
     \node(a)[fill=green,opacity=.1,text opacity=1] at (1.75,-0.75) {-2};
     \node(a)[fill=green,opacity=.1,text opacity=1] at (1.75,-1.25) {-2};
     \node(a)[fill=green,opacity=.1,text opacity=1] at (1.75,-1.75) {-2};
     \node(a)[fill=green,opacity=.1,text opacity=1] at (1.75,-2.25) {-2};
     \node(a)[fill=green,opacity=.1,text opacity=1] at (1.75,-2.75) {-2};
     
     \node(a) at (1.25,2.75) {-1};
     \node(a) at (1.25,2.25) {-1};
     \node(a) at (1.25,1.75) {-1};
     \node(a) at (1.25,1.25) {-1};
     \node(a) at (1.25,0.75) {-1};
     \node(a)[fill=blue,opacity=.1,text opacity=1] at (1.25,0.25) {-1}; 
     \node(a)[fill=green,opacity=.1,text opacity=1] at (1.25,-0.25) {-2}; 
     \node(a)[fill=red,opacity=.1,text opacity=1]  at (1.25,-0.75) {-1};
     \node(a)[fill=red,opacity=.1,text opacity=1]  at (1.25,-1.25) {-1};
     \node(a)[fill=red,opacity=.1,text opacity=1]  at (1.25,-1.75) {-1};
     \node(a)[fill=red,opacity=.1,text opacity=1]  at (1.25,-2.25) {-1};
     \node(a)[fill=red,opacity=.1,text opacity=1]  at (1.25,-2.75) {-1};
  \end{tikzpicture}

\caption{An instance of \twodmstucker, according to the construction of Lemma \ref{lem:tuckertomstucker}. Notice that two opposite sides of the \twodmstucker\ are monochromatically labelled with opposite labels. The original square of the \twodtucker\ instance is highlighted in the middle of the new square.}
\label{fig:mstucker}
 \end{figure} 

%% file: overview_general.tex
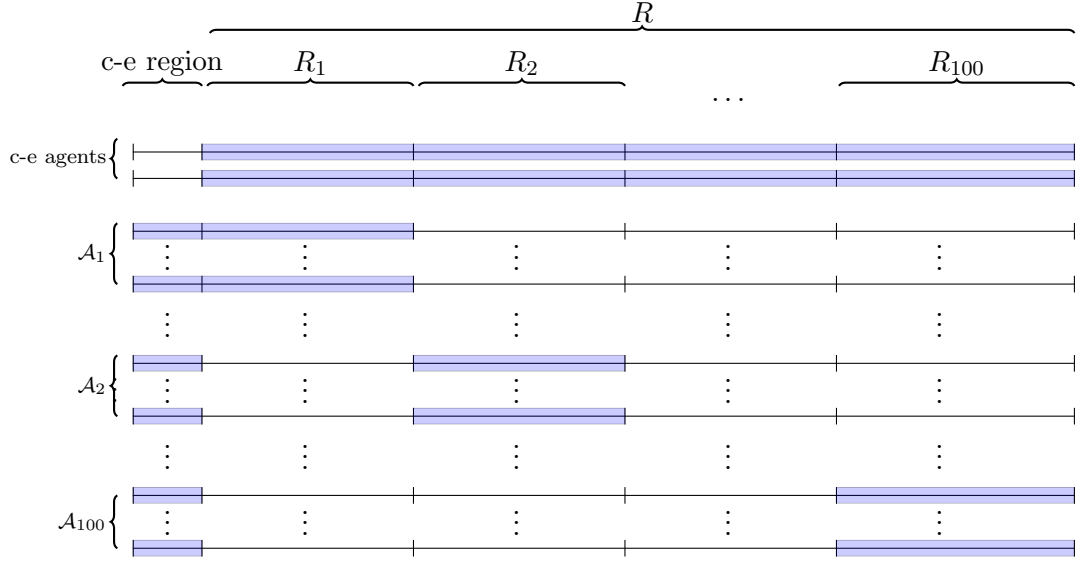
\begin{figure}

\begin{tikzpicture}[scale=1, transform shape]
	\node (a_1) at (0pt,0pt) {}; 
	\node (a_2) at (364pt, 0pt) {};
	\draw (a_1)--(a_2);
	\node (a_1) at (0pt,-10pt) {}; 
	\node (a_2) at (364pt, -10pt) {};
	\draw (a_1)--(a_2);
	\node (a_1) at (0pt,-30pt) {}; 
	\node (a_2) at (364pt, -30pt) {};
	\draw (a_1)--(a_2);
	\node (a_1) at (0pt,-50pt) {}; 
	\node (a_2) at (364pt, -50pt) {};
	\draw (a_1)--(a_2);
	\node (a_1) at (0pt,-80pt) {}; 
	\node (a_2) at (364pt, -80pt) {};
	\draw (a_1)--(a_2);
	\node (a_1) at (0pt,-100pt) {}; 
	\node (a_2) at (364pt, -100pt) {};
	\draw (a_1)--(a_2);
	
	\node (a_1) at (0pt,-130pt) {}; 
	\node (a_2) at (364pt, -130pt) {};
	\draw (a_1)--(a_2);
	\node (a_1) at (0pt,-150pt) {}; 
	\node (a_2) at (364pt, -150pt) {};
	\draw (a_1)--(a_2);

	\node (b_1) at (4pt, 7pt) {};
	\node (b_2) at (4pt, -7pt) {};	
	\draw(b_1) -- (b_2);
	
	\node (b_1) at (30pt, 7pt) {};
	\node (b_2) at (30pt, -7pt) {};	
	\draw(b_1) -- (b_2);
	
		\node (b_1) at (110pt, 7pt) {};
	\node (b_2) at (110pt, -7pt) {};	
	\draw(b_1) -- (b_2);
	
	\node (b_1) at (190pt, 7pt) {};
	\node (b_2) at (190pt, -7pt) {};	
	\draw(b_1) -- (b_2);
	
	\node (b_1) at (270pt, 7pt) {};
	\node (b_2) at (270pt, -7pt) {};	
	\draw(b_1) -- (b_2);
	
	\node (b_1) at (360pt, 7pt) {};
	\node (b_2) at (360pt, -7pt) {};	
	\draw(b_1) -- (b_2);
	
		\node (b_1) at (4pt, -3pt) {};
	\node (b_2) at (4pt, -17pt) {};	
	\draw(b_1) -- (b_2);
	
	\node (b_1) at (30pt, -3pt) {};
	\node (b_2) at (30pt, -17pt) {};	
	\draw(b_1) -- (b_2);
	
		\node (b_1) at (110pt, -3pt) {};
	\node (b_2) at (110pt, -17pt) {};	
	\draw(b_1) -- (b_2);
	
	\node (b_1) at (190pt, -3pt) {};
	\node (b_2) at (190pt, -17pt) {};	
	\draw(b_1) -- (b_2);
	
	\node (b_1) at (270pt, -3pt) {};
	\node (b_2) at (270pt, -17pt) {};	
	\draw(b_1) -- (b_2);
	
	\node (b_1) at (360pt, -3pt) {};
	\node (b_2) at (360pt, -17pt) {};	
	\draw(b_1) -- (b_2);
	
	\node (b_1) at (4pt, -23pt) {};
	\node (b_2) at (4pt, -37pt) {};	
	\draw(b_1) -- (b_2);
	
		\node (b_1) at (30pt, -23pt) {};
	\node (b_2) at (30pt, -37pt) {};	
	\draw(b_1) -- (b_2);
	
		\node (b_1) at (110pt, -23pt) {};
	\node (b_2) at (110pt, -37pt) {};	
	\draw(b_1) -- (b_2);
	
	\node (b_1) at (190pt, -23pt) {};
	\node (b_2) at (190pt, -37pt) {};	
	\draw(b_1) -- (b_2);
	
	\node (b_1) at (270pt, -23pt) {};
	\node (b_2) at (270pt, -37pt) {};	
	\draw(b_1) -- (b_2);
	
	\node (b_1) at (360pt, -23pt) {};
	\node (b_2) at (360pt, -37pt) {};	
	\draw(b_1) -- (b_2);
	
		\node (b_1) at (4pt, -43pt) {};
	\node (b_2) at (4pt, -57pt) {};	
	\draw(b_1) -- (b_2);
	
	\node (b_1) at (30pt, -43pt) {};
	\node (b_2) at (30pt, -57pt) {};	
	\draw(b_1) -- (b_2);
	
		\node (b_1) at (110pt, -43pt) {};
	\node (b_2) at (110pt, -57pt) {};	
	\draw(b_1) -- (b_2);
	
	\node (b_1) at (190pt, -43pt) {};
	\node (b_2) at (190pt, -57pt) {};	
	\draw(b_1) -- (b_2);
	
	\node (b_1) at (270pt, -43pt) {};
	\node (b_2) at (270pt, -57pt) {};	
	\draw(b_1) -- (b_2);
	
	\node (b_1) at (360pt, -43pt) {};
	\node (b_2) at (360pt, -57pt) {};	
	\draw(b_1) -- (b_2);
	
	\node (a) at (17pt, -37pt){$\vdots$};
	\node (a) at (69pt, -37pt){$\vdots$};
	\node (a) at (149pt, -37pt){$\vdots$};
	\node (a) at (229pt, -37pt){$\vdots$};
	\node (a) at (309pt, -37pt){$\vdots$};

		\node (b_1) at (4pt, -143pt) {};
	\node (b_2) at (4pt, -157pt) {};	
	\draw(b_1) -- (b_2);
	
	\node (b_1) at (30pt, -143pt) {};
	\node (b_2) at (30pt, -157pt) {};	
	\draw(b_1) -- (b_2);
	
	\node (b_1) at (110pt, -143pt) {};
	\node (b_2) at (110pt, -157pt) {};	
	\draw(b_1) -- (b_2);
	
	\node (b_1) at (190pt, -143pt) {};
	\node (b_2) at (190pt, -157pt) {};	
	\draw(b_1) -- (b_2);
	
	\node (b_1) at (270pt, -143pt) {};
	\node (b_2) at (270pt, -157pt) {};	
	\draw(b_1) -- (b_2);
	
	\node (b_1) at (360pt, -143pt) {};
	\node (b_2) at (360pt, -157pt) {};	
	\draw(b_1) -- (b_2);
	
	\node (a) at (17pt, -62.5pt){$\vdots$};
	\node (a) at (69pt, -62.5pt){$\vdots$};
	\node (a) at (149pt, -62.5pt){$\vdots$};
	\node (a) at (229pt, -62.5pt){$\vdots$};
	\node (a) at (309pt, -62.5pt){$\vdots$};
	
		\node (b_1) at (4pt, -73pt) {};
	\node (b_2) at (4pt, -87pt) {};	
	\draw(b_1) -- (b_2);
	
	\node (b_1) at (30pt, -73pt) {};
	\node (b_2) at (30pt, -87pt) {};	
	\draw(b_1) -- (b_2);
	
	\node (b_1) at (110pt, -73pt) {};
	\node (b_2) at (110pt, -87pt) {};	
	\draw(b_1) -- (b_2);
	
	\node (b_1) at (190pt, -73pt) {};
	\node (b_2) at (190pt, -87pt) {};	
	\draw(b_1) -- (b_2);
	
	\node (b_1) at (270pt, -73pt) {};
	\node (b_2) at (270pt, -87pt) {};	
	\draw(b_1) -- (b_2);
	
	\node (b_1) at (360pt, -73pt) {};
	\node (b_2) at (360pt, -87pt) {};	
	\draw(b_1) -- (b_2);
	
	\node (b_1) at (4pt, -93pt) {};
	\node (b_2) at (4pt, -107pt) {};	
	\draw(b_1) -- (b_2);
	
	\node (b_1) at (30pt, -93pt) {};
	\node (b_2) at (30pt, -107pt) {};	
	\draw(b_1) -- (b_2);
	
	\node (b_1) at (110pt, -93pt) {};
	\node (b_2) at (110pt, -107pt) {};	
	\draw(b_1) -- (b_2);
	
	\node (b_1) at (190pt, -93pt) {};
	\node (b_2) at (190pt, -107pt) {};	
	\draw(b_1) -- (b_2);
	
	\node (b_1) at (270pt, -93pt) {};
	\node (b_2) at (270pt, -107pt) {};	
	\draw(b_1) -- (b_2);
	
	\node (b_1) at (360pt, -93pt) {};
	\node (b_2) at (360pt, -107pt) {};	
	\draw(b_1) -- (b_2);

\node (a) at (17pt, -112.5pt){$\vdots$};
\node (a) at (69pt, -112.5pt){$\vdots$};
\node (a) at (149pt, -112.5pt){$\vdots$};
\node (a) at (229pt, -112.5pt){$\vdots$};
\node (a) at (309pt, -112.5pt){$\vdots$};

	\node (b_1) at (4pt, -123pt) {};
\node (b_2) at (4pt, -137pt) {};	
\draw(b_1) -- (b_2);

\node (b_1) at (30pt, -123pt) {};
\node (b_2) at (30pt, -137pt) {};	
\draw(b_1) -- (b_2);

\node (b_1) at (110pt, -123pt) {};
\node (b_2) at (110pt, -137pt) {};	
\draw(b_1) -- (b_2);

\node (b_1) at (190pt, -123pt) {};
\node (b_2) at (190pt, -137pt) {};	
\draw(b_1) -- (b_2);

\node (b_1) at (270pt, -123pt) {};
\node (b_2) at (270pt, -137pt) {};	
\draw(b_1) -- (b_2);

\node (b_1) at (360pt, -123pt) {};
\node (b_2) at (360pt, -137pt) {};	
\draw(b_1) -- (b_2);

\node (a) at (17pt, -137.5pt){$\vdots$};
\node (a) at (69pt, -137.5pt){$\vdots$};
\node (a) at (149pt, -137.5pt){$\vdots$};
\node (a) at (229pt, -137.5pt){$\vdots$};
\node (a) at (309pt, -137.5pt){$\vdots$};

\node (a) at (17pt, -87.5pt){$\vdots$};
\node (a) at (69pt, -87.5pt){$\vdots$};
\node (a) at (149pt, -87.5pt){$\vdots$};
\node (a) at (229pt, -87.5pt){$\vdots$};
\node (a) at (309pt, -87.5pt){$\vdots$};

\draw[fill=blue,opacity=0.2] (30pt, 3pt) rectangle (360pt,-3pt);
\draw[fill=blue,opacity=0.2] (30pt, -7pt) rectangle (360pt,-13pt);

\draw[fill=blue,opacity=0.2] (4pt, -33pt) rectangle (110pt, -27pt);

\draw[fill=blue,opacity=0.2] (4pt, -53pt) rectangle (110pt, -47pt);

\draw[fill=blue,opacity=0.2] (4pt, -83pt) rectangle (30pt, -77pt);
\draw[fill=blue,opacity=0.2] (110pt, -83pt) rectangle (190pt, -77pt);

\draw[fill=blue,opacity=0.2] (4pt, -103pt) rectangle (30pt, -97pt);
\draw[fill=blue,opacity=0.2] (110pt, -103pt) rectangle (190pt, -97pt);

\draw[fill=blue,opacity=0.2] (4pt, -133pt) rectangle (30pt, -127pt);
\draw[fill=blue,opacity=0.2] (270pt, -133pt) rectangle (360pt, -127pt);

\draw[fill=blue,opacity=0.2] (4pt, -153pt) rectangle (30pt, -147pt);
\draw[fill=blue,opacity=0.2] (270pt, -153pt) rectangle (360pt, -147pt);

	\node (a) at (-3pt, -87pt){$\vdots$};

%
%
%
%
\node (b) at (230pt, 20pt) {\ldots};
%
%

	\draw [
	thick,
	decoration={
		brace,
		raise=5pt
	},
	decorate
	] (0pt,20pt) -- (30pt,20pt)
	node [pos=0.5,anchor=south,yshift=5pt] {c-e region};
	
		\draw [
	thick,
	decoration={
		brace,
		raise=5pt
	},
	decorate
	] (32pt,20pt) -- (110pt,20pt)
	node [pos=0.5,anchor=south,yshift=5pt] {$R_1$};
	
		\draw [
	thick,
	decoration={
		brace,
		raise=5pt
	},
	decorate
	] (112pt,20pt) -- (190pt,20pt)
	node [pos=0.5,anchor=south,yshift=5pt] {$R_2$};
	
		\draw [
	thick,
	decoration={
		brace,
		raise=5pt
	},
	decorate
	] (271pt,20pt) -- (360pt,20pt)
	node [pos=0.5,anchor=south,yshift=5pt] {$R_{100}$};
	
			\draw [
	thick,
	decoration={
		brace,
		raise=5pt
	},
	decorate
	] (33pt,40pt) -- (360pt,40pt)
	node [pos=0.5,anchor=south,yshift=5pt] {$R$};


%
%
\draw [
thick,
decoration={
	brace,
	mirror,
	raise=5pt
},
decorate
] (3pt,5pt) -- (3pt,-10pt)
node [pos=0.5,anchor=east,xshift=-5pt] {\scriptsize{c-e agents}};

%
%
%
%
%
\draw [
thick,
decoration={
	brace,
	mirror,
	raise=5pt
},
decorate
] (3pt,-27pt) -- (3pt,-50pt)
node [pos=0.5,anchor=east,xshift=-5pt] {\scriptsize{$\mathcal{A}_1$}};

\draw [
thick,
decoration={
	brace,
	mirror,
	raise=5pt
},
decorate
] (3pt,-77pt) -- (3pt,-100pt)
node [pos=0.5,anchor=east,xshift=-5pt] {\scriptsize{$\mathcal{A}_2$}};

\draw [
thick,
decoration={
	brace,
	mirror,
	raise=5pt
},
decorate
] (3pt,-130pt) -- (3pt,-150pt)
node [pos=0.5,anchor=east,xshift=-5pt] {\scriptsize{$\mathcal{A}_{100}$}};
	\end{tikzpicture}
	\caption{\small{An overview of $I_{CH}$, denoting all the different regions and the agents of $C_1 \ldots, C_n$, as well as the coordinate-encoding agents. The highlighted areas denote that the corresponding agent has non-zero value on these regions.}}
	\label{fig:overview_general}
	\end{figure}

%% file: overview.tex
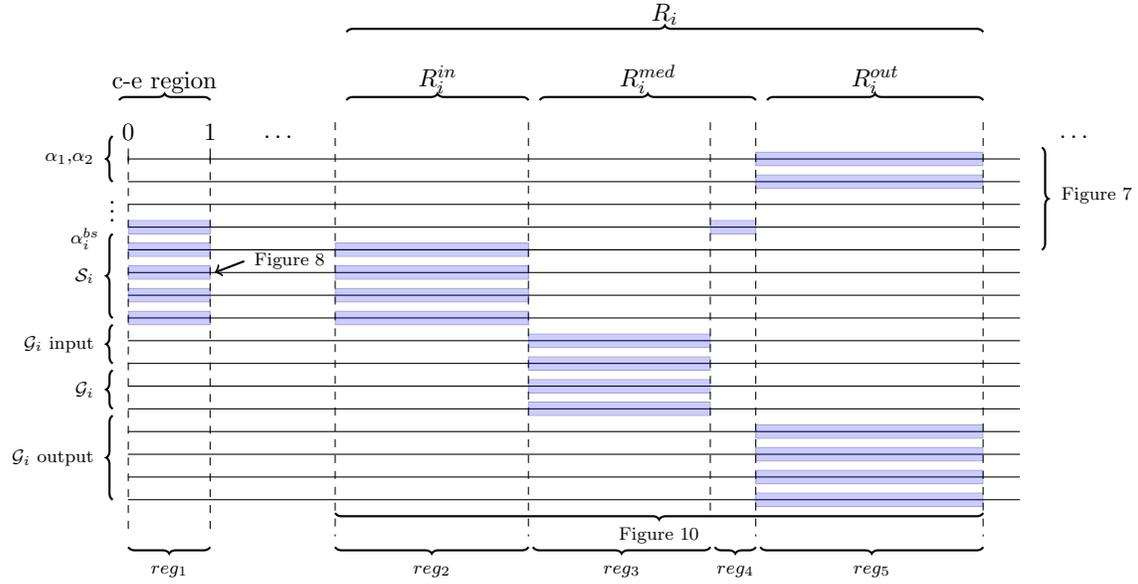
\begin{figure}
\centering

\begin{tikzpicture}[scale=0.86, transform shape]
	\node (a_1) at (0pt,0pt) {}; 
	\node (a_2) at (400pt, 0pt) {};
	\draw (a_1)--(a_2);
	\node (a_1) at (0pt,-10pt) {}; 
	\node (a_2) at (400pt, -10pt) {};
	\draw (a_1)--(a_2);
	\node (a_1) at (0pt,-30pt) {}; 
	\node (a_2) at (400pt, -30pt) {};
	\draw (a_1)--(a_2);
	\node (a_1) at (0pt,-40pt) {}; 
	\node (a_2) at (400pt, -40pt) {};
	\draw (a_1)--(a_2);
	\node (a_1) at (0pt,-50pt) {}; 
	\node (a_2) at (400pt, -50pt) {};
	\draw (a_1)--(a_2);
	\node (a_1) at (0pt,-60pt) {}; 
	\node (a_2) at (400pt, -60pt) {};
	\draw (a_1)--(a_2);
	\node (a_1) at (0pt,-70pt) {}; 
	\node (a_2) at (400pt, -70pt) {};
	\draw (a_1)--(a_2);
	\node (a_1) at (0pt,-80pt) {}; 
	\node (a_2) at (400pt, -80pt) {};
	\draw (a_1)--(a_2);
	\node (a_1) at (0pt,-90pt) {}; 
	\node (a_2) at (400pt, -90pt) {};
	\draw (a_1)--(a_2);
	\node (a_1) at (0pt,-100pt) {}; 
	\node (a_2) at (400pt, -100pt) {};
	\draw (a_1)--(a_2);
	\node (a_1) at (0pt,-110pt) {}; 
	\node (a_2) at (400pt, -110pt) {};
	\draw (a_1)--(a_2);
	\node (a_1) at (0pt,-120pt) {}; 
	\node (a_2) at (400pt, -120pt) {};
	\draw (a_1)--(a_2);
	\node (a_1) at (0pt,-130pt) {}; 
	\node (a_2) at (400pt, -130pt) {};
	\draw (a_1)--(a_2);
	\node (a_1) at (0pt,-140pt) {}; 
	\node (a_2) at (400pt, -140pt) {};
	\draw (a_1)--(a_2);
	\node (a_1) at (0pt,-150pt) {}; 
	\node (a_2) at (400pt, -150pt) {};
	\draw (a_1)--(a_2);

	\node (b_1) at (4pt, 12pt) {$0$};
	\node (b_2) at (4pt, -5pt) {};	
	\draw(b_1) -- (b_2);
	
	\node (b_1) at (40pt, 12pt) {$1$};
	\node (b_2) at (40pt, -5pt) {};	
	\draw(b_1) -- (b_2);
	
	\node (a_1) at (0pt,-20pt) {}; 
	\node (a_2) at (400pt, -20pt) {};
	\draw (a_1)--(a_2);

	\node (b_1) at (40pt, 5pt) {};
\node (b_2) at (40pt, -170pt) {};	
\draw[dashed] (b_1) -- (b_2);

	\node (b_1) at (4pt, 5pt) {};
\node (b_2) at (4pt, -170pt) {};	
\draw[dashed] (b_1) -- (b_2);

	\node (b_1) at (95pt, 20pt) {};
\node (b_2) at (95pt, -170pt) {};	
\draw[dashed] (b_1) -- (b_2);

	\node (b_1) at (280pt, 20pt) {};
\node (b_2) at (280pt, -170pt) {};	
\draw[dashed] (b_1) -- (b_2);

	\node (b_1) at (180pt, 20pt) {};
\node (b_2) at (180pt, -170pt) {};	
\draw[dashed] (b_1) -- (b_2);

	\node (b_1) at (260pt, 20pt) {};
\node (b_2) at (260pt, -160pt) {};	
\draw[dashed] (b_1) -- (b_2);
	
\node (b_1) at (380pt, 20pt) {};
\node (b_2) at (380pt, -170pt) {};	
\draw[dashed] (b_1) -- (b_2);	
	
		\node (a_1) at (0pt,-40pt) {}; 
	\node (a_2) at (400pt, -40pt) {};
	\draw (a_1)--(a_2);

	\draw [
	thick,
	decoration={
		brace,
		raise=5pt
	},
	decorate
	] (0pt,20pt) -- (40pt,20pt)
	node [pos=0.5,anchor=south,yshift=5pt] {c-e region};


\node (a) at (70pt,10pt) {$\ldots$};
\node (b) at (420pt,10pt) {$\ldots$};
	\draw [
thick,
decoration={
	brace,
	raise=5pt
},
decorate
] (100pt,20pt) -- (180pt,20pt)
node [pos=0.5,anchor=south,yshift=5pt] {$R_i^{in}$};

	\draw [
thick,
decoration={
	brace,
	raise=5pt
},
decorate
] (185pt,20pt) -- (280pt,20pt)
node [pos=0.5,anchor=south,yshift=5pt] {$R_i^{med}$};

	\draw [
thick,
decoration={
	brace,
	raise=5pt
},
decorate
] (285pt,20pt) -- (380pt,20pt)
node [pos=0.5,anchor=south,yshift=5pt] {$R_i^{out}$};

	\draw [
thick,
decoration={
	brace,
	raise=5pt
},
decorate
] (100pt,50pt) -- (380pt,50pt)
node [pos=0.5,anchor=south,yshift=5pt] {$R_i$};

\node (a) at (-3pt, -20pt){$\vdots$};

\draw [
thick,
decoration={
	brace,
	mirror,
	raise=5pt
},
decorate
] (3pt,-33pt) -- (3pt,-70pt)
node [pos=0.5,anchor=east,xshift=-10pt] {\scriptsize{$\mathcal{S}_i$}};

\draw [
thick,
decoration={
	brace,
	mirror,
	raise=5pt
},
decorate
] (3pt,10pt) -- (3pt,-10pt)
node [pos=0.5,anchor=east,xshift=-10pt] {\scriptsize{$\alpha_1$,$\alpha_2$}};

\node (a) at (-15pt,-35pt){\scriptsize{$\alpha_{i}^{bs}$}};

\draw [
thick,
decoration={
	brace,
	mirror,
	raise=5pt
},
decorate
] (3pt,-73pt) -- (3pt,-90pt)
node [pos=0.5,anchor=east,xshift=-10pt] {\scriptsize{$\mathcal{G}_i$ input}};

\draw [
thick,
decoration={
	brace,
	mirror,
	raise=5pt
},
decorate
] (3pt,-93pt) -- (3pt,-110pt)
node [pos=0.5,anchor=east,xshift=-10pt] {\scriptsize{$\mathcal{G}_i$}};

\draw [
thick,
decoration={
	brace,
	mirror,
	raise=5pt
},
decorate
] (3pt,-113pt) -- (3pt,-150pt)
node [pos=0.5,anchor=east,xshift=-10pt] {\scriptsize{$\mathcal{G}_i$ output}};

\draw[thick,->](55pt,-45pt)--(42pt,-50pt);
\node at (75pt,-45pt) {\scriptsize{Figure \ref{fig:sensor_agents}}};


\draw [
thick,
decoration={
	brace,
	mirror,
	raise=5pt
},
decorate
] (95pt,-150pt) -- (380pt,-150pt)
node [pos=0.5,anchor=north,yshift=-8pt] {\scriptsize{Figure \ref{fig:gates}}};

\draw [
thick,
decoration={
	brace,
	raise=5pt
},
decorate
] (400pt,5pt) -- (400pt,-40pt)
node [pos=0.5,anchor=south,xshift=30pt, yshift=-6pt] {\scriptsize{Figure \ref{fig:ce-agents}}};

	\draw [
thick,
decoration={
	brace,
	mirror,
	raise=5pt
},
decorate
] (4pt,-165pt) -- (40pt,-165pt)
node [pos=0.5,anchor=north,yshift=-10pt] {\scriptsize{$reg_1$}};

	\draw [
thick,
decoration={
	brace,
	mirror,
	raise=5pt
},
decorate
] (95pt,-165pt) -- (180pt,-165pt)
node [pos=0.5,anchor=north,yshift=-10pt] {\scriptsize{$reg_2$}};

	\draw [
thick,
decoration={
	brace,
	mirror,
	raise=5pt
},
decorate
] (182pt,-165pt) -- (260pt,-165pt)
node [pos=0.5,anchor=north,yshift=-10pt] {\scriptsize{$reg_3$}};

	\draw [
thick,
decoration={
	brace,
	mirror,
	raise=5pt
},
decorate
] (262pt,-165pt) -- (280pt,-165pt)
node [pos=0.5,anchor=north,yshift=-10pt] {\scriptsize{$reg_4$}};

	\draw [
thick,
decoration={
	brace,
	mirror,
	raise=5pt
},
decorate
] (282pt,-165pt) -- (380pt,-165pt)
node [pos=0.5,anchor=north,yshift=-10pt] {\scriptsize{$reg_5$}};



\draw[fill=blue,opacity=0.2] (280pt, 3pt) rectangle (380pt,-3pt);
\draw[fill=blue,opacity=0.2] (280pt, -7pt) rectangle (380pt,-13pt);

\draw[fill=blue,opacity=0.2] (4pt, -33pt) rectangle (40pt,-27pt);
\draw[fill=blue,opacity=0.2] (260pt, -33pt) rectangle (280pt,-27pt);

\draw[fill=blue,opacity=0.2] (4pt, -43pt) rectangle (40pt,-37pt);
\draw[fill=blue,opacity=0.2] (4pt, -53pt) rectangle (40pt,-47pt);
\draw[fill=blue,opacity=0.2] (4pt, -63pt) rectangle (40pt,-57pt);
\draw[fill=blue,opacity=0.2] (4pt, -73pt) rectangle (40pt,-67pt);

\draw[fill=blue,opacity=0.2] (95pt, -43pt) rectangle (180pt,-37pt);
\draw[fill=blue,opacity=0.2] (95pt, -53pt) rectangle (180pt,-47pt);
\draw[fill=blue,opacity=0.2] (95pt, -63pt) rectangle (180pt,-57pt);
\draw[fill=blue,opacity=0.2] (95pt, -73pt) rectangle (180pt,-67pt);

\draw[fill=blue,opacity=0.2] (180pt, -83pt) rectangle (260pt,-77pt);
\draw[fill=blue,opacity=0.2] (180pt, -93pt) rectangle (260pt,-87pt);
\draw[fill=blue,opacity=0.2] (180pt, -103pt) rectangle (260pt,-97pt);
\draw[fill=blue,opacity=0.2] (180pt, -113pt) rectangle (260pt,-107pt);

\draw[fill=blue,opacity=0.2] (280pt, -123pt) rectangle (380pt,-117pt);
\draw[fill=blue,opacity=0.2] (280pt, -133pt) rectangle (380pt,-127pt);
\draw[fill=blue,opacity=0.2] (280pt, -143pt) rectangle (380pt,-137pt);
\draw[fill=blue,opacity=0.2] (280pt, -153pt) rectangle (380pt,-147pt);

	\end{tikzpicture}
	\caption{\small{An overview of the construction of the agents that encode a single copy $C_i$ of the circuit-encoder. The highlighted intervals indicate that the corresponding agents have non-zero value in those regions. The precise valuations are not shown here, but the corresponding figures where they are presented in more detail are noted. In region \emph{$reg_1$}, the agents that have positive valuations are the coordinate-encoding agents, the blanket sensor agent and the bit-detectors of $C_i$.  Region \emph{$reg_2$}} contains the machinery of the bit-extracting sensor agents that extracts the position of the cuts (see Figure \ref{fig:sensor_agents}) as well as the valuation of the input gate-agents, corresponding to the input gates of $C_i$ (see Figure \ref{fig:gates}). Region \emph{$reg_3$} contains the values of all gate-agents of $C$, including the outputs of the input gate-agents and the inputs of the output gate agents. Region \emph{$reg_4$} contains the bit-detection gadget of the blanket sensor agent, which provides feedback to the coordinate encoding agents (see Figure {\ref{fig:ce-agents}). Finally, Region \emph{$reg_5$} contains the value of the output gate-agents, i.e. the encodings of the outputs to the circuit which are fed back to the coordinate encoding agents who also have value in this interval.}}
	\label{fig:overview}
	\end{figure}

%% file: boolean_gate_gadgets.tex
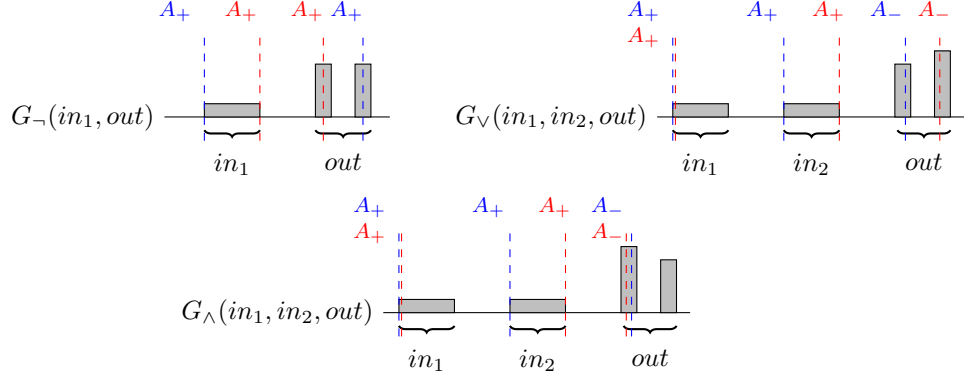
\begin{figure}
	\centering
	\begin{minipage}{0.35\textwidth}
\begin{tikzpicture}[scale=1,transform shape]
	\node (a_1) at (-10pt,0pt) {\small{$G_{\neg}(in_1,out)$}}; 
	\node (a_2) at (110pt, 0pt) {};
	\draw (a_1)--(a_2);
	\draw[fill=lightgray] (36pt,0pt) rectangle (57pt,5pt);
	\draw[fill=lightgray] (78pt,0pt) rectangle (84pt,20pt);
	\draw[fill=lightgray] (93pt,0pt) rectangle (99pt,20pt);
	
	\draw [
    thick,
    decoration={
        brace,
        mirror,
        raise=5pt
    },
    decorate
] (36pt,0pt) -- (57pt,0pt)
node [pos=0.5,anchor=north,yshift=-10pt] {\small{${in_1}$}}; 

\draw [
    thick,
    decoration={
        brace,
        mirror,
        raise=5pt
    },
    decorate
] (78pt,0pt) -- (99pt,0pt)
node [pos=0.5,anchor=north,yshift=-10pt] {\small{$out$}};

\draw[dashed,color=blue] (36pt,30pt) -- (36pt, -10pt);
\draw[dashed,color=blue] (96pt,30pt) -- (96pt, -10pt);

\draw[dashed,color=red] (57pt,30pt) -- (57pt, -10pt);
\draw[dashed,color=red] (81pt,30pt) -- (81pt, -10pt);

\node[color=blue] at (25pt,40pt) {\scriptsize{$\lplus$}};
\node[color=red] at (50pt,40pt) {\scriptsize{$\lplus$}};

\node[color=blue] at (90pt,40pt) {\scriptsize{$\lplus$}};
\node[color=red] at (75pt,40pt) {\scriptsize{$\lplus$}};

	\end{tikzpicture}
\end{minipage}
	\begin{minipage}{0.35\textwidth}
	\begin{tikzpicture}[scale=1,transform shape]
	\node (a_1) at (-10pt,0pt) {\small{$G_{\lor}(in_1,in_2,out)$}}; 
	\node (a_2) at (150pt, 0pt) {};
	\draw (a_1)--(a_2);
	\draw[fill=lightgray] (36pt,0pt) rectangle (57pt,5pt);
	\draw[fill=lightgray] (78pt,0pt) rectangle (99pt,5pt);
	\draw[fill=lightgray] (120pt,0pt) rectangle (126pt,20pt);
	\draw[fill=lightgray] (135pt,0pt) rectangle (141pt,25pt);
	
	\draw [
	thick,
	decoration={
		brace,
		mirror,
		raise=5pt
	},
	decorate
	] (36pt,0pt) -- (57pt,0pt)
	node [pos=0.5,anchor=north,yshift=-10pt] {\small{${in_1}$}}; 
	
	\draw [
	thick,
	decoration={
		brace,
		mirror,
		raise=5pt
	},
	decorate
	] (78pt,0pt) -- (99pt,0pt)
	node [pos=0.5,anchor=north,yshift=-10pt] {\small{$in_2$}};

	\draw [
	thick,
	decoration={
		brace,
		mirror,
		raise=5pt
	},
	decorate
	] (121pt,0pt) -- (141pt,0pt)
	node [pos=0.5,anchor=north,yshift=-10pt] {\small{$out$}};

	\draw[dashed,color=blue] (36pt,30pt) -- (36pt, -10pt);
	\draw[dashed,color=red] (37pt,30pt) -- (37pt, -10pt);
	\draw[dashed,color=red] (99pt,30pt) -- (99pt, -10pt);
	\draw[dashed,color=blue] (78pt,30pt) -- (78pt, -10pt);

	\draw[dashed,color=blue] (124pt,30pt) -- (124pt, -10pt);
	\draw[dashed,color=red] (137pt,30pt) -- (137pt, -10pt);
	
	\node[color=blue] at (25pt,40pt) {\scriptsize{$\lplus$}};
	\node[color=red] at (25pt,30pt) {\scriptsize{$\lplus$}};
	
	\node[color=red] at (95pt,40pt) {\scriptsize{$\lplus$}};
	\node[color=blue] at (70pt,40pt) {\scriptsize{$\lplus$}};
	
	\node[color=red] at (135pt,40pt) {\scriptsize{$\lminus$}};
	\node[color=blue] at (117pt,40pt) {\scriptsize{$\lminus$}};
	
	\end{tikzpicture}
	\end{minipage}
%
	\begin{tikzpicture}[scale=1,transform shape]
	\node at (10pt,40pt) {};
	\node (a_1) at (-10pt,0pt) {\small{$G_{\land}(in_1,in_2,out)$}}; 
	\node (a_2) at (150pt, 0pt) {};
	\draw (a_1)--(a_2);
	\draw[fill=lightgray] (36pt,0pt) rectangle (57pt,5pt);
	\draw[fill=lightgray] (78pt,0pt) rectangle (99pt,5pt);
	\draw[fill=lightgray] (120pt,0pt) rectangle (126pt,25pt);
	\draw[fill=lightgray] (135pt,0pt) rectangle (141pt,20pt);
	
	\draw [
	thick,
	decoration={
		brace,
		mirror,
		raise=5pt
	},
	decorate
	] (36pt,0pt) -- (57pt,0pt)
	node [pos=0.5,anchor=north,yshift=-10pt] {\small{${in_1}$}}; 
	
	\draw [
	thick,
	decoration={
		brace,
		mirror,
		raise=5pt
	},
	decorate
	] (78pt,0pt) -- (99pt,0pt)
	node [pos=0.5,anchor=north,yshift=-10pt] {\small{$in_2$}};
	
	\draw [
	thick,
	decoration={
		brace,
		mirror,
		raise=5pt
	},
	decorate
	] (121pt,0pt) -- (141pt,0pt)
	node [pos=0.5,anchor=north,yshift=-10pt] {\small{$out$}};
	
		\draw[dashed,color=blue] (36pt,30pt) -- (36pt, -10pt);
	\draw[dashed,color=red] (37pt,30pt) -- (37pt, -10pt);
	\draw[dashed,color=red] (99pt,30pt) -- (99pt, -10pt);
	\draw[dashed,color=blue] (78pt,30pt) -- (78pt, -10pt);

	\draw[dashed,color=blue] (124pt,30pt) -- (124pt, -10pt);
	\draw[dashed,color=red] (122pt,30pt) -- (122pt, -10pt);
	
	\node[color=blue] at (25pt,40pt) {\scriptsize{$\lplus$}};
	\node[color=red] at (25pt,30pt) {\scriptsize{$\lplus$}};
	
	\node[color=red] at (95pt,40pt) {\scriptsize{$\lplus$}};
	\node[color=blue] at (70pt,40pt) {\scriptsize{$\lplus$}};
	
	\node[color=red] at (115pt,30pt) {\scriptsize{$\lminus$}};
	\node[color=blue] at (115pt,40pt) {\scriptsize{$\lminus$}};
\end{tikzpicture}
	\caption{\small{The Boolean Gate Gadgets encoding the NOT, OR and AND gates. For visibility, the valuation blocks are not according to scale. For the NOT gate, the input has value $0.25$ and the output blocks have volume $0.375$ each. For the OR (respectively AND) gate, the input blocks have value $0.125$ each and the output blocks have value $0.3125$ and $0.4375$ (respectively $0.4375$ and $0.3125$). The cuts corresponding to pairs or triples of inputs and outputs have the same colour, and the labels on the left-side of these cuts are shown and colour-coded in the same way. For the NOT gate, when the input cut sits of the left (the blue cut), then the output cut must sit on the right (the blue cut), to compensate for the excess if $\lminus$ and oppositely for when the input cut sits of the right (the red cut). For the OR and AND gates, again the cuts corresponding to two inputs and one output have the same colour. For the OR gate, when both inputs cut sit on the left (the blue cuts), the output cut sits on the left as well, to compensate for the excess of $\lminus$ (notice that the left-hand side of the output cut is labelled $\lminus$. When one input sits on the left and the other one on the right, the inputs detect no discrepancy in the balance of labels and the output jumps to the right, because the output blocks are uneven (the red cuts). The operation of the AND gate is very similar; here the cases shown are those where the inputs are $0$ and $0$ and the ouput is $0$ (the blue cuts) and where the inputs are $0$ and $1$ and the output is still $0$ (the red cuts).}}
		\label{fig:boolean-gate-gadgets}
	\end{figure}

%% file: ce_agents.tex
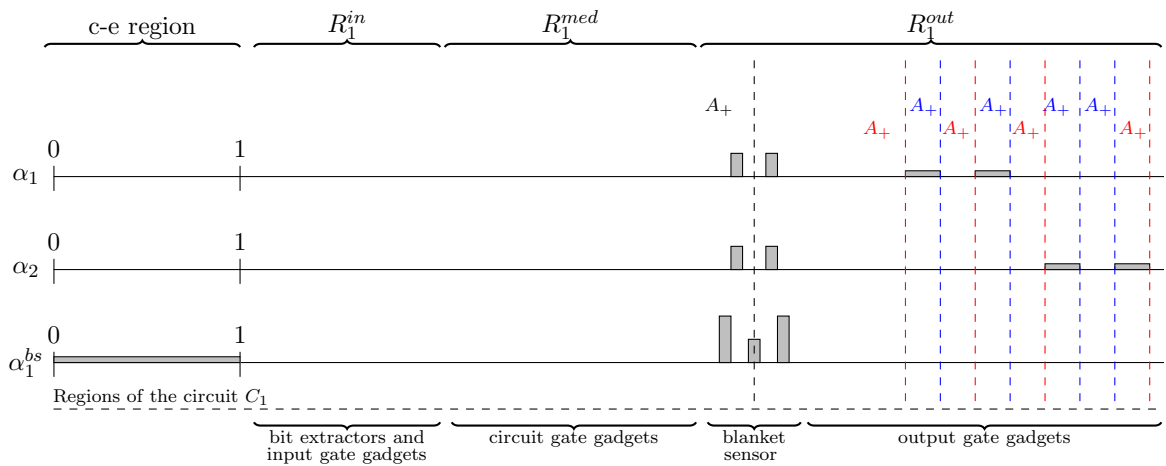
\begin{figure}

\begin{tikzpicture}[scale=0.88, transform shape]
	\node (a_1) at (0pt,0pt) {}; 
	\node (a_2) at (485pt, 0pt) {};
	\draw (a_1)--(a_2);
	
	\node (b_1) at (4pt, 12pt) {$0$};
	\node (b_2) at (4pt, -10pt) {};	
	\draw(b_1) -- (b_2);
	
	\node (b_1) at (84pt, 12pt) {$1$};
	\node (b_2) at (84pt, -10pt) {};	
	\draw(b_1) -- (b_2);
	
	\node (a_1) at (0pt,-40pt) {}; 
	\node (a_2) at (485pt, -40pt) {};
	\draw (a_1)--(a_2);
	
	\node (b_1) at (4pt, -28pt) {$0$};
	\node (b_2) at (4pt, -50pt) {};	
	\draw(b_1) -- (b_2);
	
	\node (b_1) at (84pt, -28pt) {$1$};
	\node (b_2) at (84pt, -50pt) {};	
	\draw(b_1) -- (b_2);
	
		\node (a_1) at (0pt,-100pt) {}; 
	\node (a_2) at (485pt, -100pt) {};
	\draw[dashed] (a_1)--(a_2);
	
	\node (a) at (50pt,-95pt) {\scriptsize{Regions of the circuit $C_1$}};

	\draw [
	thick,
	decoration={
		brace,
		raise=5pt
	},
	decorate
	] (0pt,50pt) -- (84pt,50pt)
	node [pos=0.5,anchor=south,yshift=5pt] {c-e region};

\draw[fill=lightgray] (295pt,0pt) rectangle (300pt,10pt);
\draw[fill=lightgray] (310pt,0pt) rectangle (315pt,10pt);	

\draw[fill=lightgray] (370pt,0pt) rectangle (385pt,2.5pt);
\draw[fill=lightgray] (400pt,0pt) rectangle (415pt,2.5pt);	

\draw[fill=lightgray] (295pt,-40pt) rectangle (300pt,-30pt);
\draw[fill=lightgray] (310pt,-40pt) rectangle (315pt,-30pt);	

\draw[fill=lightgray] (430pt,-40pt) rectangle (445pt,-37.5pt);
\draw[fill=lightgray] (460pt,-40pt) rectangle (475pt,-37.5pt);	

	\draw [
thick,
decoration={
	brace,
	raise=5pt
},
decorate
] (90pt,50pt) -- (170pt,50pt)
node [pos=0.5,anchor=south,yshift=5pt] {$R_1^{in}$};

	\draw [
thick,
decoration={
	brace,
	raise=5pt
},
decorate
] (172pt,50pt) -- (280pt,50pt)
node [pos=0.5,anchor=south,yshift=5pt] {$R_1^{med}$};

	\draw [
thick,
decoration={
	brace,
	raise=5pt
},
decorate
] (282pt,50pt) -- (480pt,50pt)
node [pos=0.5,anchor=south,yshift=5pt] {$R_1^{out}$};

	\node (a_1) at (0pt,-80pt) {}; 
\node (a_2) at (485pt, -80pt) {};
\draw (a_1)--(a_2);

\node (b_1) at (4pt, -68pt) {$0$};
\node (b_2) at (4pt, -90pt) {};	
\draw(b_1) -- (b_2);

\node (b_1) at (84pt, -68pt) {$1$};
\node (b_2) at (84pt, -90pt) {};	
\draw(b_1) -- (b_2);

\draw[fill=lightgray] (4pt,-80pt) rectangle (84pt,-77.5pt);

\draw[fill=lightgray] (290pt,-80pt) rectangle (295pt,-60pt);
\draw[fill=lightgray] (315pt,-80pt) rectangle (320pt,-60pt);

\draw[fill=lightgray] (302.5pt,-80pt) rectangle (307.5pt,-70pt);

	\draw [
thick,
decoration={
	brace,
	mirror,
	raise=5pt
},
decorate
] (285pt,-100pt) -- (325pt,-100pt)
node [pos=0.5,anchor=north,yshift=-5pt] {\scriptsize{blanket}};

\node (a) at (303.5pt,-120pt)  {\scriptsize{sensor}};

	\draw [
thick,
decoration={
	brace,
	mirror,
	raise=5pt
},
decorate
] (90pt,-100pt) -- (170pt,-100pt)
node [pos=0.5,anchor=north,yshift=-5pt] {\scriptsize{bit extractors and}};
\node (a) at (130pt,-120pt)  {\scriptsize{input gate gadgets}};

\draw [
thick,
decoration={
	brace,
	mirror,
	raise=5pt
},
decorate
] (175pt,-100pt) -- (280pt,-100pt)
node [pos=0.5,anchor=north,yshift=-5pt] {\scriptsize{circuit gate gadgets}};

\draw [
thick,
decoration={
	brace,
	mirror,
	raise=5pt
},
decorate
] (328pt,-100pt) -- (480pt,-100pt)
node [pos=0.5,anchor=north,yshift=-5pt] {\scriptsize{output gate gadgets}};

\node at (-8pt,0pt) {$\alpha_1$};
\node at (-8pt,-40pt) {$\alpha_2$};
\node at (-8pt,-80pt) {$\alpha_1^{bs}$};

\node at (290pt, 30pt) {\scriptsize{$\lplus$}};

\node[color=blue] at (378pt, 30pt) {\scriptsize{$\lplus$}};
\node[color=red] at (358pt, 20pt) {\scriptsize{$\lplus$}};

\node[color=blue]  at (408pt, 30pt) {\scriptsize{$\lplus$}};
\node[color=red] at (392pt, 20pt) {\scriptsize{$\lplus$}};


\node[color=blue]  at (435pt, 30pt) {\scriptsize{$\lplus$}};
\node[color=red] at (422pt, 20pt) {\scriptsize{$\lplus$}};

\node[color=blue]  at (453pt, 30pt) {\scriptsize{$\lplus$}};
\node[color=red] at (468pt, 20pt) {\scriptsize{$\lplus$}};

\draw[dashed] (305pt,50pt) -- (305pt, -100pt);

\draw[dashed,blue] (385pt,50pt) -- (385pt, -100pt);
\draw[dashed,blue] (415pt,50pt) -- (415pt, -100pt);

\draw[dashed,blue] (445pt,50pt) -- (445pt, -100pt);
\draw[dashed,blue] (460pt,50pt) -- (460pt, -100pt);

\draw[dashed,red] (370pt,50pt) -- (370pt, -100pt);
\draw[dashed,red] (400pt,50pt) -- (400pt, -100pt);

\draw[dashed,red] (430pt,50pt) -- (430pt, -100pt);
\draw[dashed,red] (475pt,50pt) -- (475pt, -100pt);

	\end{tikzpicture}
	\caption{\small{A representation of the valuations of the coordinate encoding agents $\alpha_1$ and $\alpha_2$ and the blanket sensor agent $\alpha_1^{bs}$ of $C_1$. For visibility, the valuation blocks are not according to scale. Each coordinate encoding agent has a total value of $1/100$ in $R_1^{in} \cup R_1^{out}$, consisting of a bit-detection gadget in $R_1^{in}$ of total volume $3/400$ and two rectangular valuation blocks in $R_1^{out}$ of volume $1/800$ each. The blanket sensor agent has value $1/10$ uniformly distributed in the c-e region $[0,1]$ and value $0.9$ in $R_1^{out}$, consisting of a bit-detection gadget with a rectangular valuation block of value $2/10$ in-between the two thin blocks of value $4/10$ of the gadget. The figure also denotes the different regions where the valuations of the agents in $C_1$ lie. If the output of the labelling circuit $C$ is $1$, which corresponds to the string $1110$ - depicted by the set of blue cuts, agent $\alpha_1$ is shown an excess of $1/400$ in favour of $\lplus$ in $R_1$. If the output of $C$ is $-1$, which corresponds to the string $0001$ - depicted by the set of red cuts, agent $\alpha_1$ is shown an excess of $1/400$ in favour of $\lminus$. In both cases, agent $\alpha_2$ is shown a balanced partition of $\lplus$ and $\lminus$ in $R_1$. Note that the label on the left-hand side of each cut is $\lplus$ in both cases and that the blanket sensor agent $\alpha_1^{bs}$ is passive.}}
	\label{fig:ce-agents}
	\end{figure}

%% file: bitextractors.tex
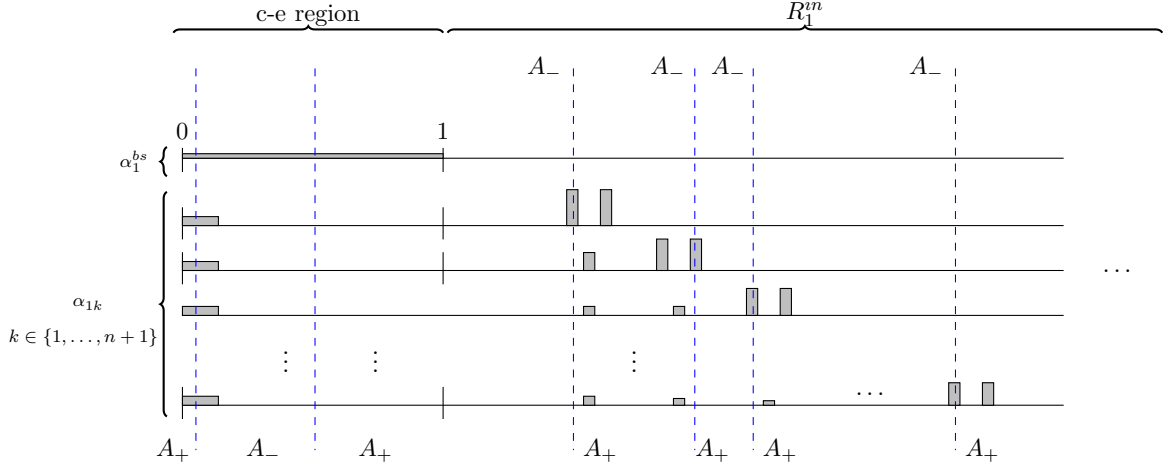
\begin{figure}
\centering

\begin{tikzpicture}[scale=0.85, transform shape]

\node (a_1) at (0pt,0pt) {}; 
\node (a_2) at (400pt, 0pt) {};
\draw (a_1)--(a_2);

	\node (a_1) at (0pt,-30pt) {}; 
\node (a_2) at (400pt, -30pt) {};
\draw (a_1)--(a_2);

\node (b_1) at (4pt, 12pt) {$0$};
\node (b_2) at (4pt, -10pt) {};	
\draw(b_1) -- (b_2);

\node (b_1) at (120pt, 12pt) {$1$};
\node (b_2) at (120pt, -10pt) {};	
\draw(b_1) -- (b_2);

\node (a_1) at (0pt,-50pt) {}; 
\node (a_2) at (400pt, -50pt) {};
\draw (a_1)--(a_2);

\node (b_1) at (4pt, -38pt) {};
\node (b_2) at (4pt, -60pt) {};	
\draw(b_1) -- (b_2);

\node (b_1) at (120pt, -38pt) {};
\node (b_2) at (120pt, -60pt) {};	
\draw(b_1) -- (b_2);



\node (b_1) at (4pt, -18pt) {};
\node (b_2) at (4pt, -40pt) {};	
\draw(b_1) -- (b_2);

\node (b_1) at (120pt, -18pt) {};
\node (b_2) at (120pt, -40pt) {};	
\draw(b_1) -- (b_2);


\node (a_1) at (0pt,-70pt) {}; 
\node (a_2) at (400pt, -70pt) {};
\draw (a_1)--(a_2);

\node (l) at (205pt, -87.5pt) {$\vdots$};
\node (l) at (50pt, -87.5pt) {$\vdots$};
\node (l) at (90pt, -87.5pt) {$\vdots$};

\node (a_1) at (0pt,-110pt) {}; 
\node (a_2) at (400pt, -110pt) {};
\draw (a_1)--(a_2);







\draw [
thick,
decoration={
	brace,
	raise=5pt
},
decorate
] (0pt,50pt) -- (120pt,50pt)
node [pos=0.5,anchor=south,yshift=5pt] {c-e region};	

\draw[
thick,
decoration={
	brace,
	raise=5pt
},
decorate
] (122pt,50pt) -- (440pt,50pt)
node [pos=0.5,anchor=south,yshift=5pt] {$R_1^{in}$};

\node (l) at (420pt, -50pt) {$\ldots$};
	
\draw [
thick,
decoration={
	brace,
	mirror,
	raise=5pt
},
decorate
] (3pt,5pt) -- (3pt,-8pt)
node [pos=0.5,anchor=east,xshift=-10pt] {\scriptsize{$\alpha_1^{bs}$}};

\draw [
thick,
decoration={
	brace,
	mirror,
	raise=5pt
},
decorate
] (3pt,-15pt) -- (3pt,-115pt)
node [pos=0.5,anchor=east,xshift=-30pt] {\scriptsize{$\alpha_{1k}$}};

\node (l) at (-40pt,-80pt) {\scriptsize{$k\in \{1,\ldots,n+1\}$}};

%

\draw[fill=lightgray] (4pt,0pt) rectangle (120pt,2pt);


\draw[fill=lightgray] (4pt,-30pt) rectangle (20pt,-26pt);

\draw[fill=lightgray] (175pt,-30pt) rectangle (180pt,-14pt);
\draw[fill=lightgray] (190pt,-30pt) rectangle (195pt,-14pt);

\draw[fill=lightgray] (4pt,-50pt) rectangle (20pt,-46pt);

\draw[fill=lightgray] (215pt,-50pt) rectangle (220pt,-36pt);
\draw[fill=lightgray] (230pt,-50pt) rectangle (235pt,-36pt);

\draw[fill=lightgray] (182.5pt,-50pt) rectangle (187.5pt,-42pt);

\draw[fill=lightgray] (4pt,-70pt) rectangle (20pt,-66pt);

\draw[fill=lightgray] (255pt,-70pt) rectangle (260pt,-58pt);
\draw[fill=lightgray] (270pt,-70pt) rectangle (275pt,-58pt);

\draw[fill=lightgray] (182.5pt,-70pt) rectangle (187.5pt,-66pt);
\draw[fill=lightgray] (222.5pt,-70pt) rectangle (227.5pt,-66pt);

\node (l) at (310pt,-105pt) {$\ldots$};

\draw[fill=lightgray] (4pt,-110pt) rectangle (20pt,-106pt);

\draw[fill=lightgray] (345pt,-110pt) rectangle (350pt,-100pt);
\draw[fill=lightgray] (360pt,-110pt) rectangle (365pt,-100pt);

\draw[fill=lightgray] (182.5pt,-110pt) rectangle (187.5pt,-106pt);
\draw[fill=lightgray] (222.5pt,-110pt) rectangle (227.5pt,-107pt);
\draw[fill=lightgray] (262.5pt,-110pt) rectangle (267.5pt,-108pt);

%
%

%
%

%
\node (b_1) at (4pt, -98pt) {};
\node (b_2) at (4pt, -120pt) {};	
\draw(b_1) -- (b_2);
\node (b_1) at (120pt, -98pt) {};
\node (b_2) at (120pt, -120pt) {};	
\draw(b_1) -- (b_2);
%

%
%

\draw [dashed,color=blue] (10pt, 40pt) -- (10pt, -130pt);

\draw [dashed,color=blue] (63pt, 40pt) -- (63pt, -130pt);


\draw [dashed,color=blue] (178pt, 40pt) -- (178pt, -130pt);
\draw [dashed,color=blue] (232pt, 40pt) -- (232pt, -130pt);

\draw [dashed,color=blue] (258pt, 40pt) -- (258pt, -130pt);

\draw [dashed,color=blue] (348pt, 40pt) -- (348pt, -130pt);

\node at (0pt,-130pt) {$\lplus$};
\node at (40pt,-130pt) {$\lminus$};
\node at (90pt,-130pt) {$\lplus$};

\node at (165pt,40pt) {$\lminus$};
\node at (190pt,-130pt) {$\lplus$};

\node at (220pt,40pt) {$\lminus$};
\node at (240pt,-130pt) {$\lplus$};

\node at (247pt,40pt) {$\lminus$};
\node at (270pt,-130pt) {$\lplus$};

\node at (335pt,40pt) {$\lminus$};
\node at (360pt,-130pt) {$\lplus$};

	\end{tikzpicture}
	\caption{\small{The $n+8$ c-e identical sensor agents of $\mathcal{S}_i$, which are responsible for extracting the raw data representing cut-positions in an interval of length $\frac{1}{8}$ to $n+8$ bits of precision, which yields $n+11$ bits of precision of the cut in $[0,1]$. For visibility, the valuation are not according to scale but the bit-detection gadgets in the sequence become smaller as one moves to the right and the same holds for the ``intersecting'' blocks of agent $\alpha_{1k}^s$ that lie in between the bit-detection gadgets of agents $\alpha_{1i}^{s}$ with $i>k$. A pair of c-e cuts is also shown: the blanket sensor agent $\alpha_1^{bs}$ detects zero discrepancy and therefore is passive. The bit extractors detect the position of the cut in $[0,1/8]$ by either producing $0$ or $1$ in their output. The sequence of labels when two cuts lie in the c-e region is also shown - note that this sequence can be ensured by the use of parity-gadgets (see Section \ref{sec:gadgets}).}}
	\label{fig:sensor_agents}
	\end{figure}

%% file: bitextractors2.tex
\begin{figure}
\centering

\begin{tikzpicture}[scale=1, transform shape]

\node (a_1) at (0pt,0pt) {}; 
\node (a_2) at (240pt, 0pt) {};
\draw (a_1)--(a_2);

\node (a_1) at (0pt,-30pt) {}; 
\node (a_2) at (240pt, -30pt) {};
\draw (a_1)--(a_2);

\node (a_1) at (0pt,-50pt) {}; 
\node (a_2) at (240pt, -50pt) {};
\draw (a_1)--(a_2);

\node (a_1) at (0pt,-70pt) {}; 
\node (a_2) at (240pt, -70pt) {};
\draw (a_1)--(a_2);

\node (a_1) at (0pt,-90pt) {}; 
\node (a_2) at (240pt, -90pt) {};
\draw (a_1)--(a_2);

\node (a_1) at (0pt,-110pt) {}; 
\node (a_2) at (240pt, -110pt) {};
\draw (a_1)--(a_2);

\node (a_1) at (0pt,-130pt) {}; 
\node (a_2) at (240pt, -130pt) {};
\draw (a_1)--(a_2);

\node (a_1) at (0pt,-150pt) {}; 
\node (a_2) at (240pt, -150pt) {};
\draw (a_1)--(a_2);

\node (a_1) at (0pt,-170pt) {}; 
\node (a_2) at (240pt, -170pt) {};
\draw (a_1)--(a_2);

\node (b_1) at (4pt, 12pt) {$0$};
\node (b_2) at (4pt, -10pt) {};	
\draw(b_1) -- (b_2);

\node (b_1) at (237pt, 12pt) {$1$};
\node (b_2) at (237pt, -10pt) {};	
\draw(b_1) -- (b_2);

\draw [
thick,
decoration={
	brace,
	raise=5pt
},
decorate
] (0pt,20pt) -- (240pt,20pt)
node [pos=0.5,anchor=south,yshift=5pt] {c-e region};	

\draw[
thick,
decoration={
	brace,
	raise=5pt
},
decorate
] (242pt,20pt) -- (300pt,20pt)
node [pos=0.5,anchor=south,yshift=5pt] {$R_1^{in}$};
\node (l) at (275pt, -15pt) {$\ldots$};

\draw[fill=lightgray] (4pt,0pt) rectangle (236pt,2pt);

\draw[fill=lightgray] (4pt,-30pt) rectangle (33pt,-22pt);
\draw[fill=lightgray] (33pt,-50pt) rectangle (62pt,-42pt);
\draw[fill=lightgray] (62pt,-70pt) rectangle (91pt,-62pt);
\draw[fill=lightgray] (91pt,-90pt) rectangle (120pt,-82pt);
\draw[fill=lightgray] (120pt,-110pt) rectangle (149pt,-102pt);
\draw[fill=lightgray] (149pt,-130pt) rectangle (178pt,-122pt);
\draw[fill=lightgray] (178pt,-150pt) rectangle (207pt,-142pt);
\draw[fill=lightgray] (207pt,-170pt) rectangle (236pt,-162pt);

\draw[
thick,
decoration={
	mirror,
	brace,
	raise=5pt
},
decorate
] (208pt,-168pt) -- (235pt,-168pt)
node [pos=0.5,anchor=north,yshift=-10pt] {$1/8$};

\node[circle,draw=black, inner sep=0pt,minimum size=35pt,label={[xshift=2.5cm, yshift=-0.7cm]\scriptsize{$n+8$ c-e identical agents}}] (b) at (19pt,-25pt) {};

\end{tikzpicture}
\caption{\small{The valuation of the blanket sensor agent and the bit-extractors in the c-e region. Only one representative from each set of $n+8$ c-e identical agents is shown.}}
\label{fig:sensor_agents2}
\end{figure}
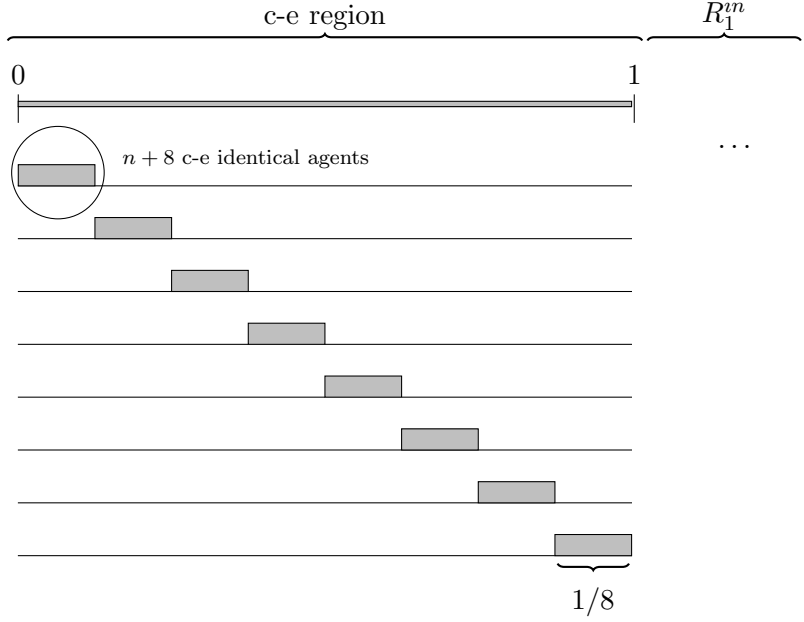

%% file: gates.tex
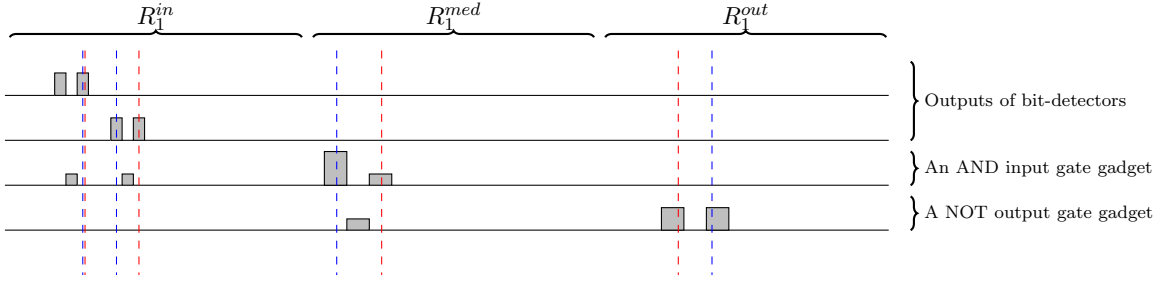
\begin{figure}

\begin{tikzpicture}[scale=0.85, transform shape]
	\node (a_1) at (84pt,0pt) {}; 
	\node (a_2) at (485pt, 0pt) {};
	\draw (a_1)--(a_2);
	
\node (a_1) at (84pt,-20pt) {}; 
\node (a_2) at (485pt, -20pt) {};
\draw (a_1)--(a_2);

	\node (a_1) at (84pt,-40pt) {}; 
	\node (a_2) at (485pt, -40pt) {};
	\draw (a_1)--(a_2);
	
		\node (a_1) at (84pt,-60pt) {}; 
	\node (a_2) at (485pt, -60pt) {};
	\draw (a_1)--(a_2);
	
	

	
	\draw[fill=lightgray] (110pt,0pt) rectangle (115pt,10pt);
	\draw[fill=lightgray] (120pt,0pt) rectangle (125pt,10pt);
	
	\draw[fill=lightgray] (135pt,-20pt) rectangle (140pt,-10pt);
	\draw[fill=lightgray] (145pt,-20pt) rectangle (150pt,-10pt);
	
	

\draw[fill=lightgray] (115pt,-40pt) rectangle (120pt,-35pt);
\draw[fill=lightgray] (140pt,-40pt) rectangle (145pt,-35pt);

	\draw[fill=lightgray] (230pt,-40pt) rectangle (240pt,-25pt);
\draw[fill=lightgray] (250pt,-40pt) rectangle (260pt,-35pt);

\draw[fill=lightgray] (240pt,-60pt) rectangle (250pt,-55pt);

\draw[fill=lightgray] (380pt,-60pt) rectangle (390pt,-50pt);
\draw[fill=lightgray] (400pt,-60pt) rectangle (410pt,-50pt);

\draw [
thick,
decoration={
	brace,
	raise=5pt
},
decorate
] (485pt,15pt) -- (485pt,0-20pt)
node [pos=0.5,anchor=west,xshift=8pt] {\scriptsize{Outputs of bit-detectors}};

\draw [
thick,
decoration={
	brace,
	raise=5pt
},
decorate
] (485pt,-25pt) -- (485pt,-40pt)
node [pos=0.5,anchor=west,xshift=8pt] {\scriptsize{An AND input gate gadget}};

\draw [
thick,
decoration={
	brace,
	raise=5pt
},
decorate
] (485pt,-45pt) -- (485pt,-60pt)
node [pos=0.5,anchor=west,xshift=8pt] {\scriptsize{A NOT output gate gadget}};

	\draw [
thick,
decoration={
	brace,
	raise=5pt
},
decorate
] (90pt,20pt) -- (220pt,20pt)
node [pos=0.5,anchor=south,yshift=5pt] {$R_1^{in}$};

	\draw [
thick,
decoration={
	brace,
	raise=5pt
},
decorate
] (225pt,20pt) -- (350pt,20pt)
node [pos=0.5,anchor=south,yshift=5pt] {$R_1^{med}$};

	\draw [
thick,
decoration={
	brace,
	raise=5pt
},
decorate
] (355pt,20pt) -- (480pt,20pt)
node [pos=0.5,anchor=south,yshift=5pt] {$R_1^{out}$};

\draw[color=blue,dashed] (137.5pt,20pt) -- (137.5pt,-80pt);
\draw[color=blue,dashed] (122.5pt,20pt) -- (122.5pt,-80pt);
\draw[color=blue,dashed] (235.5pt,20pt) -- (235.5pt,-80pt);
\draw[color=blue,dashed] (402.5pt,20pt) -- (402.5pt,-80pt);

\draw[color=red,dashed] (147.5pt,20pt) -- (147.5pt,-80pt);
\draw[color=red,dashed] (123.5pt,20pt) -- (123.5pt,-80pt);
\draw[color=red,dashed] (255.5pt,20pt) -- (255.5pt,-80pt);
\draw[color=red,dashed] (387.5pt,20pt) -- (387.5pt,-80pt);

	\end{tikzpicture}
	\caption{\small{The basic idea behind the gate-agents encoding the gates of $C_1$. The picture denotes a simplified case where two input bits from the bit-encoders are supplied to an enconding of an input AND gate of $C_i$ and the output bit of this gate is in turn supplied to the encoding of a NOT output gate of $C_i$. Note that if for example the sensor agents detect the values $1$ and $0$ respectively (the blue cuts), then the output of the AND gate is $0$ (i.e. the blue cut sits at the left of the AND gate agent's bit detector) and the output of the circuit is $1$ (again, see the blue cut that sit on the rightmost valuation block of the NOT gate agent. Similarly, if the sensor agents detect values $1$ and $1$ (the red cuts), then the output of the AND gate is $1$ and the output of the circuit is $0$.}} 
	\label{fig:gates}
	\end{figure}

%% file: circuit.tex
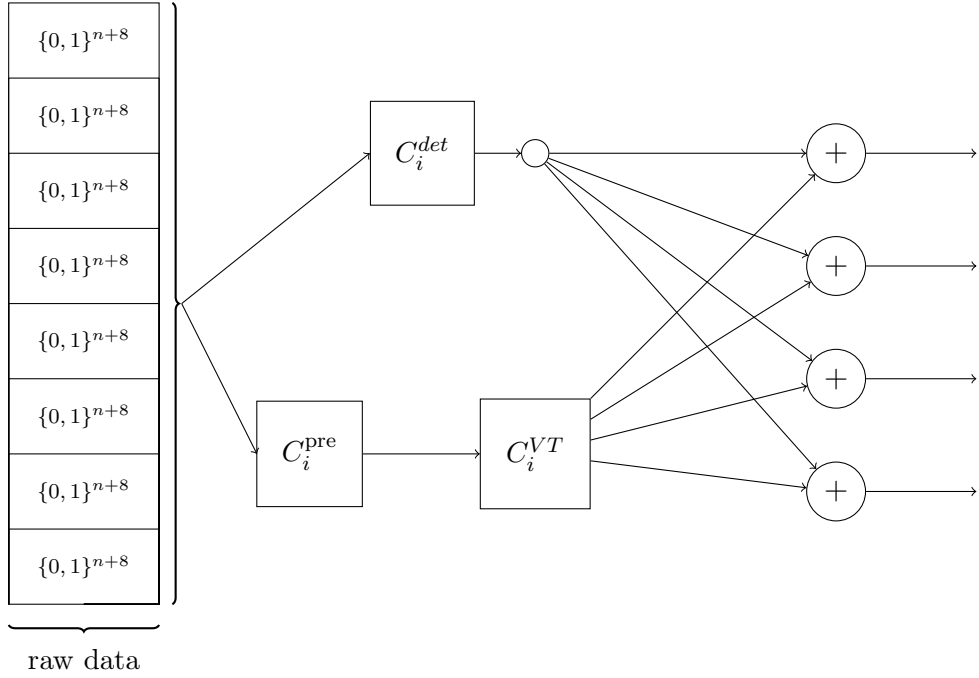
\begin{figure}
\centering
\begin{tikzpicture}[square/.style={regular polygon,regular polygon sides=4}]

\tikzstyle{XORgate} = [draw,circle]

\draw (-1,0) -- (1,0) -- (1,1) -- (-1,1) -- (-1,0);
\draw (0,0) -- (1,0) -- (1,2) -- (-1,2) -- (-1,0);
\draw (0,0) -- (1,0) -- (1,3) -- (-1,3) -- (-1,0);
\draw (0,0) -- (1,0) -- (1,4) -- (-1,4) -- (-1,0);
\draw (0,0) -- (1,0) -- (1,5) -- (-1,5) -- (-1,0);
\draw (0,0) -- (1,0) -- (1,6) -- (-1,6) -- (-1,0);
\draw (0,0) -- (1,0) -- (1,7) -- (-1,7) -- (-1,0);
\draw (0,0) -- (1,0) -- (1,8) -- (-1,8) -- (-1,0);

\node at (0,0.5) {\scriptsize{$\{0,1\}^{n+8}$}};
\node at (0,1.5) {\scriptsize{$\{0,1\}^{n+8}$}};
\node at (0,2.5) {\scriptsize{$\{0,1\}^{n+8}$}};
\node at (0,3.5) {\scriptsize{$\{0,1\}^{n+8}$}};
\node at (0,4.5) {\scriptsize{$\{0,1\}^{n+8}$}};
\node at (0,5.5) {\scriptsize{$\{0,1\}^{n+8}$}};
\node at (0,6.5) {\scriptsize{$\{0,1\}^{n+8}$}};
\node at (0,7.5) {\scriptsize{$\{0,1\}^{n+8}$}};

\node at (3,2) [square,,draw] (Cpre) {$C_{i}^{\mathrm{pre}}$};
\node at (6,2) [square,,draw] (Cvt) {$C_{i}^{VT}$};

\node at (4.5,6) [square,,draw] (Cx) {$C_{i}^{det}$};

\draw[->] (1.3,4) -- (2.3,2);
\draw[->] (1.3,4) -- (3.8,6);
\draw [->] (Cpre) edge (Cvt);

\node [XORgate] at (10,6) (XOR1) {\large +};
\node [XORgate] at (10,4.5) (XOR2) {\large +};
\node [XORgate] at (10,3) (XOR3) {\large +};
\node [XORgate] at (10,1.5) (XOR4) {\large +};

\node at (6,6) [circle,draw] (bit) {};

\draw [->] (Cx) edge (bit);
\draw [->] (Cvt) edge (XOR1);
\draw [->] (Cvt) edge (XOR2);
\draw [->] (Cvt) edge (XOR3);
\draw [->] (Cvt) edge (XOR4);

\draw [->] (bit) edge (XOR1);
\draw [->] (bit) edge (XOR2);
\draw [->] (bit) edge (XOR3);
\draw [->] (bit) edge (XOR4);

\node at (12,6) (out1) {};
\node at (12,4.5) (out2) {};
\node at (12,3) (out3) {};
\node at (12,1.5) (out4) {};

\draw [->] (XOR1) edge (out1);
\draw [->] (XOR2) edge (out2);
\draw [->] (XOR3) edge (out3);
\draw [->] (XOR4) edge (out4);

\draw [
thick,
decoration={
	mirror,
	brace,
	raise=5pt
},
decorate
] (-1,-0.1) -- (1,-0.1)
node [pos=0.5,anchor=north,yshift=-10pt] {raw data};

\draw [
thick,
decoration={
	mirror,
	brace,
	raise=5pt
},
decorate
] (1,0) -- (1,8)
node [pos=0.5,anchor=west,yshift=0pt] {};	

\end{tikzpicture}
\caption{\small{The encoding of $C_i$ consisting of the pre-processing circuit $C_i^\mathrm{pre}$, the labelling circuit $C_i^{VT}$ and the sub-circuit responsible for the XOR operation between the raw data and the output of $C_i^{VT}$.} }
\label{fig:circuit}
\end{figure}

%% file: twosolutions.tex
\tikzset{decorate sep/.style 2 args=
	{decorate,decoration={shape backgrounds,shape=circle,shape size=#1,shape sep=#2}}}

\begin{figure}
	\begin{minipage}{0.5\textwidth}
	\center{
\begin{tikzpicture}[scale=0.6]

\draw[thick,dashed](7,9)--(17,-1);
\draw[thick,dashed](7,-1)--(17,9);
\draw[thick,dashed](15,-1)--(17,1);
\draw[thick,dashed](15,9)--(17,7);

\draw[ultra thick](10,4)--(14,4)--(14,0);
\draw[ultra thick](8,2)--(12,2)--(12,6);



\draw[decorate sep={1mm}{2mm},fill] (11,2.5) -- (12.2,1.3);
\draw[color=red] (11,2.5) circle (0.15cm);
\draw[color=blue] (11.25,2.25) circle (0.15cm);
\draw[color=darkgreen] (12.2,1.3) circle (0.15cm);

\draw[color=purple] (11.6,1.85) circle (0.3cm);

\node at (9.5,3) {$-1$};
\node at (13,0.5) {$1$};
\end{tikzpicture}
}
\end{minipage}
\begin{minipage}{0.5\textwidth}
	\center{
	\begin{tikzpicture}[scale=0.9, transform shape]
	\node (a_1) at (0pt,0pt) {}; 
	\node (a_2) at (204pt, 0pt) {};
	\draw (a_1)--(a_2);
	
	\node (b_1) at (4pt, 12pt) {$0$};
	\node (b_2) at (4pt, -10pt) {};	
	\draw(b_1) -- (b_2);
	
	\node (b_1) at (200pt, 12pt) {$1$};
	\node (b_2) at (200pt, -10pt) {};	
	\draw(b_1) -- (b_2);
	
	\node (a_1) at (0pt,-40pt) {}; 
	\node (a_2) at (204pt, -40pt) {};
	\draw (a_1)--(a_2);
	
	\node (b_1) at (4pt, -28pt) {$0$};
	\node (b_2) at (4pt, -50pt) {};	
	\draw(b_1) -- (b_2);
	
	\node (b_1) at (200pt, -28pt) {$1$};
	\node (b_2) at (200pt, -50pt) {};	
	\draw(b_1) -- (b_2);

		\node (a_1) at (0pt,-80pt) {}; 
	\node (a_2) at (204pt, -80pt) {};
	\draw (a_1)--(a_2);
	
	\node (b_1) at (4pt, -68pt) {};
	\node (b_2) at (4pt, -90pt) {};	
	\draw(b_1) -- (b_2);
	
	\node (b_1) at (200pt, -68pt) {};
	\node (b_2) at (200pt, -90pt) {};	
	\draw(b_1) -- (b_2);
	
	\node (a_1) at (0pt,-100pt) {}; 
	\node (a_2) at (204pt, -100pt) {};
	\draw (a_1)--(a_2);
	
	\node (b_1) at (4pt, -88pt) {};
	\node (b_2) at (4pt, -110pt) {};	
	\draw(b_1) -- (b_2);
	
	\node (b_1) at (200pt, -88pt) {};
	\node (b_2) at (200pt, -110pt) {};	
	\draw(b_1) -- (b_2);
	
		\node (a_1) at (0pt,-140pt) {}; 
	\node (a_2) at (204pt, -140pt) {};
	\draw (a_1)--(a_2);
	
	\node (b_1) at (4pt, -128pt) {};
	\node (b_2) at (4pt, -150pt) {};	
	\draw(b_1) -- (b_2);
	
	\node (b_1) at (200pt, -128pt) {};
	\node (b_2) at (200pt, -150pt) {};	
	\draw(b_1) -- (b_2);

	\node at (80pt, -120pt) {$\vdots$};
	
	\node at (210pt, -80pt) {$C_1$};
	\node at (210pt, -100pt) {$C_2$};
	
	\node at (212pt, -140pt) {$C_{100}$};
	
	\draw [
	thick,
	decoration={
		brace,
		raise=5pt
	},
	decorate
	] (0pt,20pt) -- (200pt,20pt)
	node [pos=0.5,anchor=south,yshift=5pt] {c-e region};
	
%

\draw[<-] (0pt,45pt) -- (8pt,45pt);
\draw[->] (20pt,45pt) -- (28pt,45pt);
\node at (14pt, 45pt) {$x$};

\draw[<-] (132pt,45pt) -- (162pt,45pt);
\draw[->] (172pt,45pt) -- (200pt,45pt);
\node at (168pt, 45pt) {$y$};

	\draw[color=red,line width=0.5mm] (30pt,50pt) -- (30pt,-160pt);
	\draw[color=red,line width=0.5mm] (130pt,50pt) -- (130pt,-160pt);
	
	\draw[dashed,color=blue] (34pt,40pt)--(34pt,-150pt);
	\draw[dashed,color=blue] (134pt,40pt)--(134pt,-150pt);
	
		\draw[dashed,color=darkgreen] (40pt,40pt)--(40pt,-150pt);
	\draw[dashed,color=darkgreen] (140pt,40pt)--(140pt,-150pt);
	
	\draw[fill=red,opacity=0.2] (4pt, -83pt) rectangle (200pt,-77pt);
	
	\draw[fill=blue,opacity=0.2] (4pt, -103pt) rectangle (200pt,-97pt);
	
	\draw[fill=darkgreen,opacity=0.2] (4pt, -143pt) rectangle (200pt,-137pt);
	
	\end{tikzpicture}
}
\end{minipage}
\caption{\small{How solutions to $I_{CH}$ correspond to solutions of $I_{VT}$. The cuts of the solution of $I_{CH}$ are depicted by red thick lines. The corresponding point on the domain of $I_{VT}$ is circled by a red line. The bit-extractors of $C_1$ will detect the position of the point circled in red and will output its label (here $-1$) as a discrepancy in favour of $\lminus$ for Agent $\alpha_1$ of the c-e agents. The bit-extractors of $C_2$, since they are shifted by $2^{-(n+11)}$ to the left, will interpret the position of the cut as if it was shifted by $2^{-(n+11)}$ to the right. In other words, they detect the virtual cuts which are depicted by the blue thin dotted lines, and will thus encode the coordinates of the point circled by blue on the left, which is $(x+2^{-(n+11)},y-2^{-(n+11)})$. The label of the point (here $-1$) will be outputted by the circuit in the form of a discrepancy in favour of $\lminus$ for Agent $\alpha_1$ of the c-e agents. Circuit $C_{100}$ will interpret the cut as being shifted by $99\cdot 2^{-(n+11)}$ to the right, with the corresponding set of virtual cuts denoted by the dotted thin green lines in the figure. The encoded point will be the point of the sequence circled by green and its label (here $1$) will be outputted by the circuit in the form of a discrepancy in favour of $\lplus$ for Agent $\alpha_1$ of the c-e agents. The coordinates of all the points in the sequence are extracted in a similar way and contribute to the volumes of $\lplus$ and $\lminus$ that are shown to Agent $\alpha_1$. Note that the points near the boundary of the two neighbouring tiles (circled by a purple line) are arbitrarily labelled (i.e. they don't receive good inputs) but there are at most $4$ of them and therefore, they can't introduce any artificial solutions.}}
\label{fig:twosolutions}
\end{figure}
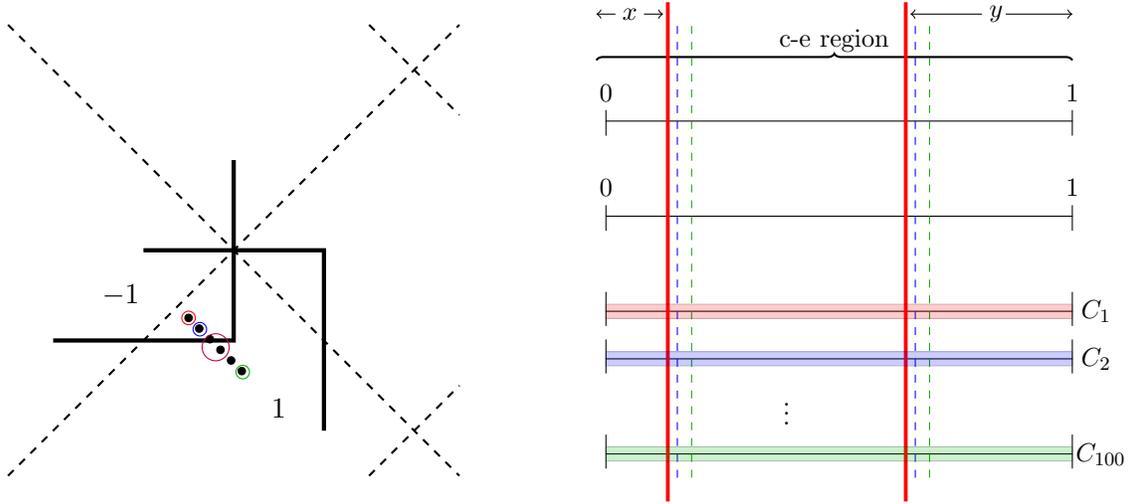

%% file: same_output.tex
\begin{figure}
\centering
%
%
%
%
%
%
%
%
%
%
%
%
%
	\centering
	\begin{tikzpicture}[scale=1, transform shape]
	\node (a_1) at (0pt,0pt) {}; 
	\node (a_2) at (150pt, 0pt) {};
	\draw (a_1)--(a_2);
	
	\node (a_1) at (0pt,-20pt) {}; 
	\node (a_2) at (150pt, -20pt) {};
	\draw (a_1)--(a_2);
	
	\node (a_1) at (0pt,-40pt) {}; 
	\node (a_2) at (150pt, -40pt) {};
	\draw (a_1)--(a_2);

	\draw[fill=lightgray] (4pt,-0pt) rectangle (34pt,8pt);
	\draw[fill=lightgray] (34pt,-20pt) rectangle (64pt,-12pt);
	\draw[fill=lightgray] (64pt,-40pt) rectangle (94pt,-32pt);
	
	\node[color=red,opacity=0.4] at (18pt, 14pt) {\scriptsize{$111\ldots 1$}};
	\node[color=red,opacity=0.4] at (50pt, -6pt) {\scriptsize{$111\ldots 1$}};
	
	\node[color=blue] at (18pt, -6pt) {\scriptsize{$000\ldots 0$}};
	\node[color=blue] at (50pt, -26pt) {\scriptsize{$000\ldots 0$}};
	
\draw[color=darkgreen, thick] (84pt,20pt) -- (84pt,-60pt);

\draw[color=darkgreen, dashed] (74pt,20pt) -- (74pt,-60pt);

\draw [
thick,
color=blue,
decoration={
	mirror,
	brace,
	raise=5pt
},
decorate
] (0pt,-90pt) -- (74pt,-90pt)
node [color=blue,pos=0.5,anchor=north,yshift=-8pt] {$x$};

	\draw [
thick,
opacity=0.4,
color=red,
decoration={
	mirror,
	brace,
	raise=5pt
},
decorate
] (0pt,-70pt) -- (74pt,-70pt)
node [color=red,pos=0.5,anchor=north,yshift=-8pt] {$x$};

	\draw [
thick,
opacity=0.4,
color=darkgreen,
decoration={
	mirror,
	brace,
	raise=5pt
},
decorate
] (64pt,-40pt) -- (73pt,-40pt)
node [opacity=0.6,color=darkgreen,pos=0.5,anchor=north,yshift=-8pt] {$z$};

	\draw [
thick,
color=darkgreen,
decoration={
	mirror,
	brace,
	raise=5pt
},
decorate
] (64pt,-55pt) -- (83pt,-55pt)
node [color=darkgreen,pos=0.5,anchor=north,xshift=-5pt,yshift=-5pt] {$z'$};
\end{tikzpicture}
\caption{\small{How the pre-processing circuit interprets the raw data to extract the coordinate of $x$. The situation after a stray cut is inserted to the left of $C_i$ is depicted. Before the insertion of the cut (shown as transparent), the solid strings detected were strings of $1$'s, the position of the cut was at distance $z$ from the left-hand side of the valuation block and the extracted value of $x$ was as shown in red. After the insertion of the cut, the solid blocks become strings of $0$'s, but the detected position of the cut is now at distance $z'$ from the left-hand side of the valuation block - this means that the position of the cut is detected at $3/8-z' = 2/8+z$, and the value of $x$, now as shown in blue, is exactly the same as before.}}
\label{fig:same_output}
\end{figure}
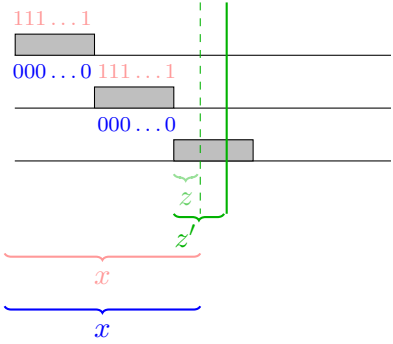

%% file: doublenegative.tex
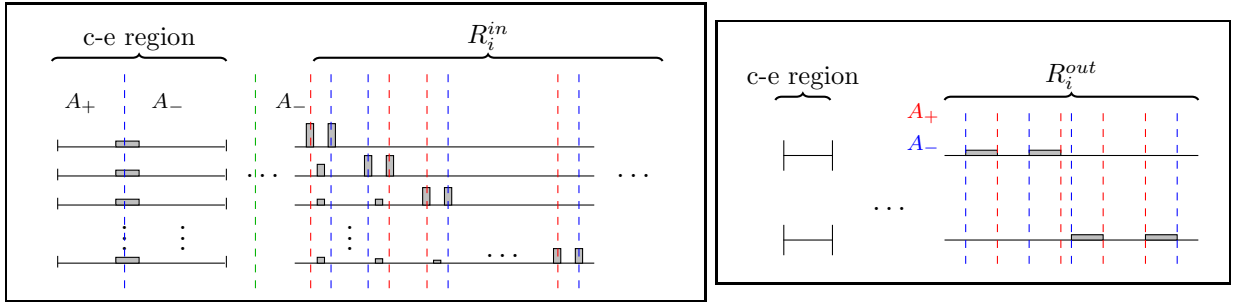
\begin{figure}
	\begin{fmpage}{0.55\textwidth}
	\center{
\begin{tikzpicture}[scale=0.55]

\node (a_1) at (-3pt,-30pt) {}; 
\node (a_2) at (126pt, -30pt) {};
\draw (a_1)--(a_2);

\node (a_1) at (160pt,-30pt) {}; 
\node (a_2) at (380pt, -30pt) {};
\draw (a_1)--(a_2);

\node (a_1) at (-3pt,-50pt) {}; 
\node (a_2) at (126pt, -50pt) {};
\draw (a_1)--(a_2);

\node (a_1) at (160pt,-50pt) {}; 
\node (a_2) at (380pt, -50pt) {};
\draw (a_1)--(a_2);

\node (b_1) at (4pt, -38pt) {};
\node (b_2) at (4pt, -60pt) {};	
\draw(b_1) -- (b_2);

\node (b_1) at (120pt, -38pt) {};
\node (b_2) at (120pt, -60pt) {};	
\draw(b_1) -- (b_2);



\node (b_1) at (4pt, -18pt) {};
\node (b_2) at (4pt, -40pt) {};	
\draw(b_1) -- (b_2);

\node (b_1) at (120pt, -18pt) {};
\node (b_2) at (120pt, -40pt) {};	
\draw(b_1) -- (b_2);

\node (b_1) at (120pt, -58pt) {};
\node (b_2) at (120pt, -80pt) {};	
\draw(b_1) -- (b_2);

\node (b_1) at (4pt, -58pt) {};
\node (b_2) at (4pt, -80pt) {};	
\draw(b_1) -- (b_2);

\node (a_1) at (-3pt,-70pt) {}; 
\node (a_2) at (126pt, -70pt) {};
\draw (a_1)--(a_2);

\node (a_1) at (160pt,-70pt) {}; 
\node (a_2) at (380pt, -70pt) {};
\draw (a_1)--(a_2);

\node (l) at (205pt, -87.5pt) {$\vdots$};
\node (l) at (50pt, -87.5pt) {$\vdots$};
\node (l) at (90pt, -87.5pt) {$\vdots$};

\node (a_1) at (-3pt,-110pt) {}; 
\node (a_2) at (126pt, -110pt) {};
\draw (a_1)--(a_2);

\node (a_1) at (160pt,-110pt) {}; 
\node (a_2) at (380pt, -110pt) {};
\draw (a_1)--(a_2);







\draw [
thick,
decoration={
	brace,
	raise=5pt
},
decorate
] (0pt,20pt) -- (120pt,20pt)
node [pos=0.5,anchor=south,yshift=5pt] {\small{c-e region}};	

\draw[
thick,
decoration={
	brace,
	raise=5pt
},
decorate
] (180pt,20pt) -- (420pt,20pt)
node [pos=0.5,anchor=south,yshift=5pt] {\small{$R_i^{in}$}};

\node (l) at (400pt, -50pt) {$\ldots$};
\node (l) at (145pt, -50pt) {$\ldots$};

%



\draw[fill=lightgray] (44pt,-30pt) rectangle (60pt,-26pt);

\draw[fill=lightgray] (175pt,-30pt) rectangle (180pt,-14pt);
\draw[fill=lightgray] (190pt,-30pt) rectangle (195pt,-14pt);

\draw[fill=lightgray] (44pt,-50pt) rectangle (60pt,-46pt);

\draw[fill=lightgray] (215pt,-50pt) rectangle (220pt,-36pt);
\draw[fill=lightgray] (230pt,-50pt) rectangle (235pt,-36pt);

\draw[fill=lightgray] (182.5pt,-50pt) rectangle (187.5pt,-42pt);

\draw[fill=lightgray] (44pt,-70pt) rectangle (60pt,-66pt);

\draw[fill=lightgray] (255pt,-70pt) rectangle (260pt,-58pt);
\draw[fill=lightgray] (270pt,-70pt) rectangle (275pt,-58pt);

\draw[fill=lightgray] (182.5pt,-70pt) rectangle (187.5pt,-66pt);
\draw[fill=lightgray] (222.5pt,-70pt) rectangle (227.5pt,-66pt);

\node (l) at (310pt,-105pt) {$\ldots$};

\draw[fill=lightgray] (44pt,-110pt) rectangle (60pt,-106pt);

\draw[fill=lightgray] (345pt,-110pt) rectangle (350pt,-100pt);
\draw[fill=lightgray] (360pt,-110pt) rectangle (365pt,-100pt);

\draw[fill=lightgray] (182.5pt,-110pt) rectangle (187.5pt,-106pt);
\draw[fill=lightgray] (222.5pt,-110pt) rectangle (227.5pt,-107pt);
\draw[fill=lightgray] (262.5pt,-110pt) rectangle (267.5pt,-108pt);

%
%

%
%

%
\node (b_1) at (4pt, -98pt) {};
\node (b_2) at (4pt, -120pt) {};	
\draw(b_1) -- (b_2);
\node (b_1) at (120pt, -98pt) {};
\node (b_2) at (120pt, -120pt) {};	
\draw(b_1) -- (b_2);
%

%
%

\draw [dashed,color=blue] (50pt, 20pt) -- (50pt, -130pt);


\draw [dashed,color=darkgreen] (140pt, 20pt) -- (140pt, -130pt);

\draw [dashed,color=blue] (192pt, 20pt) -- (192pt, -130pt);

\draw [dashed,color=blue] (217.5pt, 20pt) -- (217.5pt, -130pt);

\draw [dashed,color=blue] (272.5pt, 20pt) -- (272.5pt, -130pt);

\draw [dashed,color=blue] (362.5pt, 20pt) -- (362.5pt, -130pt);

\draw [dashed,color=red] (178pt, 20pt) -- (178pt, -130pt);
\draw [dashed,color=red] (232pt, 20pt) -- (232pt, -130pt);

\draw [dashed,color=red] (258pt, 20pt) -- (258pt, -130pt);

\draw [dashed,color=red] (348pt, 20pt) -- (348pt, -130pt);

\node at (20pt, 0pt) {\scriptsize{$\lplus$}};
\node at (80pt, 0pt) {\scriptsize{$\lminus$}};

\node at (165pt, 0pt) {\scriptsize{$\lminus$}};

\end{tikzpicture}
}
\end{fmpage}
\begin{fmpage}{0.4\textwidth}
\center{
	\begin{tikzpicture}[scale=0.8]
	\node (a_1) at (-1pt,0pt) {}; 
	\node (a_2) at (32pt, 0pt) {};
	\draw (a_1)--(a_2);
	
	\node (b_1) at (4pt, 12pt) {};
	\node (b_2) at (4pt, -12pt) {};	
	\draw(b_1) -- (b_2);
	
	\node (b_1) at (27pt, 12pt) {};
	\node (b_2) at (27pt, -12pt) {};	
	\draw(b_1) -- (b_2);
	
	\node (a_1) at (-1pt,-40pt) {}; 
	\node (a_2) at (32pt, -40pt) {};
	\draw (a_1)--(a_2);
	
	\node (b_1) at (4pt, -28pt) {};
	\node (b_2) at (4pt, -52pt) {};	
	\draw(b_1) -- (b_2);
	
	\node (b_1) at (27pt, -28pt) {};
	\node (b_2) at (27pt, -52pt) {};	
	\draw(b_1) -- (b_2);

	\node at (55pt, -25pt) {\ldots};
	
	\draw (80pt, 0pt) -- (200pt,0pt);
	\draw (80pt, -40pt) -- (200pt,-40pt);
	
	
	\draw [
	thick,
	decoration={
		brace,
		raise=5pt
	},
	decorate
	] (0pt,20pt) -- (26pt,20pt)
	node [pos=0.5,anchor=south,yshift=5pt] {\small{c-e region}};

	\draw[fill=lightgray] (90pt,0pt) rectangle (105pt,2.5pt);
	\draw[fill=lightgray] (120pt,0pt) rectangle (135pt,2.5pt);

	\draw[fill=lightgray] (140pt,-40pt) rectangle (155pt,-37.5pt);
	\draw[fill=lightgray] (175pt,-40pt) rectangle (190pt,-37.5pt);

	\draw [
	thick,
	decoration={
		brace,
		raise=5pt
	},
	decorate
	] (80pt,20pt) -- (200pt,20pt)
	node [pos=0.5,anchor=south,yshift=5pt] {\small{$R_i^{out}$}};

 	\draw [dashed,color=red] (155pt, 20pt) -- (155pt, -55pt);
 	\draw [dashed,color=blue] (90pt, 20pt) -- (90pt, -55pt);
 	\draw [dashed,color=blue] (120pt, 20pt) -- (120pt, -55pt);
 		\draw [dashed,color=blue] (140pt, 20pt) -- (140pt, -55pt);
 	
 	\draw [dashed,color=blue] (190pt, 20pt) -- (190pt, -55pt);
 	\draw [dashed,color=red] (175pt, 20pt) -- (175pt, -55pt);
 	 \draw [dashed,color=red] (135pt, 20pt) -- (135pt, -55pt);
 	  \draw [dashed,color=red] (105pt, 20pt) -- (105pt, -55pt);
 	 
	\node[color=red] at (70pt,20pt) {\scriptsize{$\lplus$}};
	\node[color=blue] at (70pt,5pt) {\scriptsize{$\lminus$}};
	
	\node at (150pt, -50pt) {};
	\node at (150pt, 55pt) {};)
	\end{tikzpicture}
}
\end{fmpage}
\caption{\small{The ``double negative'' effect of the stray cut (depicted in green). \emph{On the left:} The cut that intersects the bit extractors of circut $C_i$ (the $j$th group shown here, for $j \in \{1,\ldots,8\}$) has $\lplus$ on its left-hand side. The effect of the stray cut is that the cuts that intersect the outputs of the bit-extractors have $\lminus$ on their left-hand side and therefore they output the bit-wise complement of the binary representation of the position of the cut in the corresponding sub-interval of length $1/8$ of the c-e region where the inputs of the bit-extractors lie. The red dashed lines indicate the positions of the cuts had the stray cut not been present. \emph{On the right:} The outputs of the circuit as affected by the stray cut. Some output bits are flipped, i.e. when the output was supposed to be $1110$ (the red cuts), the output is actually $0001$ (the blue cuts), as a result of the $XOR$ sub-circuit (see Section \ref{sec:gate}). Since the stray cut also changes the sequence of labels (i.e. the red sequence starts with $\lplus$ and the blue sequence starts with $\lminus$), effectively the cut in the c-e region introduces exactly the same discrepancy for the c-e agent $\alpha_i$ in the labels of $\lplus$ and $\lminus$.}}
\label{fig:doubleneg}
\end{figure}

%% file: necklace.tex
In this section, we will prove that approximate Consensus Halving and the well-known Necklace Splitting problem \cite{Pap} are computationally equivalent, in the sense that the reduce to each other in polynomial time. It is important to point out that while our construction is quite general, it does not imply that the Necklace Splitting problem is $\ppa$-hard, because the reduction requires the approximate Consensus Halving problem to have an inverse-polynomial precision parameter $\epsilon$, but it certainly indicates in that direction. In particular, a $\ppa$-hardness result for Consensus Halving with inverse-polynomial accuracy would imply $\ppa$-completeness of the Necklace Splitting problem as well (containment in $\ppa$ is known from \cite{Pap}).

Note also that given that \cite{FFGZ} proved that approximate Consensus Halving is $\ppad$-hard for constant precision parameter $\epsilon$, we obtain as a corollary here that Necklace Splitting is $\ppa$-hard, which partially answers an open question raised in \cite{ABB}. 

\begin{theorem}
	Necklace Splitting is $\ppad$-hard.
\end{theorem}
Below, we define the problem formally.

\begin{definition}[Necklace Splitting \cite{Pap}]
In the necklace splitting problem, there is an open necklace (an interval) with $k \cdot m$ beads, each of which has one of $n$ colours. There are precisely $\alpha_i \cdot k$ beads of colour $i=1,\ldots,n$, where $\alpha_i \in \mathbb{N}^+$. The task is to partition the interval into $k$ (not necessarily) contiguous pieces such that each piece contains exactly $\alpha_i$ beads of colour $i$, using at most $(k-1)\cdot n$ cuts.
\end{definition}

We will denote the associated computational problem as $(n,k)$-\necklace. In fact, we will define a more general version of the problem where we are allowed to use $\ell$ cuts, for some $\ell$ which is bounded by a polynomial in $n$. Let $b_i(O)$ denote the number of beads in an interval $O$.
			
\begin{definition}[$(n,\ell,k)$-\necklace]
	\begin{itemize}
		\item[]
\item[-] \textbf{Input:} $k \cdot m$ beads placed on an interval $O$ with $\alpha_i \cdot k$ beads of colour $i=1,\ldots,n$ where $\alpha_i \in \mathbb{N}^+$, with $k \leq n$.
				
\item[-] \textbf{Output:} A partition of $O$ into $k$ parts $O_1,O_2,\ldots, O_k$ such that for each colour $i=1,\ldots,n$, it holds that for each $j \in \{1,\ldots,k\}$, it holds that $b_i(O_j)=\alpha_i$, using $(k-1)\cdot \ell$ cuts.
\end{itemize}
\end{definition}	
			
			\noindent Also, we define a generalisation to the consensus-halving problem, the approximate consensus $(1/k)$ division problem \cite{SS03}. In this problem, we are looking for an solution with $(k-1)\cdot \ell$ cuts that divides the interval into $k$ portions, which are of approximately the same value for each agent.
			
			\begin{definition}[$(n,\ell,\epsilon)$-\CKD]
				\begin{itemize}
					\item[]
					\item[-] \textbf{Input:} The value measures $\mu_i : O \rightarrow R_+, i=1,\cdots,n$, for $n$ agents and $k \leq n$. 
					
				\item[-]	\textbf{Output:} A partition $(O_{1},O_{2},\ldots O_{k})$ with $(k-1) \cdot \ell$ cuts such that $|u_i(O_{t})-u_i(O_{j})|\leq\epsilon$ for all $t$ and $j$ and for all agents $i \in N$.
				\end{itemize}
				\end{definition}	
				We will use the terms ``consensus division'' and ``necklace splitting'' loosely to refer to these problems without specifying the number of partitions or cuts.
			
			\subsection{From Approximate Consensus-Division to Necklace Splitting}
			
			In this subsection, we will establish a reduction from $(n,\ell,\epsilon)$-\CKD\ to $(n,\ell,k)$-\necklace, for all $\ell$ which are bounded by a polynomial in $n$. The following facts hold about any instance of $(n,\ell,\epsilon)$-\CKD. 
			\begin{itemize}
				\item All the agents' valuations are represented as piecewise constant functions.
				\item The number of pieces of these functions is upper bounded by some $p_{\mathcal{M}}(n)$ where $p_{\mathcal{M}}$ is a polynomial.
				\item The volume of each piece is upper bounded by some $p_{\mathcal{V}}(n)$ where $p_{\mathcal{V}}$ is a polynomial.\\
			\end{itemize}
			
				\begin{theorem}\label{thm:divtoneckreduction}
					$(n,\ell,\epsilon)$-\CKD\ is polynomial-time reducible to $(n,\ell,k)$-\necklace, when the number of cuts $\ell$ is bounded by a polynomial in $n$ and $\epsilon$ is inverse-polynomial in $n$.
				\end{theorem}
			
		\begin{proof}	
		Let $\mathcal{C}$ be such an instance of $(n,\ell,\epsilon)$-\CKD. We will construct an instance $\mathcal{B}$ of $(n,\ell,k)$-\necklace\ as follows: For each agent $i \in \{1,2,\ldots,n\}$ of $\mathcal{C}$,
			\begin{itemize}
				\item Associate a different colour $c_i$.
				\item Repeat for all of agent $i$'s valuation blocks:\\
				\begin{itemize}
				\item Let $V$ be the volume of the block and let $\alpha$ be the interval on which the block is defined. Divide the block into $V/\delta$ sub-blocks of volume $\delta$ each, where 
				\[\delta=\frac{\epsilon}{n^3[(k-1)(\ell+1)+p_M(n)]},\]
				except possibly the last sub-block which will have volume $\delta'\leq\delta$. We will call such a sub-block an \emph{imperfect sub-block}. Let $\alpha_j$ denote the corresponding intervals, for $j=1,\ldots, \lceil V/\delta \rceil$. \\
				\item Place a bead of colour $c_i$ in the middle of each interval $\alpha_j$. \\
				\end{itemize}
				\item If the total number $b_i$ of beads of colour $c_i$ placed is not a multiple of $k$, remove $b_i \bmod k$ beads of colour $c_i$. We will refer to these beads as the \emph{parity beads}.
			\end{itemize}
			Intuitively, each bead ``represents'' a valuation block of volume $\delta$ and some beads represent the imperfect sub-blocks of smaller volume. See Figure \ref{fig:necklace} for an example of the construction, when $k=2$ and $\ell=n=2$. Note that the construction requires to partition the instance into at most $p_{\mathcal{M}}(n) \cdot p_{\mathcal{V}}(n)/\delta$ intervals and find their midpoints and therefore runs in polynomial time. Next, we will argue for correctness.
			
			\begin{figure}[t]
				\center
			\begin{tikzpicture}[scale=0.7]
			\node (a_1) at (-2,15) {Agent 1}; 
			\node (b_1) at (-2,10) {Agent 2}; 
			\node () at (-2,13) {Color 1};
			\node () at (-2,8) {Color 2};
			\node () at (-2,7) {Necklace};

			\draw (-1,10) -- (20,10);
			
			\draw (-1,15) -- (20,15);
			\draw [dashed] (-1,13) -- (20,13);
			\draw [dashed] (-1,8) -- (20,8);
			\draw [dashed] (-1,7) -- (20,7);
			
			\draw[fill=lightgray] (0,15) rectangle (3.6,17);
			\draw[fill=red!20] (3.6,15) rectangle (4.2,17);
			\draw[fill=lightgray] (8,15) rectangle (15.8,20);
			\draw[fill=red!20] (15.8,15) rectangle (16.1,20);

			\draw[fill=lightgray] (10.8,10) rectangle (15.6,12);
			\draw[fill=red!20] (15.6,10) rectangle (16,12);
			\draw[fill=lightgray] (5,10) rectangle (7.4,12);
			\draw[fill=red!20] (7.4,10) rectangle (8,12);

			\draw[dashed] (1.2,17) -- (1.2,15);
			\node (d_1) at (0.6,16) {$\mathbf{\delta}$}; 
			
			\draw[dashed] (2.4,17) -- (2.4,15);
			\node (d_2) at (1.8,16) {$\mathbf{\delta}$}; 
			
			\draw[dashed] (3.6,17) -- (3.6,15);
			\node (d_3) at (3,16) {$\mathbf{\delta}$}; 
			
			\node (d_4) at (3.9,16) {$\mathbf{\delta'}$}; 
			
			\draw[dashed] (8.6,20) -- (8.6,15);
			\node (d_1) at (8.3,17.5) {$\mathbf{\delta}$}; 
			
			\draw[dashed] (9.2,20) -- (9.2,15);
			\node (d_1) at (8.9,17.5) {$\mathbf{\delta}$}; 
			
			\draw[dashed] (9.8,20) -- (9.8,15);
			\node (d_1) at (9.5,17.5) {$\mathbf{\delta}$}; 
			
			\draw[dashed] (10.4,20) -- (10.4,15);
			\node (d_1) at (10.1,17.5) {$\mathbf{\delta}$}; 
			
			\draw[dashed] (11,20) -- (11,15);
			\node (d_1) at (10.7,17.5) {$\mathbf{\delta}$}; 
			
			\draw[dashed] (11.6,20) -- (11.6,15);
			\node (d_1) at (11.3,17.5) {$\mathbf{\delta}$}; 
			
			\draw[dashed] (12.2,20) -- (12.2,15);
			\node (d_1) at (11.9,17.5) {$\mathbf{\delta}$}; 
			
			\draw[dashed] (12.8,20) -- (12.8,15);
			\node (d_1) at (12.5,17.5) {$\mathbf{\delta}$}; 
			
			\draw[dashed] (13.4,20) -- (13.4,15);
			\node (d_1) at (13.1,17.5) {$\mathbf{\delta}$}; 
			
			\draw[dashed] (14,20) -- (14,15);
			\node (d_1) at (13.7,17.5) {$\mathbf{\delta}$}; 
			
			\draw[dashed] (14.6,20) -- (14.6,15);
			\node (d_1) at (14.3,17.5) {$\mathbf{\delta}$}; 
			
			\draw[dashed] (15.2,20) -- (15.2,15);
			\node (d_1) at (14.9,17.5) {$\mathbf{\delta}$}; 
			
			\draw[dashed] (15.8,20) -- (15.8,15);
			\node (d_1) at (15.5,17.5) {$\mathbf{\delta}$}; 
			
			
			\draw[dashed] (6.2,12) -- (6.2,10);
			\node (d_1) at (5.6,11) {$\mathbf{\delta}$}; 
			
			\draw[dashed] (7.4,12) -- (7.4,10);
			\node (d_1) at (6.8,11) {$\mathbf{\delta}$}; 
			
			
				\draw[dashed] (12,12) -- (12,10);
				\node (d_1) at (11.4,11) {$\mathbf{\delta}$}; 
			
			\draw[dashed] (13.2,12) -- (13.2,10);
			\node (d_1) at (12.6,11) {$\mathbf{\delta}$}; 
			
			\draw[dashed] (14.4,12) -- (14.4,10);
			\node (d_1) at (13.8,11) {$\mathbf{\delta}$}; 
			
			\draw[dashed] (15.6,12) -- (15.6,10);
			\node (d_1) at (15,11) {$\mathbf{\delta}$};

			\node[circle,draw=black,fill=red] at (0.6,13) {};
			\node[circle,draw=black,fill=red] at (1.8,13) {};
			\node[circle,draw=black,fill=red] at (3,13) {};
			\node[circle,draw=black,fill=red] at (4,13) {};
			
			\node[circle,draw=black,fill=red] at (8.3,13) {};
			\node[circle,draw=black,fill=red] at (8.9,13) {};
			\node[circle,draw=black,fill=red] at (9.5,13) {};
			\node[circle,draw=black,fill=red] at (10.1,13) {};
			\node[circle,draw=black,fill=red] at (10.7,13) {};
			\node[circle,draw=black,fill=red] at (11.3,13) {};
			\node[circle,draw=black,fill=red] at (11.9,13) {};
			\node[circle,draw=black,fill=red] at (12.5,13) {};
			\node[circle,draw=black,fill=red] at (13.1,13) {};
			\node[circle,draw=black,fill=red] at (13.7,13) {};
			\node[circle,draw=black,fill=red] at (14.3,13) {};
			\node[circle,draw=black,fill=red] at (14.9,13) {};
			\node[circle,draw=black,fill=red] at (15.5,13) {};
			\node[circle,draw=black,fill=red] at (16,13) {};
			
			\node[circle,draw=black,fill=green] at (5.6,8) {};
			\node[circle,draw=black,fill=green] at (6.8,8) {};
			\node[circle,draw=black,fill=green] at (7.7,8) {};
			
			\node[circle,draw=black,fill=green] at (11.4,8) {};
			\node[circle,draw=black,fill=green] at (12.6,8) {};
			\node[circle,draw=black,fill=green] at (13.8,8) {};
			\node[circle,draw=black,fill=green] at (15,8) {};
			\node[circle,draw=black,fill=green] at (15.8,8) {};

			\node[circle,draw=black,fill=red] at (0.6,7) {};
			\node[circle,draw=black,fill=red] at (1.8,7) {};
			\node[circle,draw=black,fill=red] at (3,7) {};
			\node[circle,draw=black,fill=red] at (4,7) {};
			
			\node[circle,draw=black,fill=red] at (8.3,7) {};
			\node[circle,draw=black,fill=red] at (8.9,7) {};
			\node[circle,draw=black,fill=red] at (9.5,7) {};
			\node[circle,draw=black,fill=red] at (10.1,7) {};
			\node[circle,draw=black,fill=red] at (10.7,7) {};
			\node[circle,draw=black,fill=red] at (11.3,7) {};
			\node[circle,draw=black,fill=red] at (11.9,7) {};
			\node[circle,draw=black,fill=red] at (12.5,7) {};
			\node[circle,draw=black,fill=red] at (13.1,7) {};
			\node[circle,draw=black,fill=red] at (13.7,7) {};
			\node[circle,draw=black,fill=red] at (14.3,7) {};
			\node[circle,draw=black,fill=red] at (14.9,7) {};
			\node[circle,draw=black,fill=red] at (15.5,7) {};
			\node[circle,draw=black,fill=red] at (16,7) {};
			
			\node[circle,draw=black,fill=green] at (5.6,7) {};
			\node[circle,draw=black,fill=green] at (6.8,7) {};
			\node[circle,draw=black,fill=green] at (7.7,7) {};
			
			\node[circle,draw=black,fill=green] at (11.4,7) {};
			\node[circle,draw=black,fill=green] at (12.6,7) {};
			\node[circle,draw=black,fill=green] at (13.8,7) {};
			\node[circle,draw=black,fill=green] at (15,7) {};
			\node[circle,draw=black,fill=green] at (15.8,7) {};

			\draw[ultra thick,dotted,draw=blue,|->] (9.8,20.5) -- (9.8,6);
			\draw[ultra thick,dotted,draw=blue,|->] (15.3,20.5) -- (15.3,6);

			\end{tikzpicture}
			\caption{An example of the reduction when $k=2$ and $\ell=n=2$. The red (dark) beads correspond to Agent 1 and the green (light) beads correspond to Agent 2. The pink (lightgray) area corresponds to the last part of each valuation block, which has volume $\delta' \leq \delta$. The blue dotted lines indicate the positions of the cuts; note that while the first cut lies exactly at the boundary of two sub-blocks of valuation $\delta$ for Agent 1 and in a ``value-free'' region for Agent 2, the second cut intersects the interior of a sub-interval for both agents.}
			\label{fig:necklace}
			\end{figure}
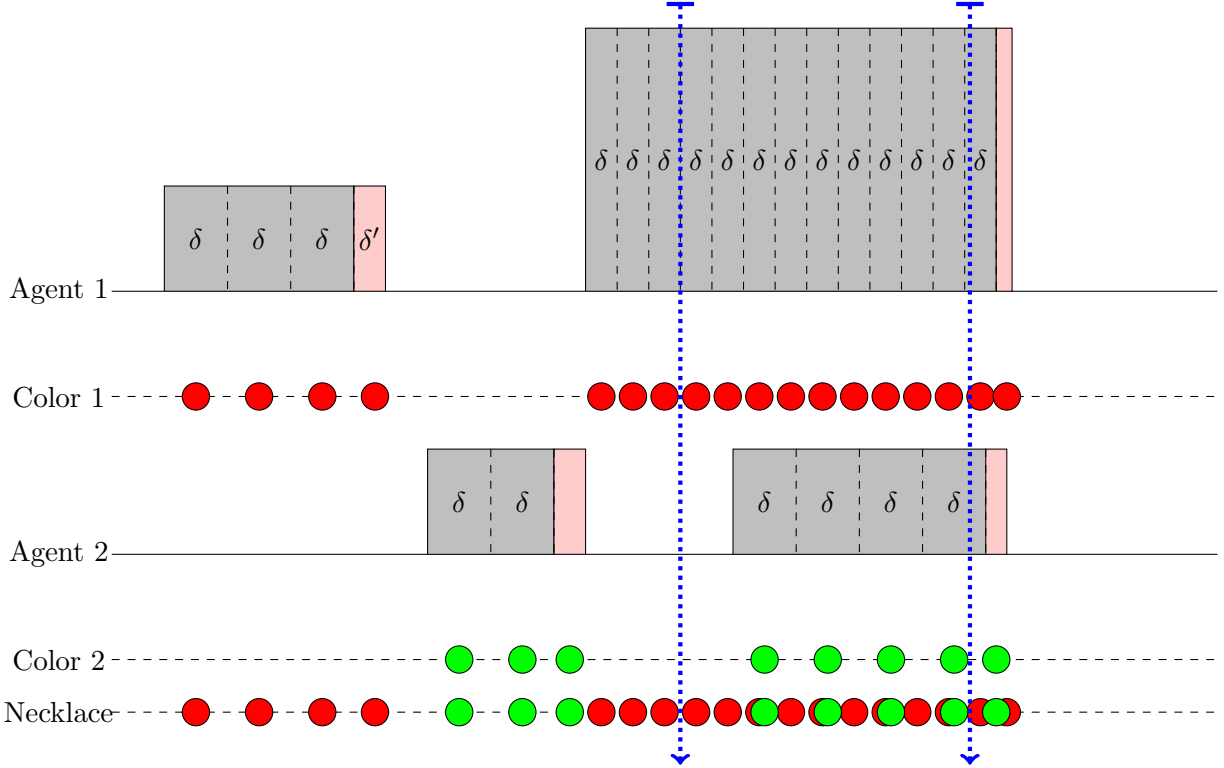
			
			 Let $\mathcal{S}$ be a solution to $\mathcal{B}$ (which uses $(k-1)\ell$ cuts); we will prove that having the same cuts in the same positions precisely gives a solution to $(n,\ell,\epsilon)$-\CKD. Consider any agent $i$ and label the beads of colour $c_i$ with $j=1,\ldots,t$ for some $t \in \mathbb{N}$ being the total number of beads of colour $c_i$ according to the construction above. Let $[d_j,d_{j+1}]$ be the interval defined by two consecutive such beads. 
			 
			 Note that by the construction above, if (i) all sub-blocks of agent $i$ have volume exactly $\delta$ (i.e. there are no imperfect sub-blocks), (ii) there are no parity beads and (iii) each cut in $\mathcal{S}$ either doesn't intersect any valuation block or is placed on the midpoint $(d_j+d_{j+1})/2$ of some interval $[d_j, d_{j+1}]$ (i.e. at the boundary of one or two valuation sub-blocks, e.g. see the first cut in Figure \ref{fig:necklace}), then $\mathcal{S}$ is an exact solution to the consensus $1/k$-division problem. However, in addition to the potential existence of imperfect sub-blocks and the parity beads, the cuts in $\mathcal{S}$ might actually be placed on different points in $[d_j,d_{j+1}]$, because of the presence of beads of different colours which might be placed inside the intervals (e.g. see the second cut in the interval between the last two green (dark) beads in Figure \ref{fig:necklace} for an example).
			 
			  Note however that such a cut will still lie inside $[d_j,d_{j+1}]$, as otherwise the partition of beads would be imbalanced; therefore the imbalance in volume for such a cut is at most $\delta$. Since there are at most $(k-1)\ell$ cuts in total, the overall imbalance in volume because of the position of the cuts is at most $(k-1)\cdot \ell \cdot \delta$. Additionally, the imbalance in volume from each imperfect sub-block is at most $\delta$, and the overall imbalance in volume because of imperfect sub-blocks is at most $p_\mathcal{M}(n) \cdot \delta$. Finally, the imbalance in volume due to the parity beads is at most $(k-1) \cdot \delta$, since the parity-preserving procedure can remove at most $k-1$ beads for each agent. In total, the overall imbalance is at most $(k-1)\cdot \ell\cdot \delta + p_\mathcal{M}(n)\cdot \delta + (k-1)\cdot \delta$ which is less than $\epsilon$, by the choice of $\delta$. \qed
			  \end{proof}

			  \subsection{From Necklace Splitting to Approximate Consensus Division}
			  
			  In this subsection, we prove that the Approximate Consensus Division solution is at least as hard as Necklace Splitting; together with the result of the previous subsection, this establishes the computational equivalence of the two problems.
			  
			  \begin{theorem}\label{thm:necktodivreduction}
			  	$(n,\ell,k)$-\necklace\ is polynomial-time reducible to $(n,\ell,\epsilon)$-\CKD , when the number of cuts $\ell$ is bounded by a polynomial in $n$ and $\epsilon$ is inverse-polynomial in $n$.
			  \end{theorem}
			  
			  \begin{proof}
			  	The idea that we will use for the reduction is very similar to the one presented by Alon \cite{Alon87} for proving that a solution to (discrete) Necklace Splitting can be obtained from a solution for the continuous version. The proof in \cite{Alon87} starts from an (exact) solution to the continuous problem and proves using induction that it can be transformed into a solution for the discrete version, but appropriately moving some of the cuts, if needed. Here, we explain how to obtain a solution to Necklace Splitting from an \emph{approximate} solution of the continuous division problem and that this process runs in time polynomial in the number of beads of the necklace. 
			
			The main idea is to design an instance of $(n,\ell,\epsilon)$-\CKD\ by representing beads of colour $i \in \{1,2,\ldots,n\}$ of the instance of $(n,\ell,k)$-\necklace\ by uniform valuation blocks of agent $i \in \{1,2,\ldots,n\}$ that have no overlap between agents. Then, there exists a solution to the Consensus Division problem that does not cut through the intervals and that solution is a valid partition of the necklace. Starting from an arbitrary solution (which might have cuts that intersect the valuation intervals), we will move these cuts (if any) to the endpoints of the intervals one by one, while maintaining the total volume of each portion $O_j$ for $j=1,\ldots,k$ unchanged. \\
			  	
			  	\noindent More concretely, given an instance $\mathcal{B}$ of $(n,\ell,k)$-\necklace\, we design an instance $\mathcal{C}$ of $(n,\ell,\epsilon)$-\CKD\ as follows:
			  	\begin{itemize}
			  		\item For every colour $c_i \in \{1,2,\ldots,n\}$ of $\mathcal{B}$, we associate an agent $i$.
			  		\item For every bead of colour $c_i$, we create a valuation block of width $\delta$ and height $1/\delta$ for some sufficiently small $\delta$, such that the bead lies in the midpoint of the interval corresponding to the valuation block. Note that without loss of generality, we can assume that in $\mathcal{B}$, the beads are sufficiently spread (this does not affect the solution) and therefore there is no overlap between any two valuation blocks and in fact, the distance between any two valuation blocks is at least $\beta$, for some sufficiently large $\beta$. 
			  	\end{itemize}
			  	Note that by taking $\beta$ to be larger than $2\epsilon$, we can ensure that in any solution $\mathcal{S}$ of $\mathcal{C}$, each cut intersects with at most one valuation block and therefore all agents have their own designated cuts. In other words, manipulating the positions of the cuts that intersect some valuation interval $[l,r]$ of one agent does not affect the quality of the solution for any other agent, as long as the cuts remain within $[l - \beta/2, r+\beta/2]$ (i.e. they do not move into other valuation blocks).
			  	
			  	Now consider a solution $\mathcal{S}$ to $\mathcal{C}$. If for all agents $i \in \{1,2,\ldots,n\}$, the cuts do not intersect any valuation blocks, then the solution can be translated verbatim to a solution to $(n,\ell,k)$-\necklace\, by keeping the cuts at the same positions. However, there might be several cuts that intersect the interior of the valuation intervals; we will refer to those as \emph{bad} cuts. Let $\mathcal{B}_i$ be the set of bad cuts for agent $i$. We will be able to move these cuts based on the following observation. Consider a bad cut at $e$ intersecting a valuation block of agent $i$, defined on an interval $[l,r]$ and let $j_1$ and $j_2$ be the labels of the pieces on the left side and on the right side of the cut respectively (we can assume that $j_1 \neq j_2$, otherwise we can simply remove the cut). Let $o^{j_1}=v_i([l,e])$ and $o^{j_1}=v_i([e,r])$ be the volumes of the sub-blocks in $[l,r]$ corresponding to each label.
			  	
			  	Assume without loss of generality that $o^{j_1}>o^{j_2}$ (the argument of the other case is symmetric). If  $o^{j_1}-o^{j_2}\geq \epsilon$, this implies that for the remaining valuation of agent $i$ (besides $o^{j_1}$ and $o^{j_2}$), there is an excess of volume labelled by $j_2$, otherwise $\mathcal{C}$ would not be an approximate solution to the Consensus Division problem. Then, we can move the cuts accordingly (by also possibly moving some other cuts in the process) such that the two excesses cancel out; we explain how to do that below. Note that since we have started from an approximate solution $\mathcal{C}$ of $(n,\ell,\epsilon)$-\CKD, after this procedure, bad cuts might still exist, but they will only account for small discrepancies and can be easily handled; we will refer to those cuts as the \emph{inaccuracy cuts}.
			  	
			  	Following \cite{Alon87}, we will consider a set of multigraphs (one for each agent) $G_i = (V_i,E_i)$, where $V_i=\{F_1,F_2,\ldots,F_k\}$, i.e. we have one vertex for each one of the $k$ possible labels. Each edge of the graph will correspond to a cut; in particular, there is an edge $(F_a,F_b)$ for each bad cut between two pieces with labels $a$ and $b$, and note that there might be multiple such edges. Note that by the discussion above, if $|v_i(O_a) - v_i(O_b)|\geq \epsilon$, then the degree of both $F_a$ and $F_b$ is at least $2$ and therefore the graph has at least one cycle. 
			  	
			  	For each agent $i$, we will use two subroutines, one to eliminate all cuts in $B_i$ except for possibly the inaccuracy cuts and the second one to eliminate the remaining bad cuts. For the first subroutine, we will work with the graph $G_i$ and we will eliminate cuts in $B_i$ in steps, by removing edges of the graph, i.e. the graph will be evolving. After each step, the following invariant will be maintained:
				\emph{The total volume of each partition remains unchanged and the number of bad cuts will be reduced by at least $1$.} The subroutine is stated below.\\
				
			  \begin{algorithmic}
			  	\While{$G_i$ has a cycle}
			  	\State Find a cycle $(F_{j_1},F_{j_2},\ldots,F_{j_m})$
			  	\State \textbf{Cycle resolving:} Move all the cuts corresponding to edges on the cycle by the same amount. For the cut corresponding to the first edge of the cycle, move it in the direction that increases\footnote{We can also move the cut in the other direction; that will correspond to the same solution with a permutation of the labels along the cycle.}  the volume of the label $F_{j_1}$ and therefore decreases the volume of the label $F_{j_2}$. We will call such a direction \emph{increasing for $F_{j_1}$} and \emph{decreasing for $F_{j_2}$}. For any other cut corresponding to an edge $(F_{j_{h-1}},F_{j_{h}})$ move it in the direction with the opposite effect of the previous movement with respect to label $F_{j_{h}}$, i.e. if the previous cut was moved in an increasing direction for $F_{j_{h}}$, the cut will move in a decreasing direction for $F_{j_{h}}$. Move the cuts until either:
			  	\begin{itemize}
			  		\item Some cut coincides with another cut. In that case, merge all the coinciding cuts and remove the labels of the pieces of volume $0$.
			  		\item Some cut coincides with the endpoint of an interval. 
			  	\end{itemize}
			  	\EndWhile
			  \end{algorithmic}
			  \bigskip
			  It is clear that at the end of each step of the procedure above, the number of bad cuts is decreased by at least one, either because the cut was merged with another bad cut, or because the cut was moved to the boundary of the interval. Additionally, since all the cuts have been moved by the same amount and for a each vertex in the cycle, one move was in an increasing direction and one was in a decreasing direction, the total volume of each portion $O_j$, for $j=1,\ldots,k$, remains unchanged. Finally, at the end of the routine, all cycles have been resolved and the graph $G_i$ is acyclic.
			  
			  Note that since the number of steps is at most $\ell$ and the number of cuts moved in each round is at most $\ell$, the subroutine runs in polynomial time, since the number of cuts $\ell$ is bounded by some polynomial $p_\ell(n)$. If $\mathcal{C}$ was an exact solution, the reduction would be completed here;  however, since we start from an approximate solution to the Consensus Division problem, we need a second subroutine to deal with the inaccuracy cuts. The subroutine will essentially transform the approximate solution of $\mathcal{C}$ into an exact solution, which in turn is a solution to $\mathcal{B}$.
			  	
			  This subroutine will be simple, just move each cut to the closest endpoint of the interval $[l,r]$ whose interior it intersects. Note that the imbalance in volume between any two labels $j_1$ and $j_2$ is due to a single bad cut, otherwise the graph $G_i$ would have a cycle. Since each valuation block is constructed to have total volume $1$, the cut must lie in $[l,l+\gamma] \cup [r-\gamma,r]$, where $\gamma < \epsilon \cdot \delta$ and therefore it can unambiguously be moved to the closest endpoint of $[l,r]$.  Additionally, this sequence of moves produces an exact solution to $\mathcal{C}$, as otherwise, the original solution would have a discrepancy larger than $\epsilon$ with respect to at least two partitions $O_{j_1},O_{j_2}$. Since there are only polynomially many cuts, the second subroutine also runs in polynomial time.
			  This completes the proof.

			  \end{proof}

			 \subsection{Hardness results for (approximate) Necklace Splitting}
			 
			First, we remark that we can define an approximate version of Necklace Splitting, where the goal is to partition the necklace into pieces that contain \emph{approximately} the same number of beads of colour $c_i$, for each $i = 1,2,\ldots,n$. Formally:
			 
			 \begin{definition}[$(n,\ell,k,\varepsilon)$-\necklace]
			 	\begin{itemize}
			 	\item[-]
			 	\item[-] \textbf{Input:} $k \cdot m$ beads placed on an interval $O$ with $\alpha_i \cdot k$ beads of colour $i=1,\ldots,n$ where $\alpha_i \in \mathbb{N}^+$, with $k \leq n$ and a parameter $0\leq \varepsilon \leq \alpha_i$.
			 	
			 	\item[-] \textbf{Output:} A partition of $O$ into $k$ parts $O_1,O_2,\ldots, O_k$ such that for each colour $i=1,\ldots,n$, for each $j \in \{1,\ldots,k\}$, it holds that $b_i(O_j) \in [\underline{\alpha_i},\overline{\alpha_i}]$, where $\underline{\alpha_i}=\lceil\alpha_i-\varepsilon\rceil$ and $\overline{\alpha_i}=\lfloor \alpha_i+\varepsilon \rfloor$, using $(k-1)\cdot \ell$ cuts.
			 	\end{itemize}
			 \end{definition}	
			 
			 \noindent Note that according to the definition above, if $\varepsilon$ is chosen to be small and the number of beads of some colour is small, the partition for that colour is required to be exact, but if there are enough beads of a colour, then the partition is allowed to be ``off'' by a few beads. The approximate Consensus-$1/k$-Division problem actually also reduces to the approximate Necklace Splitting problem, when the error factor $\varepsilon$ in the latter is inversely exponential in the number of beads of each colour; it is not hard to see that a solution to $(n,\ell,k,\varepsilon)$-\necklace\ will only add an inversely polynomial error factor to the solution to $(n,\ell,\epsilon')$-\CKD\ and the theorem will still hold. We state the result as a corollary of Theorem \ref{thm:divtoneckreduction}.
			 
			 \begin{corollary}
			 	\label{thm:divtoapproxneckreduction}
			 	$(n,\ell,\epsilon')$-\CKD\ is polynomial time reducible to $(n,\ell,k,\varepsilon)$-\necklace, when the number of cuts $\ell$ is bounded by a polynomial in $n$ and $\epsilon,\varepsilon$ are inverse-polynomial in $n$.
			 \end{corollary}
			 In \cite{FFGZ}, it was proven that the approximate Consensus Halving problem is $\ppa$-hard when $n+t$ cuts are used, for some constant $t$ and NP-hard when only $n-1$ cuts are used. From these results and from the results of this section \ref{thm:divtoapproxneckreduction} we obtain the following corollary.
			 
			 \begin{corollary}
			 	$(n,\ell,2,\varepsilon)$-\necklace\ is
			 	\begin{itemize}
			 		\item \emph{PPAD-hard}, when $\ell = n+t$ for any constant $t \geq 0$ and $0 \leq \epsilon \leq 1/p_s(n)$, where $p_s$ is a polynomial. 
			 		\item \emph{NP-hard}, when $\ell=n-1$ and $0 \leq \epsilon \leq 1/ p_s(n)$, where $p_s$ is a polynomial. 
			 	\end{itemize}
			 \end{corollary}